\documentclass[11pt,a4paper]{article}

\usepackage[margin=1in]{geometry}
\usepackage{amsmath,amssymb,amsthm}
\usepackage{mathtools}
\usepackage{bm}
\usepackage{booktabs}
\usepackage{array}
\usepackage{enumitem}
\usepackage{tabularx}
\usepackage{longtable}
\usepackage{makecell}
\usepackage{caption}
\usepackage{float}
\usepackage{placeins}
\usepackage[round]{natbib}
\usepackage{microtype}
\usepackage[colorlinks=true,linkcolor=blue,citecolor=blue,urlcolor=blue]{hyperref}
\hypersetup{
  pdftitle={The Mathematics of Heuristic Portfolio Optimization (HPO)},
  pdfauthor={Miquel Noguer i Alonso}
}
\numberwithin{equation}{section}
\setlist[itemize]{leftmargin=*, topsep=2pt}
\setlist[enumerate]{leftmargin=*, topsep=2pt}
\newtheorem{theorem}{Theorem}[section]
\newtheorem{proposition}[theorem]{Proposition}
\newtheorem{lemma}[theorem]{Lemma}
\newtheorem{corollary}[theorem]{Corollary}
\newtheorem{definition}[theorem]{Definition}
\newtheorem{assumption}[theorem]{Assumption}
\newtheorem{remark}[theorem]{Remark}

\newcommand{\E}{\mathbb{E}}
\newcommand{\R}{\mathbb{R}}
\newcommand{\Var}{\mathrm{Var}}
\newcommand{\Cov}{\mathrm{Cov}}
\newcommand{\Corr}{\mathrm{Corr}}
\newcommand{\CVaR}{\mathrm{CVaR}}
\newcommand{\diag}{\mathrm{diag}}
\newcommand{\w}{\mathbf{w}}
\newcommand{\ones}{\mathbf{1}}
\newcommand{\bmu}{\bm{\mu}}
\newcommand{\bsig}{\bm{\sigma}}
\newcommand{\bSigma}{\bm{\Sigma}}
\newcommand{\bbeta}{\bm{\beta}}
\DeclareMathOperator*{\argmax}{arg\,max}
\DeclareMathOperator*{\argmin}{arg\,min}

\title{\Large\textbf{The Mathematics of Heuristic Portfolio Optimization (HPO)}}

\author{Miquel Noguer i Alonso\\
\small Artificial Intelligence Finance Institute}

\date{\today}

\begin{document}

\maketitle

\begin{abstract}
\noindent
Practitioners often allocate capital with forecast-light rules such as equal weight, inverse volatility, risk parity, HRP, and return-adjusted HRP (RA-HRP). This paper develops \emph{Heuristic Portfolio Optimization} (HPO): portfolio construction as an information-restricted projection of the Markowitz/tangency solution onto a stable rule class. The implied-return principle, $\w$ is maximum-Sharpe iff $\bmu_e\propto\bSigma\w$, yields closed-form optimality sets for the main heuristics and exposes the Schur-complement substitutions behind HRP. For RA-HRP we add fixed-tree cluster-Sharpe recursion, unit-free HRP--RA-HRP interpolation, nonlinear tangency-coincidence conditions, Schur-RA-HRP conditional-risk splits, and pathwise/KL decompositions of HRP--RA-HRP weight distortion; a first-order Sharpe calculus along the unit-free homotopy expresses the marginal value of return information as a sum of nodewise alphas against HRP and certifies a KL trust budget linear in the interpolation weight. We then formalize generic HPO maps, define the implied-return defect, prove that it equals squared Sharpe inefficiency, characterize tree-HPO coincidence by nodewise mass ratios, and give an exact bias--variance decomposition for estimated heuristic rules. The final contribution embeds HPO into Reinforcement Learning Portfolio Optimization (RLPO): every HPO map induces a deterministic stationary policy; static HPO is the $\gamma=0$ no-friction face of the Bellman problem; RA-HRP provides a hierarchical policy prior through node-level split probabilities; and dynamic policy improvement is justified exactly when continuation value exceeds myopic HPO defect and trading frictions; a performance-difference identity prices the total value gap of acting myopically, yielding an $\varepsilon/(1-\gamma)$ myopia bound, and the nodewise alphas reappear as the myopic policy-gradient coordinates of the hierarchical actor at its prior. The resulting framework treats HPO as the static optimality layer and RLPO as the dynamic control layer. The conditions remain GRS-testable, transfer to mean--CVaR and expected utility under ellipticity, and become Kelly-growth conditions in diffusion limits.
\\[0.5em]
\end{abstract}

\newpage
\tableofcontents
\newpage

\section{Introduction}

Seventy years after \citet{markowitz1952}, a striking share of institutional capital is allocated by rules that refuse to optimize. The equal-weight portfolio ignores both means and covariances; inverse-volatility weighting uses only marginal volatilities; equal-risk-contribution (ERC, ``risk parity'') portfolios use the full covariance matrix but no means \citep{maillard2010,roncalli2013,qian2006}; maximum diversification uses volatilities and covariances in a specific ratio \citep{choueifaty2008}; hierarchical risk parity (HRP) uses the covariance matrix only through a clustering tree and a sequence of within-cluster variances \citep{lopezdeprado2016}; and return-adjusted HRP (RA-HRP) keeps the tree but lets local cluster-Sharpe scores tilt the recursive budget \citep{noguer2026rahrp}. The key refinement is not merely that RA-HRP uses expected returns, but that return information is consumed locally, after clustering, through dimensionless split fractions that can be interpolated continuously from HRP without mixing inverse variances and Sharpe ratios. The standard defense of these rules is empirical: optimized portfolios inherit estimation error in $\bmu$ at first order \citep{merton1980,michaud1989,bestgrauer1991,chopraziemba1993}, and naive rules are hard to beat out of sample \citep{demiguel2009}.

We use the term \emph{Heuristic Portfolio Optimization} (HPO) for this discipline: the mathematical study of portfolio maps that deliberately consume less information than the full Markowitz optimizer, together with the exact optimality, efficiency, estimation, and belief conditions under which that deletion of information is harmless or beneficial. In this sense HPO is not the opposite of optimization. It is optimization after restricting the admissible decision maps to those that are stable, interpretable, low-inversion, low-turnover, or compatible with a reduced information set.

This paper develops the theory behind that practice at five levels of depth.

\paragraph{Level 1: exact optimality conditions.} The implied-return principle (Theorem~\ref{thm:foc}) states that a fully invested portfolio with positive expected excess return is maximum-Sharpe if and only if expected excess returns are proportional to $\bSigma\w$. Consequently \emph{every} heuristic is optimal for \emph{some} parameter configuration, and the substantive object is the exact shape of its \emph{optimality set}
\[
\mathcal{O}(\w^\circ) = \{(\bmu_e,\bSigma) : \bmu_e = a\bSigma\w^\circ,\; a>0\},
\]
because adopting the rule is behaviorally equivalent to asserting membership. We compute this set in closed form for every rule in common use (Sections~\ref{sec:flat}--\ref{sec:hrp}): zero alpha against itself for equal weight (Theorem~\ref{thm:ew}); Sharpe ratios proportional to correlation row sums for inverse volatility (Theorem~\ref{thm:iv}); equal premia for minimum variance (Theorem~\ref{thm:gmv}); equal Sharpe ratios, with no restriction on correlations, for maximum diversification (Theorem~\ref{thm:md}); and, for risk parity, the new characterization that the equal \emph{risk} budget must double as an equal \emph{performance} budget, $w_i\mu_{e,i}$ constant (Theorem~\ref{thm:ercchar}), of which the equal-Sharpe/constant-correlation condition of \citet{maillard2010} is a corollary and demonstrably far from necessary. For HRP we first derive the exact Schur-complement bisection identity satisfied by the true minimum-variance portfolio (Theorem~\ref{thm:schur}), which isolates the three substitutions the algorithm makes --- raw block for Schur complement, flat budget for tilted budget, inverse variance for within-cluster minimum variance --- and then prove the two regimes in which all three are simultaneously costless: diagonal covariance under \emph{any} dendrogram (Theorem~\ref{thm:hrpdiag}) and uncorrelated exchangeable blocks under balanced tree-aligned bisection (Proposition~\ref{prop:hrpblocks}), together with the minimal two-asset counterexample marking the boundary (Remark~\ref{rem:hrpcounter}). For RA-HRP we show that adding returns changes the optimality set from a ray to a nonlinear fixed-point set: the exact condition is $\bmu_e = a\bSigma\Phi^{\mathrm{RA}}_{\mathcal T,\varepsilon}(\bmu_e,\bSigma)$, equivalently node-by-node equality between the RA-HRP floored cluster-Sharpe split and the exact Schur-tangency split (Theorem~\ref{thm:rahrp-coincidence}). We also prove a two-asset boundary result, a balanced equal-volatility diagonal exactness result, and a pathwise/KL decomposition of the distortion from HRP to RA-HRP (Corollary~\ref{cor:rahrp-twoasset}, Proposition~\ref{prop:rahrp-equalvol}, Theorem~\ref{thm:rahrp-pathwise}). A first-order Sharpe calculus along the homotopy (Section~\ref{sec:homotopy-calculus}) then computes the derivative of efficiency in the interpolation weight as a sum of nodewise alphas against HRP itself (Theorem~\ref{thm:homotopy-alpha}), making the marginal value of return information a GRS-estimable quantity, and bounds the induced reallocation by a KL trust budget linear in $\lambda$ (Proposition~\ref{prop:homotopy-kl}). The interpolation from HRP to RA-HRP is made at the level of local branch probabilities $r_v,s_v\in(0,1)$, not by adding inverse-variance scores to Sharpe scores; this unit-free homotopy is the mathematically clean way to control how much local return information enters the hierarchy. We then add Schur-RA-HRP, which replaces standalone cluster risk by conditional Schur risk and therefore targets exactly the cross-cluster covariance gap diagnosed by the HRP Schur identity.

\paragraph{Level 2: the geometry and cost of suboptimality.} Conditions that hold on a measure-zero set would be a curiosity unless violations were cheap. Both halves are theorems. The Sharpe ratio of \emph{any} portfolio obeys the exact angular identity $S(\w) = S_{\max}\cos\theta_{\bSigma}(\w, \w_T)$ in the $\bSigma$-inner product (Theorem~\ref{thm:cosine}): inefficiency is purely angular. Each optimality set is a ray --- Lebesgue-null in $\R^N$ --- so under any diffuse belief, exact optimality of a fixed heuristic is a probability-zero event (Proposition~\ref{prop:generic}). Yet the efficiency map is \emph{stationary} on the set: perturbing $\bmu_e$ off the coincidence ray by $\varepsilon v$ degrades squared Sharpe efficiency by an explicitly computable $O(\varepsilon^2)$ term whose coefficient is a generalized residual variance (Theorem~\ref{thm:secondorder}). Never optimal, rarely costly --- with the rate and the constant. Each condition is moreover a zero-intercept restriction in time-series regressions on the heuristic itself, hence directly testable by the $F$-test of \citet[\textbf{C}]{gibbons1989} (Remark~\ref{rem:grs}). Under one-factor structure, every condition becomes a statement about observable betas and idiosyncratic variances (Section~\ref{sec:onefactor}), including a self-contained derivation of the minimum-variance threshold-beta representation of \citet{clarke2011} and the result that maximum diversification is optimal in a priced one-factor world precisely when all assets are equally correlated with the factor.

\paragraph{Level 3: in what sense the heuristics are right.} Two transfer theorems show the conditions are not artifacts of the quadratic objective: under elliptical returns they govern mean--CVaR optimality at every confidence level and expected-utility optimality for every risk-averse investor (Theorem~\ref{thm:elliptical}), and in the continuous-time limit they are simultaneously the conditions for (fractional) Kelly growth optimality (Proposition~\ref{prop:kelly}). The deepest reading then inverts the question: each heuristic is the \emph{exactly optimal} policy of an investor whose information set is smaller than $(\bmu_e,\bSigma)$, with ignorance resolved by symmetry. We prove Bayes-optimality theorems at three rungs (Section~\ref{sec:bayes}): full exchangeability of beliefs forces the Bayes portfolio to be $1/N$ (Theorem~\ref{thm:bayesew}); exchangeability of beliefs about \emph{standardized} returns, volatilities known, forces inverse volatility (Theorem~\ref{thm:bayesiv}); and conjugate mean uncertainty with an uninformative premium prior forces minimum variance (Proposition~\ref{prop:bayesgmv}). The $1/N$ rule is also the exact limit of distributionally robust optimization as the ambiguity radius diverges \citep[\textbf{C}]{pflug2012}, and constraint-induced shrinkage \citep[\textbf{C}]{jagannathanma2003} plays the analogous role for minimum variance. Comparable foundations for ERC and HRP are, to our knowledge, open; we say so explicitly rather than assert them.

\paragraph{Level 4: implementation and statistical diagnostics.} The population theory is extended to mandate constraints and trading frictions. For any closed convex feasible set, exact optimality becomes a normal-cone implied-return condition: residual alphas must be shadow prices of active constraints (Theorem~\ref{thm:normalcone}); with quadratic transaction costs, implied returns acquire an explicit turnover wedge (Proposition~\ref{prop:costimplied}). We then derive the first differential of the tangency map (Theorem~\ref{thm:tangencydiff}), showing exactly why mean errors dominate covariance errors, convert every optimality condition into a projection diagnostic against its implied-return ray (Proposition~\ref{prop:projection}), state the usable GRS formula, and add the high-dimensional random-matrix inversion tax for plug-in GMV (Proposition~\ref{prop:rmtgmv}). Finally, we isolate the exact place where the elliptical transfer fails: nonconstant standardized tail shape adds a first-order downside-risk wedge.

\paragraph{Level 5: HPO as the static policy face of RLPO.} The final layer connects the static theory to Reinforcement Learning Portfolio Optimization (RLPO). An HPO rule $\Phi$ is a statewise portfolio map; in a Markov decision problem it therefore induces the deterministic stationary policy $\pi_\Phi(\cdot\mid s)=\delta_{\Phi(s)}$. When the discount factor is zero, trading frictions are absent, and the one-period reward is the same static objective used by HPO, the Bellman problem collapses exactly to the HPO problem. For $\gamma>0$, state-dependent costs, regime transitions, and feedback effects, RLPO is the dynamic policy-improvement operator over the HPO prior. The paper formalizes this bridge in Section~\ref{sec:rlpo}: the HPO implied-return defect becomes a reward-shaping penalty, RA-HRP becomes a hierarchical actor prior through node-level branch probabilities, and dynamic deviations from HPO are justified precisely when continuation value exceeds myopic HPO inefficiency and turnover costs. A performance-difference identity (Proposition~\ref{prop:hpo-value-gap}) prices the total value gap of the myopic policy as an occupancy-weighted HPO advantage with an $\varepsilon/(1-\gamma)$ myopia bound (Corollary~\ref{cor:eps-myopic}), and the nodewise alphas of the homotopy calculus reappear as the myopic policy-gradient coordinates of the hierarchical actor at its prior (Corollary~\ref{cor:actor-gradient-alpha}).

\paragraph{Claim tiering.} Following the discipline used throughout our research program, every result is labeled. \textbf{P} marks results proved in this paper with self-contained proofs; \textbf{C} marks results proved in the cited literature and stated without proof; \textbf{E} marks empirical regularities reported in cited work. We make no conjectural claims and report no numerics of our own.

\section{Scope, Standing Assumptions, and Paper Map}\label{sec:scope}

This section fixes the boundary of the results. The paper does not claim that any single heuristic dominates out of sample, nor that a neural policy can replace portfolio theory. It claims something narrower and stronger: for each information-restricted allocation map, one can write the exact population condition under which the rule is optimal, the exact geometric loss when it is not, the statistical diagnostics that test the condition, and the dynamic Bellman inequality that decides whether an RL policy should deviate from it.

\subsection{Standing assumptions}

Unless explicitly stated otherwise, the paper uses the following standing assumptions.
\begin{enumerate}[label=\textbf{A\arabic*.}]
\item \textbf{Finite asset universe.} There are $N<\infty$ risky assets and a riskless asset with rate $r_f$. Expected excess returns are $\bmu_e=\bmu-r_f\ones$ and the covariance matrix satisfies $\bSigma\succ 0$.
\item \textbf{Population versus implementation.} Population statements use the true pair $(\bmu_e,\bSigma)$. Estimated rules replace this pair by statewise estimates $(\widehat{\bmu}_{e,t},\widehat{\bSigma}_t)$ and are analyzed separately through sampling, turnover, and diagnostic terms.
\item \textbf{Feasibility.} The baseline tangency and mean--variance results are stated on the fully invested affine hyperplane. Long-only constraints, caps, benchmark bands, and transaction-cost wedges are introduced explicitly in Section~\ref{sec:constrained} through normal cones and cost-adjusted implied returns.
\item \textbf{Rule versus family.} A fixed heuristic is a map $\Phi$ from a reduced information statistic to the simplex. A tunable heuristic family $\{\Phi_\lambda\}$ is selected by projecting the full optimizer onto a lower-complexity image set, not by pretending that the heuristic image set equals the full feasible set.
\item \textbf{Tree stability.} HRP, RA-HRP, and Schur-RA-HRP are first analyzed for a fixed tree. Rolling-window implementations with re-estimated dendrograms add a discrete topology layer; the relevant quantities are merge preservation, cophenetic perturbation, and pathwise split distortion.
\item \textbf{Claim tiering.} A label \textbf{P} in a result header denotes a result proved here, \textbf{C} denotes a cited theorem or classical result used without reproving it, and \textbf{E} denotes an empirical regularity reported in the literature. Open problems are explicitly marked as open.
\end{enumerate}

\subsection{Contribution map}

\begin{table}[htbp]
\centering
\small
\renewcommand{\arraystretch}{1.25}
\begin{tabularx}{\textwidth}{p{3.0cm} X X}
\toprule
\textbf{Layer} & \textbf{Question answered} & \textbf{Mathematical object produced} \\
\midrule
Static optimality & When is a named heuristic exactly Markowitz/tangency optimal? & Implied-return sets $\{\bmu_e=a\bSigma\w\}$ and rule-specific coincidence restrictions. \\
Approximate efficiency & If exact optimality is generically false, how costly is misspecification? & Angular Sharpe identity and second-order efficiency expansion around each coincidence ray. \\
Hierarchical HPO & What information does HRP delete, and what does RA-HRP restore? & Schur-bisection identity, score-tree kernels, pathwise/KL split attribution, and Schur-RA-HRP conditional-risk scores. \\
Statistical implementation & Why can a biased heuristic beat a plug-in optimizer? & Exact mean--variance bias--variance identity, tangency-map differential, GRS/Wald diagnostics, and high-dimensional inversion tax. \\
Dynamic policy improvement & How should a learned policy use a heuristic prior? & HPO-induced stationary policies, deviation inequality, performance-difference identity, and nodewise hierarchical actor gradients. \\
\bottomrule
\end{tabularx}
\caption{Contribution map: each layer converts a practitioner question into a theorem-level object.}
\label{tab:contribution-map}
\end{table}

\subsection{What the theory deliberately does not assert}

The framework is falsifiable. A heuristic can fail its zero-alpha or implied-return condition, exhibit a large angular loss, generate excessive turnover after realistic costs, or break down under non-elliptical tail-shape heterogeneity. A learned RLPO policy can also fail if its continuation-value gain is not large enough to pay for added HPO defect and trading frictions. These are not exceptions to the theory; they are exactly the diagnostics the theory prescribes.

\section{Setup, the Implied-Return Principle, and the Geometry of Optimality}\label{sec:implied}

\subsection{Notation}

There are $N$ risky assets with return vector $\mathbf{r}$, mean $\bmu = \E[\mathbf{r}]$, and covariance matrix $\bSigma \succ 0$. A riskless rate $r_f$ is available; $\bmu_e = \bmu - r_f\ones$ denotes expected excess returns, assumed nonzero. Volatilities are $\sigma_i = \sqrt{\Sigma_{ii}}$, $D = \diag(\sigma_1,\dots,\sigma_N)$, and the correlation matrix is $R = D^{-1}\bSigma D^{-1}$, so $\bSigma = DRD$. Sharpe ratios are $S_i = \mu_{e,i}/\sigma_i$. A portfolio is $\w \in \R^N$, fully invested if $\ones^\top\w = 1$. Portfolio variance is $\sigma_p^2(\w) = \w^\top\bSigma\w$, the Sharpe ratio is $S(\w) = \w^\top\bmu_e/\sigma_p(\w)$, and $\langle \mathbf{u},\mathbf{v}\rangle_{\bSigma} = \mathbf{u}^\top\bSigma\mathbf{v}$, $\|\mathbf{u}\|_{\bSigma} = \langle\mathbf{u},\mathbf{u}\rangle_{\bSigma}^{1/2}$ denote the $\bSigma$-inner product and norm.

\begin{definition}[Tangency coincidence]\label{def:tangency}
Assume $\ones^\top\bSigma^{-1}\bmu_e > 0$. The \emph{tangency} (maximum-Sharpe) portfolio is
\begin{equation}
\w_T \;=\; \frac{\bSigma^{-1}\bmu_e}{\ones^\top\bSigma^{-1}\bmu_e},
\qquad
S_{\max} := S(\w_T) = \sqrt{\bmu_e^\top\bSigma^{-1}\bmu_e}.
\end{equation}
A fully invested $\w^\circ$ \emph{coincides with the tangency portfolio} if $\w^\circ = \w_T$. This is the canonical notion of optimality in this paper: with a riskless asset, every mean--variance investor holds cash plus $\w_T$, so a heuristic equal to $\w_T$ is optimal for all risk aversions simultaneously.
\end{definition}

\begin{definition}[Frontier membership]\label{def:frontier}
A fully invested $\w^\circ$ is \emph{frontier-efficient} if it solves $\max_\w \{\w^\top\bmu - \tfrac{\lambda}{2}\w^\top\bSigma\w : \ones^\top\w = 1\}$ for some $\lambda > 0$.
\end{definition}

\subsection{The first-order characterization}

\begin{theorem}[Implied-return principle; \textbf{P}]\label{thm:foc}
Let $\w^\circ$ be fully invested.
\begin{enumerate}
\item[(i)] $\w^\circ = \w_T$ if and only if there exists $a > 0$ such that
\begin{equation}\label{eq:foc-tangency}
\bmu_e \;=\; a\,\bSigma\w^\circ .
\end{equation}
\item[(ii)] $\w^\circ$ is frontier-efficient if and only if there exist $a > 0$ and $b \in \R$ such that
\begin{equation}\label{eq:foc-frontier}
\bmu \;=\; a\,\bSigma\w^\circ + b\,\ones .
\end{equation}
\end{enumerate}
\end{theorem}

\begin{proof}
(i) If $\bmu_e = a\bSigma\w^\circ$ with $a>0$, then $\bSigma^{-1}\bmu_e = a\w^\circ$, and normalizing by $\ones^\top\bSigma^{-1}\bmu_e = a > 0$ (using $\ones^\top\w^\circ = 1$) gives $\w_T = \w^\circ$. Conversely, if $\w^\circ = \w_T$ then $\bmu_e = (\ones^\top\bSigma^{-1}\bmu_e)\,\bSigma\w^\circ$ with positive scalar.

(ii) The program in Definition~\ref{def:frontier} is strictly concave with one affine constraint; the Karush--Kuhn--Tucker conditions $\bmu - \lambda\bSigma\w - \gamma\ones = 0$, $\ones^\top\w = 1$ are necessary and sufficient. Identify $a = \lambda$, $b = \gamma$.
\end{proof}

\begin{corollary}[Reverse optimization; \textbf{P}, classical]\label{cor:reverse}
Every fully invested $\w^\circ$ with $\w^{\circ\top}\bSigma\w^\circ > 0$ is the tangency portfolio of the \emph{implied} excess returns $\bmu_e^{\mathrm{imp}} = \bSigma\w^\circ$ (up to positive scaling), and frontier-efficient for $\bSigma\w^\circ + b\ones$, any $b$.
\end{corollary}

This is the reverse-optimization step of \citet{blacklitterman1992} and the sensitivity logic of \citet{bestgrauer1991}. Its consequence is worth stating bluntly: \emph{no heuristic can be unconditionally suboptimal}. The substantive question is the shape of the optimality set, because adopting the rule $\w^\circ$ is behaviorally equivalent to asserting $(\bmu_e,\bSigma) \in \mathcal{O}(\w^\circ)$.

\begin{remark}[Beta form]\label{rem:beta}
Condition \eqref{eq:foc-tangency} says expected excess returns are proportional to covariances with the candidate: $\mu_{e,i} = a\,\Cov(r_i, \w^{\circ\top}\mathbf{r})$. Dividing by the candidate's own excess return gives the exact one-factor pricing relation
\begin{equation}\label{eq:zeroalpha}
\mu_{e,i} \;=\; \beta_i(\w^\circ)\,\bigl(\w^{\circ\top}\bmu_e\bigr), \qquad
\beta_i(\w^\circ) = \frac{\Cov(r_i, \w^{\circ\top}\mathbf{r})}{\Var(\w^{\circ\top}\mathbf{r})}.
\end{equation}
A heuristic is maximum-Sharpe iff no asset has alpha against it --- it prices the cross-section as if it were the market portfolio of a one-factor CAPM \citep{sharpe1964}. All conditions below specialize \eqref{eq:zeroalpha}.
\end{remark}

\begin{remark}[Testability; \textbf{C}]\label{rem:grs}
Relation \eqref{eq:zeroalpha} is the null hypothesis $\bm{\alpha} = 0$ in the system of time-series regressions $r_{e,i,t} = \alpha_i + \beta_i\,r^{\circ}_{e,t} + \epsilon_{i,t}$ where $r^\circ_{e,t}$ is the heuristic's realized excess return. Hence every optimality condition in this paper is directly testable by the finite-sample $F$-statistic of \citet{gibbons1989} for the mean--variance efficiency of a given portfolio: ``is equal weight optimal in this universe'' is a GRS test with the equal-weight portfolio as the candidate factor, and identically for IV, ERC, MD, GMV, HRP, and any fixed realization of RA-HRP. If the RA-HRP weights are estimated from the same return panel used in the test, the clean interpretation requires a split-sample or out-of-sample implementation; conditionally on the realized candidate weights, the GRS null is still the same zero-alpha restriction. The asymptotic and finite-sample distribution theory is otherwise standard.
\end{remark}

\subsection{Efficiency is an angle}\label{sec:angle}

Coincidence conditions classify the boundary; the angular identity measures everything else.

\begin{theorem}[Angular efficiency identity; \textbf{P}]\label{thm:cosine}
For any portfolio $\w \neq 0$,
\begin{equation}\label{eq:cosine}
S(\w) \;=\; S_{\max}\,\cos\theta_{\bSigma}(\w, \w_T),
\qquad
\cos\theta_{\bSigma}(\w,\w_T) := \frac{\langle \w, \w_T\rangle_{\bSigma}}{\|\w\|_{\bSigma}\,\|\w_T\|_{\bSigma}} .
\end{equation}
In particular the Sharpe efficiency $\eta(\w) := S(\w)/S_{\max} = \cos\theta_{\bSigma}$, and $1 - \eta(\w)^2 = \sin^2\theta_{\bSigma}$ exactly. A portfolio attains the maximum Sharpe ratio iff it is $\bSigma$-parallel to $\bSigma^{-1}\bmu_e$, recovering Theorem~\ref{thm:foc}(i).
\end{theorem}

\begin{proof}
Let $\tilde\w = \bSigma^{-1}\bmu_e$ (unnormalized tangency direction). Then $\w^\top\bmu_e = \w^\top\bSigma\tilde\w = \langle\w,\tilde\w\rangle_{\bSigma}$ and $\|\tilde\w\|_{\bSigma}^2 = \bmu_e^\top\bSigma^{-1}\bmu_e = S_{\max}^2$. Hence
\[
S(\w) = \frac{\langle\w,\tilde\w\rangle_{\bSigma}}{\|\w\|_{\bSigma}}
= \|\tilde\w\|_{\bSigma}\,\frac{\langle\w,\tilde\w\rangle_{\bSigma}}{\|\w\|_{\bSigma}\|\tilde\w\|_{\bSigma}}
= S_{\max}\cos\theta_{\bSigma}(\w,\tilde\w),
\]
and $\theta_{\bSigma}(\w,\tilde\w) = \theta_{\bSigma}(\w,\w_T)$ since $\w_T \propto \tilde\w$ with positive constant.
\end{proof}

The identity converts every question about heuristic performance into a question about an angle in the metric defined by risk itself. Two structural consequences follow: optimality sets are thin, and the efficiency surface is flat across them.

\begin{proposition}[Genericity: exact optimality is null; \textbf{P}]\label{prop:generic}
Fix $\bSigma \succ 0$ and a fully invested $\w^\circ$. The tangency-coincidence set $\mathcal{O}_{\bSigma}(\w^\circ) = \{\bmu_e : \bmu_e = a\bSigma\w^\circ,\ a>0\}$ is an open ray --- a one-dimensional, Lebesgue-null subset of $\R^N$ for $N \ge 2$. The frontier set $\{\bmu : \bmu = a\bSigma\w^\circ + b\ones,\ a>0\}$ is an open half-plane, null for $N \ge 3$. Consequently, under any belief about $\bmu_e$ (or $\bmu$) absolutely continuous with respect to Lebesgue measure, the probability that a fixed heuristic is exactly maximum-Sharpe (exactly frontier-efficient) is zero.
\end{proposition}

\begin{proof}
The sets are images of $(0,\infty)$ and $(0,\infty)\times\R$ under injective affine maps, hence smooth manifolds of dimension $1$ and $2$; sets of dimension $< N$ have Lebesgue measure zero, and null sets have probability zero under absolutely continuous laws.
\end{proof}

\begin{theorem}[Second-order robustness; \textbf{P}]\label{thm:secondorder}
Fix $\bSigma \succ 0$, a fully invested $\w^\circ$ with $A := \w^{\circ\top}\bSigma\w^\circ > 0$, and a point on the coincidence ray $\bmu_e^0 = a\bSigma\w^\circ$, $a > 0$. Perturb in an arbitrary direction $v \in \R^N$: $\bmu_e(\varepsilon) = a\bSigma\w^\circ + \varepsilon v$. Then the squared Sharpe efficiency of $\w^\circ$ satisfies, exactly,
\begin{equation}\label{eq:exact-eta}
\eta^2(\varepsilon) \;=\; \frac{(aA + \varepsilon b)^2}{A\,\bigl(a^2A + 2a\varepsilon b + \varepsilon^2 c\bigr)},
\qquad b := v^\top\w^\circ,\quad c := v^\top\bSigma^{-1}v,
\end{equation}
and therefore
\begin{equation}\label{eq:secondorder}
1 - \eta^2(\varepsilon) \;=\; \frac{\varepsilon^2}{a^2}\cdot\frac{cA - b^2}{A^2} \;+\; O(\varepsilon^3),
\qquad cA - b^2 \;\ge\; 0,
\end{equation}
with $cA - b^2 = 0$ iff $v$ is itself on the ray ($v \propto \bSigma\w^\circ$). In particular $\frac{d}{d\varepsilon}\eta^2\big|_{\varepsilon=0} = 0$ for \emph{every} direction $v$: the efficiency map is stationary on the coincidence set, and the leading loss is quadratic with coefficient equal to the $\bSigma^{-1}$-residual of $v$ after projection on the ray.
\end{theorem}

\begin{proof}
$S(\w^\circ)^2 = (\w^{\circ\top}\bmu_e)^2/A = (aA + \varepsilon b)^2/A$ and $S_{\max}^2 = \bmu_e^\top\bSigma^{-1}\bmu_e = a^2A + 2a\varepsilon b + \varepsilon^2 c$, giving \eqref{eq:exact-eta}. Expanding,
\[
\eta^2 = \frac{a^2A^2 + 2a\varepsilon bA + \varepsilon^2 b^2}{a^2A^2 + 2a\varepsilon bA + \varepsilon^2 cA}
= 1 - \varepsilon^2\,\frac{cA - b^2}{a^2A^2} + O(\varepsilon^3).
\]
Nonnegativity of $cA - b^2$ is Cauchy--Schwarz in the $\bSigma^{-1}$-inner product applied to $v$ and $\bSigma\w^\circ$: $\langle v, \bSigma\w^\circ\rangle_{\bSigma^{-1}} = v^\top\w^\circ = b$ and $\|\bSigma\w^\circ\|^2_{\bSigma^{-1}} = A$, so $b^2 \le cA$ with equality iff $v \parallel \bSigma\w^\circ$.
\end{proof}

\begin{remark}[Never optimal, rarely costly]\label{rem:flat}
Proposition~\ref{prop:generic} and Theorem~\ref{thm:secondorder} jointly formalize the folklore. Exact optimality of any fixed rule is a probability-zero event; but because the efficiency surface has a vanishing gradient on the entire coincidence ray, an order-$\varepsilon$ misspecification of premia costs only order-$\varepsilon^2$ in squared Sharpe efficiency. The same flatness, viewed from the optimizer's side, is the source of mean--variance instability --- the optimum moves at order $\varepsilon$ while the objective moves at order $\varepsilon^2$ \citep{bestgrauer1991,michaud1989,chopraziemba1993} --- so a single geometric fact explains simultaneously why heuristics are cheap and why optimizers are fragile.
\end{remark}

\subsection{The catalogue}

\begin{definition}[Allocation rules]\label{def:rules}
All rules are fully invested and long-only by construction.
\begin{align}
\textnormal{Equal weight (EW):}\quad & w_i = 1/N. \\
\textnormal{Inverse volatility (IV):}\quad & w_i = \sigma_i^{-1}\Big/\textstyle\sum_j \sigma_j^{-1}. \\
\textnormal{Inverse variance (IVar):}\quad & w_i = \sigma_i^{-2}\Big/\textstyle\sum_j \sigma_j^{-2}. \\
\textnormal{Global minimum variance (GMV):}\quad & \w_{\mathrm{GMV}} = \bSigma^{-1}\ones\big/\ones^\top\bSigma^{-1}\ones. \\
\textnormal{Equal risk contribution (ERC):}\quad & w_i\,(\bSigma\w)_i = \frac{\w^\top\bSigma\w}{N} \quad \forall i,\;\; \w > 0,\;\ones^\top\w = 1. \\
\textnormal{Maximum diversification (MD):}\quad & \w_{\mathrm{MD}} = \argmax_{\ones^\top\w = 1}\; \frac{\w^\top\bsig}{\sqrt{\w^\top\bSigma\w}}. \\
\textnormal{Hierarchical risk parity (HRP):}\quad & \textnormal{Definition~\ref{def:hrp} below.}
\end{align}
The MD portfolio satisfies $\w_{\mathrm{MD}} \propto \bSigma^{-1}\bsig$ by the argument of Theorem~\ref{thm:foc}(i) with $\bsig$ in place of $\bmu_e$; when $\bSigma^{-1}\bsig$ has negative entries the long-only MD portfolio differs, and our statements concern the unconstrained solution. Existence and uniqueness of ERC is Theorem~\ref{thm:ercexist} below.
\end{definition}

\section{Exact Conditions: Equal Weight, Inverse Volatility, Minimum Variance, Maximum Diversification}\label{sec:flat}

\subsection{Equal weight}

\begin{theorem}[\textbf{P}]\label{thm:ew}
The equal-weight portfolio coincides with the tangency portfolio if and only if
\begin{equation}\label{eq:ewcond}
\bmu_e \;\propto\; \bSigma\ones \qquad\textnormal{(positive constant)},
\end{equation}
i.e.\ iff every asset's expected excess return is proportional to its covariance with the equal-weight portfolio --- equivalently, every asset has zero alpha in \eqref{eq:zeroalpha} with $\w^\circ = \ones/N$. It is frontier-efficient iff $\bmu \in \{a\bSigma\ones + b\ones : a > 0,\, b\in\R\}$.
\end{theorem}

\begin{proof}
Theorem~\ref{thm:foc} with $\w^\circ = \ones/N$. For the beta form: with $\bmu_e = c\,\bSigma\ones$, $c>0$, $\beta_i\cdot(\ones^\top\bmu_e/N) = \frac{N(\bSigma\ones)_i}{\ones^\top\bSigma\ones}\cdot\frac{c\,\ones^\top\bSigma\ones}{N} = c(\bSigma\ones)_i = \mu_{e,i}$.
\end{proof}

\begin{corollary}[Symmetry suffices; \textbf{P}]\label{cor:ewsym}
If $\bmu_e = m\ones$ with $m>0$ and $\bSigma\ones = c\ones$ (equal row sums; e.g.\ equal volatilities with constant correlation, or any $(\bmu_e,\bSigma)$ invariant under all permutations of assets), then EW is the tangency portfolio.
\end{corollary}

\begin{proof}
$\bmu_e = m\ones = (m/c)\bSigma\ones$, $c = \ones^\top\bSigma\ones/N > 0$. Permutation invariance forces both hypotheses.
\end{proof}

Note what is \emph{not} required: volatilities need not be equal and correlations need not be constant. The exact requirement couples means to covariance row sums; an asset may carry twice the covariance with the market of another provided it also carries twice the risk premium. Section~\ref{sec:bayes} upgrades the symmetry observation from a sufficient condition on parameters to a Bayes-optimality theorem on beliefs.

\subsection{Inverse volatility}

\begin{theorem}[\textbf{P}]\label{thm:iv}
Write $\bSigma = DRD$. The inverse-volatility portfolio coincides with the tangency portfolio if and only if
\begin{equation}\label{eq:ivcond}
S_i \;\propto\; (R\ones)_i \;=\; \sum_{j=1}^N \rho_{ij} \qquad \forall i
\quad\textnormal{(positive constant)}:
\end{equation}
each asset's Sharpe ratio must be proportional to the row sum of its correlations.
\end{theorem}

\begin{proof}
$\w_{\mathrm{IV}} \propto D^{-1}\ones$; coincidence iff $\bmu_e \propto \bSigma D^{-1}\ones = DRD\,D^{-1}\ones = DR\ones$, i.e.\ $\mu_{e,i} \propto \sigma_i(R\ones)_i$, i.e.\ $S_i \propto (R\ones)_i$, positive constants throughout.
\end{proof}

\begin{corollary}[\textbf{P}]\label{cor:ivsuff}
If $S_i \equiv S > 0$ and $R$ has equal row sums, $R\ones = \kappa\ones$ --- in particular under constant pairwise correlation $\rho > -1/(N-1)$, where $\kappa = 1+(N-1)\rho > 0$ --- then IV is the tangency portfolio.
\end{corollary}

\subsection{Minimum variance}

\begin{theorem}[\textbf{P}]\label{thm:gmv}
The GMV portfolio coincides with the tangency portfolio if and only if expected excess returns are equal across assets: $\bmu_e = m\ones$, $m > 0$.
\end{theorem}

\begin{proof}
$\bSigma^{-1}\bmu_e \propto \bSigma^{-1}\ones$ (positive constant) iff $\bmu_e \propto \ones$, since $\bSigma^{-1}$ is a bijection.
\end{proof}

\begin{remark}
GMV is always frontier-efficient --- the vertex of the frontier --- so the content of Theorem~\ref{thm:gmv} concerns coincidence with the portfolio every investor holds when a riskless asset exists. The equal-premium condition is the formal statement of the agnosticism motivating minimum-variance investing \citep{clarke2011}: GMV is the optimizer of an investor who believes risk is forecastable and premia are not. Proposition~\ref{prop:bayesgmv} makes the belief statement exact.
\end{remark}

\subsection{Maximum diversification}

\begin{theorem}[\textbf{P}]\label{thm:md}
The (unconstrained) maximum-diversification portfolio coincides with the tangency portfolio if and only if all Sharpe ratios are equal: $\bmu_e \propto \bsig$, positive constant. No restriction on $R$ is required.
\end{theorem}

\begin{proof}
$\w_{\mathrm{MD}} \propto \bSigma^{-1}\bsig$, $\w_T \propto \bSigma^{-1}\bmu_e$; equality up to positive scaling iff $\bmu_e \propto \bsig$. At the level of objectives: if $\bmu_e = S\bsig$ then $S(\w) = S\cdot\w^\top\bsig/\sigma_p(\w)$, so the Sharpe functional and the diversification ratio are positive multiples and share a maximizer.
\end{proof}

This was the design premise of \citet{choueifaty2008}: MD is the bet that risk is priced uniformly per unit of volatility. Comparing Theorems~\ref{thm:iv} and~\ref{thm:md} quantifies the gap between IV and MD: both encode equal Sharpe ratios, but IV additionally requires equal correlation row sums, because it ignores the correlation structure MD inverts.

\section{One-Factor Structure: Conditions in Observables}\label{sec:onefactor}

The conditions above restrict the unobservable pair $(\bmu_e,\bSigma)$. Under one-factor structure they become statements about betas and idiosyncratic variances --- quantities the practitioner actually estimates.

\begin{assumption}[Priced one-factor model]\label{ass:onefactor}
$r_{e,i} = \beta_i f_e + u_i$ with $\E[f_e] = \lambda_m > 0$, $\Var(f_e) = \sigma_m^2$, $\E[u_i] = 0$, $\Cov(u_i,u_j) = \delta_i^2\,\mathbf{1}\{i=j\}$, $\Cov(f_e, u_i) = 0$. Hence $\bmu_e = \lambda_m\bbeta$ and $\bSigma = \sigma_m^2\bbeta\bbeta^\top + \Delta$, $\Delta = \diag(\delta_1^2,\dots,\delta_N^2) \succ 0$. Exact pricing ($\bm\alpha = 0$) is part of the assumption.
\end{assumption}

\begin{lemma}[Sherman--Morrison; \textbf{P}]\label{lem:sm}
Under Assumption~\ref{ass:onefactor}, with $q := \bbeta^\top\Delta^{-1}\bbeta$,
\[
\bSigma^{-1} = \Delta^{-1} - \frac{\sigma_m^2\,\Delta^{-1}\bbeta\bbeta^\top\Delta^{-1}}{1 + \sigma_m^2 q},
\qquad
\bSigma^{-1}\bbeta = \frac{\Delta^{-1}\bbeta}{1+\sigma_m^2 q}.
\]
\end{lemma}

\begin{proof}
The first display is the Sherman--Morrison identity for the rank-one update $\bSigma = \Delta + (\sigma_m\bbeta)(\sigma_m\bbeta)^\top$. Applying it to $\bbeta$: $\bSigma^{-1}\bbeta = \Delta^{-1}\bbeta - \sigma_m^2\Delta^{-1}\bbeta\,q/(1+\sigma_m^2 q) = \Delta^{-1}\bbeta\,/(1+\sigma_m^2 q)$.
\end{proof}

\begin{proposition}[Tangency and GMV in observables; \textbf{P}; GMV form after \citealp{clarke2011}]\label{prop:onefactor}
Under Assumption~\ref{ass:onefactor}:
\begin{enumerate}
\item[(i)] The tangency portfolio is $\w_T \propto \Delta^{-1}\bbeta$, i.e.\ $w_{T,i} \propto \beta_i/\delta_i^2$.
\item[(ii)] The GMV portfolio admits the threshold-beta representation
\[
w_{\mathrm{GMV},i} \;\propto\; \frac{1}{\delta_i^2}\Bigl(1 - \frac{\beta_i}{\beta_L}\Bigr),
\qquad
\beta_L \;=\; \frac{1 + \sigma_m^2\,\bbeta^\top\Delta^{-1}\bbeta}{\sigma_m^2\,\bbeta^\top\Delta^{-1}\ones},
\]
(assuming $\bbeta^\top\Delta^{-1}\ones > 0$): assets with $\beta_i > \beta_L$ are shorted, and long-only GMV concentrates in low-beta, low-idiosyncratic-risk names.
\item[(iii)] Coincidence conditions in observables: EW is tangency iff $\beta_i \propto \delta_i^2$; IVar is tangency iff $\beta_i \propto \delta_i^2/\sigma_i^2$; IV is tangency iff $\beta_i \propto \delta_i^2/\sigma_i$; GMV is tangency iff $\beta_i$ is constant (consistent with (ii): constant betas give $\w_{\mathrm{GMV}} \propto \Delta^{-1}\ones \propto \w_T$); MD is tangency iff $\beta_i \propto \sigma_i$, i.e.\ iff all assets have the same correlation with the factor, $\Corr(r_i, f) = \beta_i\sigma_m/\sigma_i$ constant.
\end{enumerate}
\end{proposition}

\begin{proof}
(i) $\w_T \propto \bSigma^{-1}\bmu_e = \lambda_m\bSigma^{-1}\bbeta \propto \Delta^{-1}\bbeta$ by Lemma~\ref{lem:sm}. (ii) Apply Lemma~\ref{lem:sm} to $\ones$:
\[
\bSigma^{-1}\ones = \Delta^{-1}\ones - \frac{\sigma_m^2(\bbeta^\top\Delta^{-1}\ones)}{1+\sigma_m^2 q}\,\Delta^{-1}\bbeta
\;\Longrightarrow\;
(\bSigma^{-1}\ones)_i = \frac{1}{\delta_i^2}\Bigl(1 - \frac{\beta_i}{\beta_L}\Bigr)
\]
with $\beta_L$ as displayed. (iii) Each statement is Theorem~\ref{thm:foc}(i): the heuristic must be proportional to $\Delta^{-1}\bbeta$. EW: $\ones \propto \Delta^{-1}\bbeta \iff \beta_i \propto \delta_i^2$. IVar: $D^{-2}\ones \propto \Delta^{-1}\bbeta \iff \beta_i \propto \delta_i^2/\sigma_i^2$. IV: $D^{-1}\ones \propto \Delta^{-1}\bbeta \iff \beta_i \propto \delta_i^2/\sigma_i$. GMV: $\bmu_e \propto \ones$ (Theorem~\ref{thm:gmv}) iff $\beta_i$ constant. MD: $\bmu_e \propto \bsig$ (Theorem~\ref{thm:md}) iff $\lambda_m\beta_i \propto \sigma_i$.
\end{proof}

Statement (iii) is the practical payoff: in a priced single-factor world, the question ``which heuristic should be closest to optimal in this universe'' reduces to inspecting the cross-sectional relation between $\beta_i$, $\delta_i^2$, and $\sigma_i$ --- no premium estimates required, because exact pricing has already collapsed $\bmu_e$ onto $\bbeta$.

\section{Risk Parity: Existence, Uniqueness, and Optimality as Performance Parity}\label{sec:erc}

Risk contributions $\mathrm{RC}_i(\w) = w_i(\bSigma\w)_i$ satisfy the Euler identity $\sum_i \mathrm{RC}_i = \sigma_p^2$, which underwrites their interpretation as a decomposition of risk \citep{qian2006}. Define analogously the \emph{performance contribution} $\mathrm{PC}_i(\w) = w_i\,\mu_{e,i}$, with $\sum_i \mathrm{PC}_i = \w^\top\bmu_e$.

\subsection{Existence and uniqueness}

\begin{theorem}[\textbf{P}; cf.\ \citealp{maillard2010,roncalli2013}]\label{thm:ercexist}
For every $\bSigma \succ 0$ there exists exactly one fully invested $\w > 0$ with equal risk contributions.
\end{theorem}

\begin{proof}
Fix $c > 0$ and consider $F_c(\w) = \tfrac12\w^\top\bSigma\w - c\sum_{i}\ln w_i$ on the open orthant $\R^N_{++}$. $F_c$ is strictly convex ($\bSigma\succ0$ plus convexity of $-\ln$) and coercive: $F_c \to \infty$ as $\|\w\|\to\infty$ (the quadratic dominates the logarithm) and as any $w_i \downarrow 0$ (the barrier diverges; on bounded sets the quadratic and the remaining logarithms are bounded below). Hence a unique minimizer $\w(c) \in \R^N_{++}$ exists, characterized by the first-order condition $\bSigma\w = c\,\w^{-1}$ (elementwise reciprocal), i.e.\ $w_i(\bSigma\w)_i = c$ for all $i$: every barrier solution is an unnormalized ERC portfolio. The family is a single ray: if $\bSigma\w = c\,\w^{-1}$ then $\bSigma(t\w) = (t^2c)\,(t\w)^{-1}$, so $\w(t^2 c) = t\,\w(c)$. Normalizing to the simplex gives existence. For uniqueness, suppose $\mathbf{u}, \mathbf{v}$ are both fully invested ERC portfolios with constants $c_u = \mathbf{u}^\top\bSigma\mathbf{u}/N$, $c_v$. Then $\mathbf{u} = \w(c_u)$ and $\mathbf{v} = \w(c_v)$ by uniqueness of barrier minimizers, and by the ray property $\mathbf{v} = t\,\mathbf{u}$ with $t = \sqrt{c_v/c_u}$; full investment of both forces $t = 1$.
\end{proof}

\subsection{Optimality}

\begin{theorem}[ERC optimality $=$ performance parity; \textbf{P}]\label{thm:ercchar}
Let $\w$ be the ERC portfolio. Then $\w$ coincides with the tangency portfolio if and only if performance contributions are also equal:
\begin{equation}\label{eq:perfparity}
w_i\,\mu_{e,i} \;=\; \frac{\w^\top\bmu_e}{N} \qquad \forall i .
\end{equation}
\end{theorem}

\begin{proof}
Since $\w$ is ERC, $(\bSigma\w)_i = \sigma_p^2/(N w_i)$, $w_i > 0$. By Theorem~\ref{thm:foc}(i), coincidence holds iff $\bmu_e = a\bSigma\w$, $a>0$, i.e.\ iff $\mu_{e,i} = a\sigma_p^2/(Nw_i)$ for all $i$, i.e.\ iff $w_i\mu_{e,i}$ is constant; summing identifies the constant, and $a > 0$ iff $\w^\top\bmu_e > 0$.
\end{proof}

The economic reading is sharp: a risk parity portfolio is mean--variance optimal precisely when the equal risk budget doubles as an equal \emph{return} budget --- each line contributes the same expected excess return in currency terms. When risk budgets and performance budgets disagree under the true $\bmu_e$, ERC is on neither the tangency ray nor, generically, the frontier.

\begin{corollary}[Frontier version; \textbf{P}]\label{cor:ercfrontier}
The ERC portfolio is frontier-efficient iff there exist $\kappa > 0$, $b \in \R$ with $w_i\mu_i = \kappa + b\,w_i$ for all $i$: performance contributions affine in weights, with the tangency case $b = r_f\!\cdot\!1$-adjusted to \eqref{eq:perfparity}.
\end{corollary}

\begin{proof}
Theorem~\ref{thm:foc}(ii) with $(\bSigma\w)_i = \sigma_p^2/(Nw_i)$: $\mu_i = a\sigma_p^2/(Nw_i) + b \iff w_i\mu_i = a\sigma_p^2/N + b\,w_i$.
\end{proof}

\begin{theorem}[Budgeted generalization; \textbf{P}]\label{thm:budget}
Let $b_i > 0$, $\sum_i b_i = 1$, and let $\w$ satisfy $\mathrm{RC}_i(\w) = b_i\,\sigma_p^2$. Then $\w$ coincides with the tangency portfolio iff $\mathrm{PC}_i(\w) = b_i\,\w^\top\bmu_e$ for all $i$: performance must be budgeted in the same proportions as risk.
\end{theorem}

\begin{proof}
As in Theorem~\ref{thm:ercchar} with $\sigma_p^2/N$ replaced by $b_i\sigma_p^2$.
\end{proof}

\begin{proposition}[Constant correlation collapses ERC to IV; \textbf{P}, after \citealp{maillard2010}]\label{prop:ercconst}
If $R = (1-\rho)I + \rho\,\ones\ones^\top$ with $\rho > -1/(N-1)$, the ERC portfolio is the inverse-volatility portfolio.
\end{proposition}

\begin{proof}
Take $w_i = \sigma_i^{-1}/T$, $T = \sum_j\sigma_j^{-1}$. Then $(\bSigma\w)_i = \sigma_i\sum_j\rho_{ij}\sigma_jw_j = \frac{\sigma_i}{T}\bigl(1+(N-1)\rho\bigr)$, so $w_i(\bSigma\w)_i = (1+(N-1)\rho)/T^2$ is constant. Conclude by Theorem~\ref{thm:ercexist}.
\end{proof}

\begin{corollary}[Sufficient condition of \citealp{maillard2010}; \textbf{P}]\label{cor:mrt}
Equal Sharpe ratios and constant pairwise correlation imply that ERC coincides with the tangency portfolio.
\end{corollary}

\begin{proof}
Proposition~\ref{prop:ercconst} gives ERC $=$ IV; Corollary~\ref{cor:ivsuff} gives IV $=$ tangency.
\end{proof}

\begin{remark}[The sufficient condition is far from necessary; \textbf{P}]\label{rem:ercnotnec}
By Corollary~\ref{cor:reverse}, take \emph{any} $\bSigma \succ 0$ --- in particular one with wildly non-constant correlations --- compute its ERC portfolio $\w$, and set $\bmu_e = \bSigma\w$. Then ERC is exactly the tangency portfolio while the \citet{maillard2010} hypotheses fail. The exact optimality set is the performance-parity condition \eqref{eq:perfparity}, a codimension-$(N-1)$ restriction; equal-Sharpe-plus-constant-correlation is a thin, interpretable slice of it.
\end{remark}

\begin{proposition}[Variance ordering and interpolation; \textbf{C}]\label{prop:order}
On the long-only simplex one has the bracketing
$\sigma_p(\w^{+}_{\mathrm{GMV}}) \le \sigma_p(\w_{\mathrm{ERC}}) \le \sigma_p(\w_{\mathrm{EW}})$,
where $\w^{+}_{\mathrm{GMV}}$ is long-only GMV. The proof of \citet{maillard2010} runs through the variational problem of Theorem~\ref{thm:ercexist} in constrained form, $\min\{\sqrt{\w^\top\bSigma\w} : \sum_i\ln w_i \ge c,\ \ones^\top\w=1,\ \w\ge0\}$: as $c$ increases to its maximal feasible value $-N\ln N$ (attained only at EW, by strict concavity of $\sum_i\ln w_i$ on the simplex), the feasible set shrinks from the full simplex to $\{\w_{\mathrm{EW}}\}$ and the minimized volatility increases monotonically; ERC arises at an intermediate level. See also \citet{roncalli2013}.
\end{proposition}

The interpolation reading anticipates Section~\ref{sec:ambiguity}: ERC sits between the rule that trusts $\bSigma$ completely (GMV) and the rule that trusts nothing (EW), exactly as an estimator shrunk between a sample-efficient and a maximally robust target.

\section{Hierarchical and Return-Adjusted Hierarchical Risk Parity}\label{sec:hrp}

\subsection{The algorithm}

\begin{definition}[HRP; \citealp{lopezdeprado2016}]\label{def:hrp}
Given $\bSigma$ with correlation matrix $R$:
\begin{enumerate}
\item[(1)] \emph{Tree clustering.} Compute $d_{ij} = \sqrt{(1-\rho_{ij})/2}$ and build a dendrogram $\mathcal{T}$ by hierarchical (in the original, single-linkage) clustering.
\item[(2)] \emph{Quasi-diagonalization.} Reorder assets so the dendrogram's leaves are contiguous.
\item[(3)] \emph{Recursive bisection.} Initialize $w_i = 1$. Recursively split each cluster $C$ into $A$ and $B$; assign
\begin{equation}\label{eq:hrpsplit}
\alpha_A = \frac{\widetilde V_B}{\widetilde V_A + \widetilde V_B}, \qquad \alpha_B = 1 - \alpha_A,
\qquad
\widetilde V_K = \tilde\w_K^\top \bSigma_K \tilde\w_K, \quad
\tilde w_{K,i} = \frac{\sigma_i^{-2}}{\sum_{j\in K}\sigma_j^{-2}},
\end{equation}
multiply the running weights in $A$ by $\alpha_A$ and in $B$ by $\alpha_B$, and recurse to singletons.
\end{enumerate}
We call the bisection \emph{tree-aligned} if every split occurs at a node of $\mathcal{T}$, and \emph{balanced} if the two sub-clusters at every split have equal cardinality. (The original implementation splits the quasi-diagonalized list in half; common variants split at dendrogram nodes. Our exactness results hold for whichever scheme satisfies the stated hypotheses.)
\end{definition}

\subsection{What the optimizer actually does at a split}

To see precisely which information HRP discards, write the minimum-variance solution of a two-block system in split form.

\begin{theorem}[Exact Schur bisection identity; \textbf{P}]\label{thm:schur}
Consider a partition $C = A \cup B$ with
$\bSigma_C = \begin{pmatrix} \bSigma_A & \bSigma_{AB} \\ \bSigma_{BA} & \bSigma_B \end{pmatrix} \succ 0$.
Define the Schur complements and tilted budget vectors
\begin{align*}
\mathbf{S}_A &= \bSigma_A - \bSigma_{AB}\bSigma_B^{-1}\bSigma_{BA}, &
\mathbf{S}_B &= \bSigma_B - \bSigma_{BA}\bSigma_A^{-1}\bSigma_{AB}, \\
\mathbf{b}_A &= \ones_A - \bSigma_{AB}\bSigma_B^{-1}\ones_B, &
\mathbf{b}_B &= \ones_B - \bSigma_{BA}\bSigma_A^{-1}\ones_A .
\end{align*}
Then the unnormalized GMV vector splits exactly as
\begin{equation}\label{eq:schursplit}
\bigl(\bSigma_C^{-1}\ones_C\bigr)_A = \mathbf{S}_A^{-1}\mathbf{b}_A,
\qquad
\bigl(\bSigma_C^{-1}\ones_C\bigr)_B = \mathbf{S}_B^{-1}\mathbf{b}_B,
\end{equation}
so the GMV mass allocated to $A$ is $\ones_A^\top\mathbf{S}_A^{-1}\mathbf{b}_A \big/ \bigl(\ones_A^\top\mathbf{S}_A^{-1}\mathbf{b}_A + \ones_B^\top\mathbf{S}_B^{-1}\mathbf{b}_B\bigr)$ and the within-$A$ composition is $\mathbf{S}_A^{-1}\mathbf{b}_A$ renormalized.
\end{theorem}

\begin{proof}
By the block-inverse formula for $\bSigma_C \succ 0$ (both Schur complements are positive definite),
\[
\bigl(\bSigma_C^{-1}\bigr)_{AA} = \mathbf{S}_A^{-1}, \qquad
\bigl(\bSigma_C^{-1}\bigr)_{AB} = -\,\mathbf{S}_A^{-1}\bSigma_{AB}\bSigma_B^{-1},
\]
hence $(\bSigma_C^{-1}\ones_C)_A = \mathbf{S}_A^{-1}\ones_A - \mathbf{S}_A^{-1}\bSigma_{AB}\bSigma_B^{-1}\ones_B = \mathbf{S}_A^{-1}\mathbf{b}_A$; symmetrically for $B$.
\end{proof}

\begin{corollary}[The three substitutions; \textbf{P}]\label{cor:threesubs}
At every split, HRP computes its allocation from the triple $(\bSigma_K,\ \ones_K,\ \textnormal{inverse-variance weights})$ where the exact recursion \eqref{eq:schursplit} requires $(\mathbf{S}_K,\ \mathbf{b}_K,\ \textnormal{minimum-variance weights})$. The algorithm therefore commits exactly three substitutions:
(\emph{i}) raw block for Schur complement, $\bSigma_K \leftarrow \mathbf{S}_K$ --- discarding the variance reduction available from cross-hedging against the sibling cluster;
(\emph{ii}) flat budget for tilted budget, $\ones_K \leftarrow \mathbf{b}_K$ --- discarding the demand tilt induced by cross-cluster covariance;
(\emph{iii}) inverse variance for within-cluster minimum variance --- discarding within-cluster correlations in the composition and in the cluster-variance functional used by the split rule.
All three are exact when $\bSigma_{AB} = 0$ \emph{and} the within-cluster IV portfolio is the within-cluster GMV; the theorems below identify when that happens, and \citet{cotton2024} constructs the one-parameter family that reinstates (i)--(ii) continuously, interpolating between HRP and exact GMV, with the recursive Schur reduction developed further by \citet{mograby2025}.
\end{corollary}

\subsection{Exact coincidence results}

\begin{lemma}[Harmonic aggregation; \textbf{P}]\label{lem:harmonic}
If $\bSigma_K$ is diagonal, the cluster variance in \eqref{eq:hrpsplit} equals
$\widetilde V_K = \bigl(\sum_{i\in K}\sigma_i^{-2}\bigr)^{-1} = 1/T_K$, $T_K := \sum_{i\in K}\sigma_i^{-2}$,
and the inverse-variance weights are exactly the GMV weights of $\bSigma_K$, with $\widetilde V_K = (\ones^\top\bSigma_K^{-1}\ones)^{-1}$.
\end{lemma}

\begin{proof}
$\widetilde V_K = \sum_{i\in K}\tilde w_{K,i}^2\sigma_i^2 = \sum_{i\in K}\sigma_i^{-2}/T_K^2 = 1/T_K$; and $\bSigma_K^{-1}\ones = (\sigma_i^{-2})_{i\in K}$ gives the rest.
\end{proof}

\begin{theorem}[Diagonal exactness; \textbf{P}]\label{thm:hrpdiag}
If $\bSigma$ is diagonal, then for every dendrogram and every bisection scheme HRP returns the inverse-variance portfolio, which is the GMV portfolio:
$w^{\mathrm{HRP}}_i = \sigma_i^{-2}/\sum_j\sigma_j^{-2} = w^{\mathrm{GMV}}_i$.
\end{theorem}

\begin{proof}
By Lemma~\ref{lem:harmonic}, at any split of $C$ into $A,B$,
$\alpha_A = \frac{1/T_B}{1/T_A + 1/T_B} = \frac{T_A}{T_A+T_B} = \frac{T_A}{T_C}$.
Asset $i$'s final weight is the product of fractions along its root-to-leaf path, which telescopes:
$w^{\mathrm{HRP}}_i = \prod \frac{T_{\mathrm{child}}}{T_{\mathrm{parent}}} = \frac{T_{\{i\}}}{T_{\mathrm{root}}} = \frac{\sigma_i^{-2}}{\sum_j\sigma_j^{-2}}$,
independent of the tree; GMV for diagonal $\bSigma$ is the same vector. Equivalently: $\bSigma_{AB} = 0$ everywhere and IV $=$ GMV in every cluster, so all three substitutions of Corollary~\ref{cor:threesubs} are identities.
\end{proof}

\begin{proposition}[Uncorrelated exchangeable blocks; \textbf{P}]\label{prop:hrpblocks}
Let $\bSigma$ be block-diagonal with two blocks, block $K \in \{A,B\}$ exchangeable:
$\bSigma_K = \sigma_K^2[(1-\rho_K)I + \rho_K\ones\ones^\top]$ with $n_K$ assets, $\rho_K > -1/(n_K-1)$, $n_K$ a power of two. Suppose the bisection is tree-aligned with the top split at the block boundary and balanced within blocks. Then HRP $=$ GMV exactly:
\[
w^{\mathrm{GMV}}_{i\in K} = \frac{m_K}{n_K(m_A+m_B)},
\qquad
m_K := \ones^\top\bSigma_K^{-1}\ones = \frac{n_K}{\sigma_K^2(1+(n_K-1)\rho_K)} .
\]
\end{proposition}

\begin{proof}
\emph{Within blocks.} Equal volatilities make inverse-variance weights on any sub-cluster equal weights. A sub-cluster of size $k$ inside block $K$ has cluster variance at equal weights
$\widetilde V(k) = \sigma_K^2(1+(k-1)\rho_K)/k$.
Balanced bisection gives both children equal size, hence equal $\widetilde V$, hence $1/2$--$1/2$ splits; recursing down $n_K = 2^q$ yields equal weights $1/n_K$ within the block. Since $\bSigma_K\ones = \sigma_K^2(1+(n_K-1)\rho_K)\ones$, block GMV is also equal weights.
\emph{Across blocks.} At the top split, $\widetilde V(n_K) = \sigma_K^2(1+(n_K-1)\rho_K)/n_K = 1/m_K$, so $\alpha_A = m_A/(m_A+m_B)$, which is the GMV mass on $A$: block-diagonality stacks $\bSigma_A^{-1}\ones$ and $\bSigma_B^{-1}\ones$, with block masses $m_K/(m_A+m_B)$. (In the language of Corollary~\ref{cor:threesubs}: $\bSigma_{AB}=0$ makes (i)--(ii) identities at the top split; within blocks the substitutions are not identities individually, but exchangeability plus balance makes their composition weight-neutral.)
\end{proof}

\begin{remark}[Minimal counterexample; \textbf{P}]\label{rem:hrpcounter}
Exactness fails as soon as a split must absorb within-cluster correlation with unequal volatilities. For $N = 2$, correlation $\rho$, volatilities $\sigma_1 \neq \sigma_2$: HRP ignores $\rho$ (the children are singletons) and returns
$w^{\mathrm{HRP}}_1 = \sigma_2^2/(\sigma_1^2+\sigma_2^2)$, while
$w^{\mathrm{GMV}}_1 = (\sigma_2^2 - \rho\sigma_1\sigma_2)/(\sigma_1^2+\sigma_2^2-2\rho\sigma_1\sigma_2)$.
These coincide iff $\rho = 0$ or $\sigma_1 = \sigma_2$ (when $\sigma_1=\sigma_2$ both equal $1/2$ for every $\rho$). The hypotheses of Theorem~\ref{thm:hrpdiag} and Proposition~\ref{prop:hrpblocks} --- zero correlation wherever volatilities differ within a cluster, equal volatilities wherever correlation survives --- are not artifacts of the proofs but the boundary of exactness, and the deviation is exactly substitution (iii) of Corollary~\ref{cor:threesubs} at the bottom split.
\end{remark}

\begin{corollary}[Mean--variance optimality of HRP; \textbf{P}]\label{cor:hrpmv}
Under the hypotheses of Theorem~\ref{thm:hrpdiag} or Proposition~\ref{prop:hrpblocks}, HRP coincides with the tangency portfolio iff, in addition, $\bmu_e \propto \ones$ (Theorem~\ref{thm:gmv}). Under the diagonal hypothesis this reads: equal premia and uncorrelated assets, in which case HRP $=$ IVar $=$ GMV $=$ tangency simultaneously.
\end{corollary}

\begin{remark}[What HRP is actually for; \textbf{E}]\label{rem:hrpe}
None of the above is the operative defense of HRP. \citet{lopezdeprado2016} motivates the algorithm by the instability of $\hat\bSigma^{-1}$ under estimation error --- quasi-diagonalization and recursion avoid inverting an ill-conditioned matrix --- and reports Monte Carlo evidence of superior out-of-sample variance relative to plug-in minimum variance. That defense is empirical and belongs to the information-set logic of Sections~\ref{sec:bayes}--\ref{sec:ambiguity}, not to exact optimality; Corollary~\ref{cor:threesubs} gives it structure by listing precisely which estimated objects ($\bSigma_{AB}$, within-cluster correlations) HRP declines to trust.
\end{remark}

\subsection{Return-adjusted hierarchical risk parity}

Standard HRP is deliberately risk-only. RA-HRP keeps the clustering and quasi-diagonalization stages fixed, but replaces the inverse-variance split score by a local cluster-Sharpe score. This makes the rule materially different from the flat heuristics above: its weight vector is a function of both $\bSigma$ and $\bmu_e$, so the implied-return condition is no longer a fixed ray in $\bmu_e$ for fixed $\bSigma$. It is a nonlinear fixed point.

\begin{definition}[Fixed-tree RA-HRP; \textbf{P}, notation following \citealp{noguer2026rahrp}]\label{def:rahrp}
Fix a binary dendrogram $\mathcal T$ and a score floor $\varepsilon>0$. For every nonempty cluster $K\subset\{1,\dots,N\}$ define the inverse-variance portfolio
\[
\tilde w_{K,i}=\frac{\sigma_i^{-2}}{\sum_{j\in K}\sigma_j^{-2}},\qquad i\in K,
\]
its local mean, variance, and floored Sharpe score by
\begin{equation}\label{eq:ra-scores}
M_K := \sum_{i\in K}\tilde w_{K,i}\mu_{e,i},\qquad
V_K := \tilde\w_K^\top\bSigma_K\tilde\w_K,\qquad
Q_K := \max\left\{\frac{M_K}{\sqrt{V_K}},\varepsilon\right\}.
\end{equation}
At an internal node $C=A\cup B$, RA-HRP allocates the parent budget by
\begin{equation}\label{eq:ra-split}
\beta_A^{\mathrm{RA}}=\frac{Q_A}{Q_A+Q_B},\qquad
\beta_B^{\mathrm{RA}}=1-\beta_A^{\mathrm{RA}}=\frac{Q_B}{Q_A+Q_B}.
\end{equation}
The final portfolio obtained by multiplying these branch fractions along root-to-leaf paths is denoted
\[
\w^{\mathrm{RA}}=\Phi^{\mathrm{RA}}_{\mathcal T,\varepsilon}(\bmu_e,\bSigma).
\]
When the raw score $M_K/\sqrt{V_K}$ exceeds $\varepsilon$ for every cluster, the floor is inactive and the split is exactly a cluster-Sharpe split.
\end{definition}

\begin{theorem}[Well-posedness and lower bounds; \textbf{P}]\label{thm:rahrp-wellposed}
For fixed $(\bmu_e,\bSigma,\mathcal T,\varepsilon)$ with $\bSigma\succ0$ and $\varepsilon>0$, RA-HRP defines a unique portfolio $\w^{\mathrm{RA}}\in\Delta^{N-1}$. If $D_{\mathcal T}$ is the maximum root-to-leaf depth and $M_Q:=\max_{K\in\mathcal T} Q_K$, then every leaf satisfies
\begin{equation}\label{eq:ra-lower}
w_i^{\mathrm{RA}}\ge \left(\frac{\varepsilon}{\varepsilon+M_Q}\right)^{D_{\mathcal T}}.
\end{equation}
On any region where the active/inactive status of the floor does not change and all cluster variances remain positive, the map $(\bmu_e,\bSigma)\mapsto \Phi^{\mathrm{RA}}_{\mathcal T,\varepsilon}(\bmu_e,\bSigma)$ is smooth.
\end{theorem}

\begin{proof}
At every node, $Q_A,Q_B\ge\varepsilon$, so the denominator in \eqref{eq:ra-split} is strictly positive and
\[
\frac{\varepsilon}{\varepsilon+M_Q}\le \beta_A^{\mathrm{RA}},\beta_B^{\mathrm{RA}}\le \frac{M_Q}{\varepsilon+M_Q}<1.
\]
The tree has finitely many nodes, and the recursive multiplication of deterministic split fractions therefore gives a unique strictly positive leaf weight. Budgets are conserved at every split because $\beta_A^{\mathrm{RA}}+\beta_B^{\mathrm{RA}}=1$, so the leaf weights sum to one. The lower bound follows by multiplying the worst possible branch fraction along a path of length at most $D_{\mathcal T}$. Smoothness on a fixed floor regime follows because the IV weights, $M_K$, $V_K$, $Q_K$, and the rational split map are compositions of smooth functions with positive denominators.
\end{proof}

\begin{definition}[Unit-free HRP--RA-HRP interpolation; \textbf{P}]
At an internal node $C=A\cup B$, write the standard HRP left-branch fraction and the RA-HRP left-branch fraction as
\begin{equation}\label{eq:unitfree-splits}
r_A:=\frac{V_B}{V_A+V_B},\qquad
s_A:=\frac{Q_A}{Q_A+Q_B},
\end{equation}
where $V_K$ and $Q_K$ are the cluster IVP variance and floored cluster-Sharpe score from \eqref{eq:ra-scores}. For $\lambda\in[0,1]$, define the interpolated local split
\begin{equation}\label{eq:unitfree-interp}
\beta_A^{(\lambda)}:=(1-\lambda)r_A+\lambda s_A,\qquad
\beta_B^{(\lambda)}:=1-\beta_A^{(\lambda)}.
\end{equation}
Applying \eqref{eq:unitfree-interp} recursively on the same tree defines the interpolated portfolio $\w^{(\lambda)}$.
\end{definition}

\begin{proposition}[Properties of the unit-free interpolation; \textbf{P}]\label{prop:unitfree-interp}
For the family \eqref{eq:unitfree-interp}: (i) $\lambda=0$ gives HRP; (ii) $\lambda=1$ gives RA-HRP; (iii) $\lambda\mapsto\beta_A^{(\lambda)}$ is affine and continuous on $[0,1]$; and (iv)
\begin{equation}\label{eq:unitfree-deriv}
\frac{d}{d\lambda}\beta_A^{(\lambda)}=s_A-r_A.
\end{equation}
Consequently the local path is monotone whenever the HRP and RA-HRP branch fractions differ.
\end{proposition}

\begin{proof}
Substitution of $\lambda=0$ and $\lambda=1$ gives the two endpoints. The map in \eqref{eq:unitfree-interp} is an affine function of $\lambda$, hence continuous and differentiable, with derivative \eqref{eq:unitfree-deriv}. Since both $r_A$ and $s_A$ lie in $(0,1)$, every interpolated split remains feasible.
\end{proof}

\begin{remark}[Why the interpolation is economically clean]\label{rem:unitfree-why}
The interpolation is performed on branch probabilities, not on raw scores. This matters because inverse variances and Sharpe ratios have different units and different estimation errors. Blending $r_A$ and $s_A$ gives $\lambda$ a literal interpretation: it is the local weight placed on the return-adjusted allocation decision relative to the risk-only decision. This is the correct homotopy for empirical sensitivity analysis and for any future theorem that treats RA-HRP as a controlled perturbation of HRP.
\end{remark}

\begin{theorem}[Exact RA-HRP tangency coincidence; \textbf{P}]\label{thm:rahrp-coincidence}
Fix $(\mathcal T,\varepsilon)$ and define
\[
\w^{\mathrm{RA}}=\Phi^{\mathrm{RA}}_{\mathcal T,\varepsilon}(\bmu_e,\bSigma).
\]
Assume $\ones^\top\bSigma^{-1}\bmu_e>0$. Then the following are equivalent:
\begin{enumerate}
\item[(i)] RA-HRP is the tangency portfolio: $\w^{\mathrm{RA}}=\w_T$.
\item[(ii)] There exists $a>0$ such that
\begin{equation}\label{eq:ra-fixedpoint}
\bmu_e = a\,\bSigma\,\Phi^{\mathrm{RA}}_{\mathcal T,\varepsilon}(\bmu_e,\bSigma).
\end{equation}
\item[(iii)] RA-HRP has zero alpha against itself: for every asset $i$,
\[
\mu_{e,i}=\beta_i(\w^{\mathrm{RA}})\,\bigl(\w^{\mathrm{RA}\top}\bmu_e\bigr).
\]
\end{enumerate}
Moreover, at any split $C=A\cup B$ of an exact tangency recursion, with local covariance--premium pair $(\Gamma_C,\nu_C)$, write
\[
\Gamma_C=\begin{pmatrix}\Gamma_A&\Gamma_{AB}\\\Gamma_{BA}&\Gamma_B\end{pmatrix},\qquad
S_A=\Gamma_A-\Gamma_{AB}\Gamma_B^{-1}\Gamma_{BA},\qquad
S_B=\Gamma_B-\Gamma_{BA}\Gamma_A^{-1}\Gamma_{AB},
\]
\[
\nu_A^{\mathrm{Sch}}=\nu_A-\Gamma_{AB}\Gamma_B^{-1}\nu_B,
\qquad
\nu_B^{\mathrm{Sch}}=\nu_B-\Gamma_{BA}\Gamma_A^{-1}\nu_A,
\]
and
\begin{equation}\label{eq:true-split-ra}
L_A:=\ones_A^\top S_A^{-1}\nu_A^{\mathrm{Sch}},\qquad
L_B:=\ones_B^\top S_B^{-1}\nu_B^{\mathrm{Sch}}.
\end{equation}
When $L_A,L_B>0$, the exact tangency split is $L_A/(L_A+L_B)$. Hence the RA-HRP split is locally exact at that node if and only if
\begin{equation}\label{eq:ra-node-exact}
\frac{Q_A}{Q_B}=\frac{L_A}{L_B}.
\end{equation}
Full tree coincidence holds if and only if \eqref{eq:ra-node-exact} holds at every internal node when the right-hand side is computed from the Schur-updated conditional pair generated by the exact tangency recursion.
\end{theorem}

\begin{proof}
The equivalence of (i) and (ii) is Theorem~\ref{thm:foc}(i), with the important difference that the candidate $\w^{\mathrm{RA}}$ itself depends on $(\bmu_e,\bSigma)$. The equivalence of (i) and (iii) is the beta form in Remark~\ref{rem:beta}.

For the split statement, the unnormalized tangency direction on a local pair $(\Gamma_C,\nu_C)$ is $x_C=\Gamma_C^{-1}\nu_C$. The block-inverse identity gives
\[
x_A=S_A^{-1}\bigl(\nu_A-\Gamma_{AB}\Gamma_B^{-1}\nu_B\bigr),\qquad
x_B=S_B^{-1}\bigl(\nu_B-\Gamma_{BA}\Gamma_A^{-1}\nu_A\bigr).
\]
Therefore the unnormalized capital masses assigned by the exact tangency direction to the two children are $L_A=\ones_A^\top x_A$ and $L_B=\ones_B^\top x_B$. If both are positive, the exact local budget share of $A$ is $L_A/(L_A+L_B)$. The RA-HRP local share is $Q_A/(Q_A+Q_B)$; equality is equivalent to the ratio identity \eqref{eq:ra-node-exact}. Recursing is legitimate because $x_A=S_A^{-1}\nu_A^{\mathrm{Sch}}$ is itself the tangency direction of the Schur-updated pair $(S_A,\nu_A^{\mathrm{Sch}})$ inside child $A$, and analogously for $B$. Matching every local split on the tree is equivalent to matching every root-to-leaf product of split fractions, hence to full portfolio coincidence.
\end{proof}

\begin{remark}[What RA-HRP fixes and what it does not]\label{rem:rahrp-substitutions}
RA-HRP repairs one weakness of standard HRP: it no longer allocates solely by inverse cluster variance, and therefore it can tilt toward clusters with superior estimated reward per unit of local risk. But Theorem~\ref{thm:rahrp-coincidence} shows that the exact Markowitz split is not a raw cluster-Sharpe split. The optimizer uses Schur-adjusted covariance $S_A,S_B$ and Schur-adjusted premia $\nu_A^{\mathrm{Sch}},\nu_B^{\mathrm{Sch}}$. RA-HRP therefore replaces HRP's purely risk-based local score by a return-aware local score, while still avoiding the exact Schur and global-inversion machinery. Its optimality set is larger and more economically flexible than HRP's equal-premium GMV set, but it remains a restrictive nonlinear fixed-point set rather than a generic property.
\end{remark}

\begin{definition}[Schur-RA-HRP local score; \textbf{P}]\label{def:schur-rahrp}
This is the local conditional-risk extension of RA-HRP developed in \citet{noguer2026rahrp}. At a sibling split $C=A\cup B$, partition the local covariance matrix as
\[
\bSigma_C=\begin{pmatrix}\bSigma_{AA}&\bSigma_{AB}\\ \bSigma_{BA}&\bSigma_{BB}\end{pmatrix}.
\]
Define the conditional Schur covariance blocks
\begin{equation}\label{eq:schur-cond-blocks}
\bSigma_{A\mid B}:=\bSigma_{AA}-\bSigma_{AB}\bSigma_{BB}^{-1}\bSigma_{BA},\qquad
\bSigma_{B\mid A}:=\bSigma_{BB}-\bSigma_{BA}\bSigma_{AA}^{-1}\bSigma_{AB}.
\end{equation}
Using the same within-cluster IVP weights $\tilde\w_A,\tilde\w_B$ and the same cluster means $M_A,M_B$, set
\begin{equation}\label{eq:schur-ra-score}
V_{A\mid B}^{\mathrm{Sch}}:=\tilde\w_A^\top\bSigma_{A\mid B}\tilde\w_A,
\qquad
Q_{A\mid B}^{\mathrm{Sch}}:=\max\left\{\frac{M_A}{\sqrt{V_{A\mid B}^{\mathrm{Sch}}}},\varepsilon\right\},
\end{equation}
and analogously for $B\mid A$. The Schur-RA-HRP split is
\begin{equation}\label{eq:schur-ra-split}
\beta_A^{\mathrm{SchRA}}=
\frac{Q_{A\mid B}^{\mathrm{Sch}}}{Q_{A\mid B}^{\mathrm{Sch}}+Q_{B\mid A}^{\mathrm{Sch}}},
\qquad
\beta_B^{\mathrm{SchRA}}=1-\beta_A^{\mathrm{SchRA}}.
\end{equation}
\end{definition}

\begin{proposition}[Conditional-risk reduction and the Schur score wedge; \textbf{P}]\label{prop:schur-score-wedge}
If $\bSigma_C\succ0$, then
\begin{equation}\label{eq:schur-risk-reduction}
0<V_{A\mid B}^{\mathrm{Sch}}
=V_A-\tilde\w_A^\top\bSigma_{AB}\bSigma_{BB}^{-1}\bSigma_{BA}\tilde\w_A
\le V_A,
\end{equation}
with equality if and only if $\bSigma_{BA}\tilde\w_A=0$. Moreover,
\begin{equation}\label{eq:schur-risk-bound}
0\le V_A-V_{A\mid B}^{\mathrm{Sch}}
\le \frac{\|\bSigma_{BA}\tilde\w_A\|_2^2}{\lambda_{\min}(\bSigma_{BB})},
\end{equation}
and analogously for $B\mid A$. Hence, whenever the cluster mean is positive, the raw Schur score $M_A/\sqrt{V_{A\mid B}^{\mathrm{Sch}}}$ is weakly larger than the standalone RA-HRP score $M_A/\sqrt{V_A}$; after the floor, the score cannot decrease.
\end{proposition}

\begin{proof}
Since $\bSigma_C\succ0$, both principal blocks are positive definite and the Schur complements in \eqref{eq:schur-cond-blocks} are positive definite. Expanding the quadratic form gives the identity in \eqref{eq:schur-risk-reduction}; the subtracted term equals $\|\bSigma_{BB}^{-1/2}\bSigma_{BA}\tilde\w_A\|_2^2\ge0$, with equality exactly when $\bSigma_{BA}\tilde\w_A=0$. The operator inequality $\bSigma_{BB}^{-1}\preceq \lambda_{\min}(\bSigma_{BB})^{-1}I$ yields \eqref{eq:schur-risk-bound}. If $M_A>0$, reducing the denominator weakly raises the raw score; if $M_A\le0$, both raw scores are at or below the positive floor, so the floored score remains $\varepsilon$ unless the Schur adjustment pushes it above the floor.
\end{proof}

\begin{remark}[Relation to the exact Schur-tangency split]\label{rem:schur-ra-exactness}
Schur-RA-HRP repairs a second HRP substitution. Standard RA-HRP replaces inverse-variance splitting by return-aware cluster-Sharpe splitting, but it still uses standalone cluster variance. Definition~\ref{def:schur-rahrp} replaces that standalone risk by conditional risk. It therefore moves the rule toward the exact nodewise condition \eqref{eq:ra-node-exact}. It is still not a global Markowitz optimizer: it uses local block inversions only, keeps the tree fixed, and uses IVP aggregation inside clusters. Its purpose is to recover the economically important cross-cluster hedging term without paying the full global inversion tax.
\end{remark}

\begin{corollary}[Two-asset boundary; \textbf{P}]\label{cor:rahrp-twoasset}
Consider two assets, assume the floor is inactive and the relevant quantities are positive. Write $\rho$ for their correlation. RA-HRP coincides with tangency if and only if
\begin{equation}\label{eq:ra-twoasset}
\frac{\mu_{e,1}/\sigma_1}{\mu_{e,2}/\sigma_2}
=
\frac{\sigma_2^2\mu_{e,1}-\rho\sigma_1\sigma_2\mu_{e,2}}
{\sigma_1^2\mu_{e,2}-\rho\sigma_1\sigma_2\mu_{e,1}}.
\end{equation}
In particular, when $\rho=0$ and $\mu_{e,1},\mu_{e,2}>0$, RA-HRP is tangency if and only if $\sigma_1=\sigma_2$. Thus even for diagonal covariance, RA-HRP is generally not the plug-in Markowitz rule: with two uncorrelated assets it allocates by signal-to-volatility, whereas tangency allocates by signal-to-variance.
\end{corollary}

\begin{proof}
With two singleton children, the inactive-floor RA-HRP ratio is
\[
\frac{w^{\mathrm{RA}}_1}{w^{\mathrm{RA}}_2}
=\frac{\mu_{e,1}/\sigma_1}{\mu_{e,2}/\sigma_2}.
\]
The unnormalized tangency direction is
\[
\bSigma^{-1}\bmu_e \propto
\begin{pmatrix}
\sigma_2^2\mu_{e,1}-\rho\sigma_1\sigma_2\mu_{e,2}\\
\sigma_1^2\mu_{e,2}-\rho\sigma_1\sigma_2\mu_{e,1}
\end{pmatrix},
\]
which gives \eqref{eq:ra-twoasset}. Setting $\rho=0$ and cancelling the positive premia gives $\sigma_2/\sigma_1=\sigma_2^2/\sigma_1^2$, hence $\sigma_1=\sigma_2$.
\end{proof}

\begin{proposition}[Balanced equal-volatility diagonal exactness; \textbf{P}]\label{prop:rahrp-equalvol}
Suppose $\bSigma=\sigma^2 I$, every split in $\mathcal T$ is balanced, all cluster premium sums are positive, and the floor is inactive. Then RA-HRP coincides with the tangency portfolio:
\[
w_i^{\mathrm{RA}}=\frac{\mu_{e,i}}{\sum_j\mu_{e,j}}=w_{T,i}.
\]
\end{proposition}

\begin{proof}
For a cluster $K$, IV weights are equal weights, $V_K=\sigma^2/|K|$, and
\[
Q_K=\frac{|K|^{-1}\sum_{i\in K}\mu_{e,i}}{\sigma/\sqrt{|K|}}
=\frac{\sum_{i\in K}\mu_{e,i}}{\sigma\sqrt{|K|}}.
\]
At a balanced split $K=A\cup B$, $|A|=|B|$, so
\[
\beta_A^{\mathrm{RA}}=\frac{\sum_{i\in A}\mu_{e,i}}{\sum_{i\in A}\mu_{e,i}+\sum_{i\in B}\mu_{e,i}}.
\]
Multiplying these branch ratios along a root-to-leaf path telescopes to $\mu_{e,i}/\sum_j\mu_{e,j}$. Since $\bSigma=\sigma^2 I$, the tangency portfolio is proportional to $\bSigma^{-1}\bmu_e=\sigma^{-2}\bmu_e$, giving the result.
\end{proof}

\begin{theorem}[Pathwise HRP--RA-HRP distortion and KL attribution; \textbf{P}]\label{thm:rahrp-pathwise}
Fix a common tree $\mathcal T$. Let $r_v\in(0,1)$ be the HRP left-branch split at internal node $v$ and $s_v\in(0,1)$ the RA-HRP left-branch split. For a leaf $i$, let $\mathcal P(i)$ be its set of ancestor internal nodes and let $\eta_v(i)=1$ if $i$ descends through the left child of $v$ and $0$ otherwise. Then
\begin{equation}\label{eq:ra-hrp-ratio}
\frac{w_i^{\mathrm{RA}}}{w_i^{\mathrm{HRP}}}
=
\prod_{v\in\mathcal P(i)}
\left(\frac{s_v}{r_v}\right)^{\eta_v(i)}
\left(\frac{1-s_v}{1-r_v}\right)^{1-\eta_v(i)}.
\end{equation}
Moreover, viewing the two long-only portfolios as probability measures on leaves,
\begin{equation}\label{eq:ra-kl}
D_{\mathrm{KL}}(\w^{\mathrm{RA}}\Vert\w^{\mathrm{HRP}})
=
\sum_{v\in\mathcal V_{\mathrm{int}}} W_v^{\mathrm{RA}}
\left[s_v\log\frac{s_v}{r_v}+(1-s_v)\log\frac{1-s_v}{1-r_v}\right],
\end{equation}
where $W_v^{\mathrm{RA}}$ is the RA-HRP total capital entering node $v$.
\end{theorem}

\begin{proof}
Both portfolios are products of local branch fractions along a root-to-leaf path. Dividing the two products gives \eqref{eq:ra-hrp-ratio}. Taking logarithms, multiplying by $w_i^{\mathrm{RA}}$, summing over leaves, and exchanging the order of summation gives \eqref{eq:ra-kl}. At node $v$, the sum of RA-HRP weights in the left subtree is $W_v^{\mathrm{RA}}s_v$ and in the right subtree is $W_v^{\mathrm{RA}}(1-s_v)$, which yields the binary KL term.
\end{proof}

\begin{corollary}[Exact Sharpe-efficiency comparison; \textbf{P}]\label{cor:ra-efficiency}
Let $\eta_H=S(\w^{\mathrm{HRP}})/S_{\max}$ and $\eta_R=S(\w^{\mathrm{RA}})/S_{\max}$. The reduction in squared Sharpe loss obtained by moving from HRP to RA-HRP is exactly
\begin{align}\label{eq:ra-efficiency}
(1-\eta_H^2)-(1-\eta_R^2)
&=\eta_R^2-\eta_H^2 \\
&=\frac{S(\w^{\mathrm{RA}})^2-S(\w^{\mathrm{HRP}})^2}{S_{\max}^2} \\
&=\cos^2\theta_{\bSigma}(\w^{\mathrm{RA}},\w_T)
  -\cos^2\theta_{\bSigma}(\w^{\mathrm{HRP}},\w_T).
\end{align}
Therefore RA-HRP improves the exact-optimality gap precisely when its pathwise split distortion moves the final portfolio closer to the tangency ray in the $\bSigma$-angle. Equation~\eqref{eq:ra-hrp-ratio} attributes the source of that movement to individual tree nodes; Equation~\eqref{eq:ra-kl} measures how much capital was reallocated at each node.
\end{corollary}

\begin{proof}
This is Theorem~\ref{thm:cosine} applied once to $\w^{\mathrm{HRP}}$ and once to $\w^{\mathrm{RA}}$.
\end{proof}

\begin{remark}[Statistical status and empirical scope; \textbf{C}/\textbf{E}]\label{rem:rahrp-stat}
The fixed-tree RA-HRP paper proves existence, uniqueness, explicit leaf lower bounds, local split curvature, pathwise HRP--RA-HRP distortion, KL decomposition, and fixed-tree consistency/asymptotic normality under fixed dimension, finite fourth moments, and an inactive floor \citep{noguer2026rahrp}. It also emphasizes the theory--empirics gap created by re-estimating dendrograms in rolling windows: tree changes make the implemented map piecewise smooth rather than globally smooth. Empirically, the reported 15-stock US backtest is best read as proof of concept: RA-HRP improves traditional HRP in that sample, but a capped within-universe market-cap proxy attains an even higher Sharpe ratio. This is exactly the right evidentiary status for the present paper: RA-HRP is a disciplined return-aware extension of HRP, not a theorem of universal dominance.
\end{remark}

\begin{remark}[Open problem: data-driven trees as random topologies]\label{rem:datadriventrees}
The fixed-tree theory is the right first theorem because all split maps are smooth away from floor and variance boundaries. The implemented strategy, however, usually re-estimates the dendrogram. Theorem~\ref{thm:rahrp-pathwise} suggests a natural route to a full estimated-tree theory: compare two random trees by the split ratios induced on common root-to-leaf paths and by distances between their merge sets, cophenetic matrices, or tree-edit representations. One possible program is: (i) prove high-probability topology preservation under a merge-separation condition on the population distance matrix; (ii) on that event, apply the fixed-tree delta method; (iii) off that event, control the portfolio jump by pathwise products such as \eqref{eq:ra-hrp-ratio} and by a metric on dendrogram topologies; and (iv) combine the two pieces into a uniform-in-tree convergence bound. The empirical merge-preservation and cophenetic-correlation diagnostics reported in \citet{noguer2026rahrp} are therefore not merely descriptive; they are the measurable quantities a future random-tree asymptotic theorem must control.
\end{remark}

\subsection{First-order Sharpe calculus along the HRP--RA-HRP homotopy}\label{sec:homotopy-calculus}

The unit-free interpolation of Proposition~\ref{prop:unitfree-interp} turns ``how much local return information to admit'' into a one-dimensional decision. This subsection differentiates the Sharpe objective along that path. The derivative turns out to be a sum over internal nodes of the products of two estimable quantities: the displacement of the local split and the \emph{alpha}, against the HRP portfolio itself, of the self-financing switch between the node's two children. The marginal value of return information is therefore a nodewise zero-alpha computation --- the same object the GRS machinery of Remark~\ref{rem:grs} estimates.

\begin{lemma}[Node-switch vectors; \textbf{P}]\label{lem:node-switch}
Fix $\mathcal T$, strictly interior splits $r_v,s_v\in(0,1)$ at every internal node $v$, and the interpolated portfolio $\w^{(\lambda)}$ of \eqref{eq:unitfree-interp}. Let $W_v^{\mathrm H}:=\sum_{i:\,v\in\mathcal P(i)}w_i^{\mathrm{HRP}}$ be the HRP capital entering $v$. Define $g_v\in\R^N$ by
\begin{equation}\label{eq:node-switch}
g_{v,i}
:=
w_i^{\mathrm{HRP}}
\left(
\frac{\eta_v(i)}{r_v}-\frac{1-\eta_v(i)}{1-r_v}
\right)
\quad\textnormal{if } v\in\mathcal P(i),
\qquad
g_{v,i}:=0 \quad\textnormal{otherwise.}
\end{equation}
Then:
\begin{enumerate}
\item[(i)] $\ones^\top g_v=0$: each $g_v$ is a self-financing reallocation.
\item[(ii)] $g_v=W_v^{\mathrm H}\bigl(\mathbf u_{L(v)}-\mathbf u_{R(v)}\bigr)$, where $\mathbf u_K$ is the HRP composition inside subtree $K$ renormalized to unit mass: $g_v$ moves one unit of node-$v$ capital from the right-child HRP composition to the left-child HRP composition, levered by the capital actually present at the node.
\item[(iii)] $\lambda\mapsto\w^{(\lambda)}$ is a polynomial in $\lambda$, and
\begin{equation}\label{eq:homotopy-derivative}
\frac{d}{d\lambda}\w^{(\lambda)}\Big|_{\lambda=0}
=
\sum_{v\in\mathcal V_{\mathrm{int}}}(s_v-r_v)\,g_v .
\end{equation}
\end{enumerate}
\end{lemma}

\begin{proof}
Each leaf weight is the root-to-leaf product of branch fractions (the representation underlying Theorem~\ref{thm:rahrp-pathwise}), here of the affine functions $\beta_v^{(\lambda)}=(1-\lambda)r_v+\lambda s_v$ along its path, hence a polynomial in $\lambda$. Logarithmic differentiation at $\lambda=0$ gives
\[
\frac{d}{d\lambda}w_i^{(\lambda)}\Big|_{\lambda=0}
=
w_i^{\mathrm{HRP}}
\sum_{v\in\mathcal P(i)}(s_v-r_v)
\left(\frac{\eta_v(i)}{r_v}-\frac{1-\eta_v(i)}{1-r_v}\right),
\]
which is the $i$-th coordinate of \eqref{eq:homotopy-derivative}. For (i)--(ii): the HRP mass in the left subtree of $v$ is $W_v^{\mathrm H}r_v$ and in the right subtree $W_v^{\mathrm H}(1-r_v)$, so summing \eqref{eq:node-switch} over the left child gives $W_v^{\mathrm H}r_v/r_v=W_v^{\mathrm H}$ and over the right child $-W_v^{\mathrm H}$; the totals cancel, and dividing each child's restriction of $\w^{\mathrm{HRP}}$ by its subtree mass exhibits the two unit-mass compositions $\mathbf u_{L(v)},\mathbf u_{R(v)}$.
\end{proof}

\begin{theorem}[Nodewise-alpha representation of the homotopy derivative; \textbf{P}]\label{thm:homotopy-alpha}
Let $\bSigma\succ0$, $\bmu_e\neq0$, and write $a^*(\w)=\w^\top\bmu_e/\w^\top\bSigma\w$ and
\[
\bm\alpha(\w):=\bmu_e-a^*(\w)\,\bSigma\w
\]
for the implied-return residual of $\w$; by Remark~\ref{rem:beta}, $\alpha_i(\w)=\mu_{e,i}-\beta_i(\w)(\w^\top\bmu_e)$ is exactly the vector of alphas in \eqref{eq:zeroalpha}, and $\bm\alpha(\w)$ is the minimizing residual in the defect \eqref{eq:hpo-defect}. Then, with $\w^{\mathrm H}:=\w^{\mathrm{HRP}}$ and
\begin{equation}\label{eq:node-alpha}
\alpha_v:=g_v^\top\bm\alpha(\w^{\mathrm H})
=
g_v^\top\bmu_e-\frac{g_v^\top\bSigma\w^{\mathrm H}}{\w^{\mathrm H\top}\bSigma\w^{\mathrm H}}\,\bigl(\w^{\mathrm H\top}\bmu_e\bigr),
\end{equation}
the Sharpe ratio and the squared Sharpe efficiency satisfy
\begin{equation}\label{eq:homotopy-sharpe-derivative}
\begin{split}
\frac{d}{d\lambda}S\bigl(\w^{(\lambda)}\bigr)\Big|_{\lambda=0}
&=
\frac{1}{\sigma_p(\w^{\mathrm H})}
\sum_{v\in\mathcal V_{\mathrm{int}}}(s_v-r_v)\,\alpha_v,
\\
\frac{d}{d\lambda}\bigl(1-\eta^2\bigr)\bigl(\w^{(\lambda)}\bigr)\Big|_{\lambda=0}
&=
-\,\frac{2S(\w^{\mathrm H})}{S_{\max}^2\,\sigma_p(\w^{\mathrm H})}
\sum_{v\in\mathcal V_{\mathrm{int}}}(s_v-r_v)\,\alpha_v .
\end{split}
\end{equation}
The scalar $\alpha_v$ is precisely the alpha of the self-financing node-$v$ switch $g_v$ in the one-factor regression of Remark~\ref{rem:grs} with the HRP portfolio as the candidate factor.
\end{theorem}

\begin{proof}
For $\w\neq0$, $\nabla S(\w)=\bmu_e/\sigma_p(\w)-(\w^\top\bmu_e)\bSigma\w/\sigma_p(\w)^3=\bm\alpha(\w)/\sigma_p(\w)$. The chain rule with Lemma~\ref{lem:node-switch}(iii) gives the first display in \eqref{eq:homotopy-sharpe-derivative}; the second follows from $1-\eta^2=1-S^2/S_{\max}^2$ with $S_{\max}$ independent of $\lambda$. For the alpha interpretation: $g_v$ is zero-cost by Lemma~\ref{lem:node-switch}(i), its beta against $\w^{\mathrm H}$ is $\beta(g_v)=g_v^\top\bSigma\w^{\mathrm H}/\w^{\mathrm H\top}\bSigma\w^{\mathrm H}$, and its alpha in the sense of \eqref{eq:zeroalpha} is $g_v^\top\bmu_e-\beta(g_v)(\w^{\mathrm H\top}\bmu_e)$, which is \eqref{eq:node-alpha}. That this equals the minimizing residual of Definition~\ref{def:hpo-defect} contracted with $g_v$ is Theorem~\ref{thm:hpo-defect}: $a^*$ in \eqref{eq:hpo-a-star} is the same scalar.
\end{proof}

\begin{corollary}[Marginal value of return information; \textbf{P}]\label{cor:homotopy-sign}
To first order at $\lambda=0$, admitting local return information improves the Sharpe ratio if and only if
\begin{equation}\label{eq:homotopy-sign}
\sum_{v\in\mathcal V_{\mathrm{int}}}(s_v-r_v)\,\alpha_v\;>\;0 .
\end{equation}
Each node contributes the product of the split displacement $s_v-r_v$ (how far the floored cluster-Sharpe score moves the local budget away from the inverse-variance split) and the node alpha $\alpha_v$ (whether the implied child switch is priced by the HRP portfolio). In particular: (i) if every node switch has zero alpha against HRP --- the nodewise refinement of the zero-alpha null of Remark~\ref{rem:grs} restricted to the $|\mathcal V_{\mathrm{int}}|$ tradable directions $\{g_v\}$ --- the homotopy is first-order flat regardless of the score displacements; (ii) a node helps exactly when its score tilts toward the child whose switch carries positive alpha. Since $\alpha_v=\sum_i g_{v,i}\,\alpha_i(\w^{\mathrm H})$ is a fixed linear combination of the regression intercepts entering the GRS statistic \eqref{eq:grsformula}, the entire derivative \eqref{eq:homotopy-sharpe-derivative} is estimable from a single set of time-series regressions on the realized HRP return, with standard errors inherited from $\hat{\bm\alpha}$.
\end{corollary}

\begin{proof}
Sign reading of \eqref{eq:homotopy-sharpe-derivative}; (i)--(ii) are immediate, and the estimability statement is the linearity of \eqref{eq:node-alpha} in $\bm\alpha(\w^{\mathrm H})$.
\end{proof}

\begin{proposition}[KL trust budget along the homotopy; \textbf{P}]\label{prop:homotopy-kl}
For every internal node $v$, the map $\lambda\mapsto d_{\mathrm{KL}}\bigl(\beta_v^{(\lambda)}\Vert r_v\bigr)$ (binary KL divergence) is convex on $[0,1]$, vanishes together with its derivative at $\lambda=0$, and satisfies the chord bound
$d_{\mathrm{KL}}\bigl(\beta_v^{(\lambda)}\Vert r_v\bigr)\le\lambda\,d_{\mathrm{KL}}(s_v\Vert r_v)$.
Consequently, with $W_v^{(\lambda)}$ the interpolated capital entering $v$,
\begin{equation}\label{eq:homotopy-kl-budget}
\begin{split}
D_{\mathrm{KL}}\bigl(\w^{(\lambda)}\,\Vert\,\w^{\mathrm{HRP}}\bigr)
=
\sum_{v\in\mathcal V_{\mathrm{int}}}W_v^{(\lambda)}\,
d_{\mathrm{KL}}\bigl(\beta_v^{(\lambda)}\Vert r_v\bigr)
&\;\le\;
\lambda
\sum_{v\in\mathcal V_{\mathrm{int}}}W_v^{(\lambda)}\,d_{\mathrm{KL}}(s_v\Vert r_v)
\\
&\;\le\;
\lambda\,D_{\mathcal T}\max_{v}d_{\mathrm{KL}}(s_v\Vert r_v),
\end{split}
\end{equation}
since $\sum_v W_v^{(\lambda)}=\sum_i w_i^{(\lambda)}\,|\mathcal P(i)|\le D_{\mathcal T}$ is the capital-weighted tree depth.
\end{proposition}

\begin{proof}
The identity in \eqref{eq:homotopy-kl-budget} is Theorem~\ref{thm:rahrp-pathwise} with $s_v$ replaced by $\beta_v^{(\lambda)}$: the proof there uses only the product structure of tree weights, never the specific form of the splits. For fixed $r\in(0,1)$, $x\mapsto d_{\mathrm{KL}}(x\Vert r)=x\log(x/r)+(1-x)\log\bigl((1-x)/(1-r)\bigr)$ is convex with derivative $\log(x/r)-\log\bigl((1-x)/(1-r)\bigr)$, which vanishes at $x=r$; composing with the affine path $\beta_v^{(\lambda)}$ preserves convexity, and the value and derivative at $\lambda=0$ are zero, so the map is nondecreasing. Convexity with the endpoint values $0$ and $d_{\mathrm{KL}}(s_v\Vert r_v)$ gives the chord bound. Finally each unit of leaf capital is counted once per ancestor, so $\sum_v W_v^{(\lambda)}=\sum_i w_i^{(\lambda)}|\mathcal P(i)|$, bounded by the maximal root-to-leaf depth.
\end{proof}

The economic reading: $\lambda$ is a \emph{certified trust budget}. Choosing the interpolation level caps the capital reallocated away from the risk-only baseline linearly in $\lambda$, node by node, before any data are seen --- which is exactly the role the $\lambda_{\mathrm{KL}}$ penalty plays dynamically in the RLPO objective \eqref{eq:hpo-regularized-rlpo} below. Theorem~\ref{thm:homotopy-alpha} says where that budget earns alpha; Proposition~\ref{prop:homotopy-kl} says how much structural movement it can cause.


\section{The Mathematics of Heuristic Portfolio Optimization (HPO)}\label{sec:hpo}

The preceding sections compute exact coincidence conditions rule by rule. This section abstracts from the individual names and treats heuristic portfolio construction itself as a mathematical object. The key move is to separate three layers that are often conflated in practice: the full population optimizer, the information-restricted portfolio map, and the estimated implementation of that map.

\subsection{HPO maps as information-restricted portfolio rules}

\begin{definition}[Information-restricted portfolio map]\label{def:hpo-map}
Let
\[
\Theta := \{(\bmu_e,\bSigma): \bmu_e\in\R^N,\; \bSigma\in\mathbb{S}_{++}^N\}
\]
be the population parameter space. An \emph{information set} is a statistic
\[
\mathcal I:\Theta\to\mathcal Z,
\]
and a \emph{heuristic portfolio map} is a measurable function
\[
\Phi:\mathcal Z\to\Delta^{N-1},\qquad
\Delta^{N-1}:=\{\w\in\R^N:\ones^\top\w=1,\; w_i\ge 0\}.
\]
The induced portfolio at population point $\theta=(\bmu_e,\bSigma)$ is
\[
\w_\Phi(\theta):=\Phi(\mathcal I(\theta)).
\]
A parametric HPO family is a collection $\{\Phi_\lambda:\lambda\in\Lambda\}$, where $\lambda$ may represent a shrinkage intensity, a score floor, a tree depth, a turnover penalty, or the HRP--RA-HRP interpolation weight.
\end{definition}

Examples are immediate. Equal weight has a trivial information set. Inverse volatility uses only $\diag(\bSigma)$. ERC uses $\bSigma$ but not $\bmu_e$. HRP uses a dendrogram and recursively aggregated cluster variances. RA-HRP uses the same tree plus local cluster-premium and cluster-risk scores. Schur-RA-HRP enlarges the local information set by allowing sibling-block Schur complements. Thus the rule name is less important than the map
\[
(\bmu_e,\bSigma)\;\longmapsto\; \mathcal I(\bmu_e,\bSigma)
\;\longmapsto\; \Phi(\mathcal I(\bmu_e,\bSigma)).
\]
HPO studies the geometry and sampling behavior of this composition.

\begin{definition}[Population HPO projection]\label{def:hpo-projection}
For a family $\{\Phi_\lambda\}_{\lambda\in\Lambda}$ and a population point $\theta=(\bmu_e,\bSigma)$ with positive tangency portfolio, the population-optimal member of the heuristic family is
\begin{equation}\label{eq:hpo-proj}
\lambda^*(\theta)
\in
\argmax_{\lambda\in\Lambda} S\bigl(\w_{\Phi_\lambda}(\theta)\bigr)
=
\argmin_{\lambda\in\Lambda}
\theta_{\bSigma}\bigl(\w_{\Phi_\lambda}(\theta),\w_T(\theta)\bigr),
\end{equation}
where $\theta_{\bSigma}$ is the angle in the $\bSigma$-inner product from Theorem~\ref{thm:cosine}.
\end{definition}

Equation~\eqref{eq:hpo-proj} is the clean mathematical definition of HPO: it is not full optimization over all portfolios; it is angular projection of the tangency direction onto a lower-complexity image set generated by admissible heuristic maps. The image set may be finite (choose among EW, IV, ERC, HRP), one-dimensional (HRP--RA-HRP interpolation), or high-dimensional (a family of trees, floors, shrinkage operators, and cost penalties).

\subsection{The implied-return defect of a heuristic rule}

The implied-return principle gives a natural residual for every heuristic. If $\w$ were exactly optimal, $\bmu_e$ would lie on the ray spanned by $\bSigma\w$. The distance to that ray is the exact HPO defect.

\begin{definition}[HPO implied-return defect]\label{def:hpo-defect}
For a fully invested portfolio $\w$ with $\w^\top\bmu_e>0$, define
\begin{equation}\label{eq:hpo-defect}
\mathfrak d(\w;\bmu_e,\bSigma)
:=
\min_{a>0}
\left\|\bmu_e-a\bSigma\w\right\|_{\bSigma^{-1}}^2,
\qquad
\left\|x\right\|_{\bSigma^{-1}}^2:=x^\top\bSigma^{-1}x.
\end{equation}
The normalized defect is
\begin{equation}\label{eq:hpo-normalized-defect}
\mathfrak D(\w;\bmu_e,\bSigma)
:=
\frac{\mathfrak d(\w;\bmu_e,\bSigma)}{\bmu_e^\top\bSigma^{-1}\bmu_e}.
\end{equation}
For a rule $\Phi$, write $\mathfrak D_\Phi(\theta):=\mathfrak D(\w_\Phi(\theta);\theta)$.
\end{definition}

\begin{theorem}[Defect equals squared Sharpe inefficiency; \textbf{P}]\label{thm:hpo-defect}
Let $\w$ be fully invested with $\w^\top\bmu_e>0$. Then the minimizer in \eqref{eq:hpo-defect} is
\begin{equation}\label{eq:hpo-a-star}
a^*=
\frac{\w^\top\bmu_e}{\w^\top\bSigma\w},
\end{equation}
and
\begin{equation}\label{eq:hpo-defect-sharpe}
\mathfrak d(\w;\bmu_e,\bSigma)
=
S_{\max}^2-S(\w)^2,
\qquad
\mathfrak D(\w;\bmu_e,\bSigma)
=
1-\eta(\w)^2
=
\sin^2\theta_{\bSigma}(\w,\w_T).
\end{equation}
Thus the implied-return residual, the angular error, and squared Sharpe inefficiency are the same object in three coordinate systems.
\end{theorem}

\begin{proof}
Let $x=\bmu_e$ and $y=\bSigma\w$. In the $\bSigma^{-1}$-inner product,
\[
\langle x,y\rangle_{\bSigma^{-1}}=x^\top\bSigma^{-1}y=\bmu_e^\top\w=\w^\top\bmu_e,
\qquad
\|y\|_{\bSigma^{-1}}^2=\w^\top\bSigma\w.
\]
The least-squares projection of $x$ onto the positive ray spanned by $y$ has coefficient \eqref{eq:hpo-a-star}, which is positive by assumption. Therefore
\[
\mathfrak d
=\|x\|_{\bSigma^{-1}}^2-
\frac{\langle x,y\rangle_{\bSigma^{-1}}^2}{\|y\|_{\bSigma^{-1}}^2}
=
\bmu_e^\top\bSigma^{-1}\bmu_e-
\frac{(\w^\top\bmu_e)^2}{\w^\top\bSigma\w}
=S_{\max}^2-S(\w)^2.
\]
Dividing by $S_{\max}^2$ and applying Theorem~\ref{thm:cosine} gives \eqref{eq:hpo-defect-sharpe}.
\end{proof}

\begin{corollary}[Near-optimality tubes; \textbf{P}]\label{cor:hpo-tubes}
Fix $\bSigma\succ0$ and a portfolio $\w$ with $\w^\top\bSigma\w>0$. The set
\[
\mathcal T_\delta(\w;\bSigma)
:=
\{\bmu_e:\mathfrak D(\w;\bmu_e,\bSigma)\le\delta\}
\]
is exactly the cone around the ray $\{a\bSigma\w:a>0\}$ with aperture $\arcsin\sqrt\delta$ in the $\bSigma^{-1}$ geometry. Hence exact optimality is the zero-aperture cone, while approximate optimality is a positive-volume tube around the same ray. If a prior density for $\bmu_e$ is bounded on bounded $\bSigma^{-1}$-balls, then for small $\delta$ the probability of $\mathcal T_\delta$ inside any fixed ball is of order $\delta^{(N-1)/2}$.
\end{corollary}

\begin{proof}
By Theorem~\ref{thm:hpo-defect}, $\mathfrak D=\sin^2\angle_{\bSigma^{-1}}(\bmu_e,\bSigma\w)$, so the sublevel set is the stated cone. In $N$ dimensions a cone of angular radius $\alpha$ has solid angle proportional to $\alpha^{N-1}$ as $\alpha\downarrow0$. Since $\alpha=\arcsin\sqrt\delta\sim\sqrt\delta$, bounded-density probabilities inside fixed-radius balls scale as $\delta^{(N-1)/2}$.
\end{proof}

\begin{remark}[Why this defect is useful]\label{rem:hpo-defect-use}
The defect $\mathfrak D_\Phi$ is simultaneously a diagnostic, a test target, and an optimization objective. As a diagnostic, it reports exactly how far the heuristic's implied premia $a^*\bSigma\w_\Phi$ are from the estimated premia. As a test target, it is equivalent to the zero-alpha restriction in Remark~\ref{rem:grs}. As an HPO objective, it ranks heuristic families without pretending that every member is globally Markowitz-optimal.
\end{remark}

\subsection{Tree HPO: local scores as stochastic kernels}

HRP, RA-HRP, Schur-RA-HRP, and their interpolations are all examples of a common object: a tree-indexed product of local binary allocation kernels.

\begin{definition}[Positive score-tree HPO]\label{def:tree-score-hpo}
Fix a binary tree $\mathcal T$ with leaf set $\{1,\ldots,N\}$. At each internal node $v$, let $L(v)$ and $R(v)$ denote its left and right child clusters. A positive score-tree HPO rule assigns scores
\[
G_v(L(v))>0,
\qquad
G_v(R(v))>0,
\]
and local split
\begin{equation}\label{eq:hpo-tree-split}
\beta_v
:=
\frac{G_v(L(v))}{G_v(L(v))+G_v(R(v))}.
\end{equation}
For a leaf $i$, let $\mathcal P(i)$ be the internal nodes on the root-to-leaf path and let $\eta_v(i)=1$ if $i$ descends through the left branch of $v$ and $0$ otherwise. The resulting leaf weight is
\begin{equation}\label{eq:hpo-tree-weight}
w_i^\Phi
=
\prod_{v\in\mathcal P(i)}
\beta_v^{\eta_v(i)}(1-\beta_v)^{1-\eta_v(i)}.
\end{equation}
\end{definition}

The main rules correspond to different choices of $G_v$. HRP uses $G_v(K)=\widetilde V_K^{-1}$. RA-HRP uses $G_v(K)=S^+(K)$. Schur-RA-HRP uses conditional scores $G_v(K)=S^+_{\mathrm{Sch}}(K\mid K^c_v)$, where $K^c_v$ is the sibling cluster at node $v$. Unit-free HRP--RA-HRP interpolation is obtained by interpolating the local branch probabilities $\beta_v$, not by adding raw scores with incompatible units.

\begin{theorem}[Exact coincidence for score-tree HPO; \textbf{P}]\label{thm:tree-hpo-coincidence}
Fix $\theta=(\bmu_e,\bSigma)$ and suppose $x_T:=\bSigma^{-1}\bmu_e$ has strictly positive subtree masses
\[
M_T(K):=\ones_K^\top x_T>0
\]
for every cluster $K$ appearing in $\mathcal T$. Let $\Phi$ be a positive score-tree HPO rule on $\mathcal T$. Then $\w^\Phi$ coincides with the tangency composition $\w_T=x_T/(\ones^\top x_T)$ if and only if, at every internal node $v$,
\begin{equation}\label{eq:hpo-node-coincidence}
\frac{G_v(L(v))}{G_v(R(v))}
=
\frac{M_T(L(v))}{M_T(R(v))}.
\end{equation}
Equivalently, every heuristic local split must equal the exact tangency mass split:
\[
\beta_v=\frac{M_T(L(v))}{M_T(L(v))+M_T(R(v))}.
\]
\end{theorem}

\begin{proof}
If $\w^\Phi=\w_T$, then the total portfolio mass assigned to the left child of any node $v$ divided by the total mass assigned to its sibling equals $M_T(L(v))/M_T(R(v))$. But under the recursive rule, the same ratio is $\beta_v/(1-\beta_v)=G_v(L(v))/G_v(R(v))$, proving necessity. Conversely, if \eqref{eq:hpo-node-coincidence} holds at every node, then the heuristic sends exactly the same fraction of each parent subtree mass to each child as the tangency portfolio. Starting from root mass one, induction down the tree shows that every leaf receives the tangency mass, hence $\w^\Phi=\w_T$.
\end{proof}

\begin{remark}[RA-HRP as one score choice]\label{rem:hpo-rahrp}
Theorem~\ref{thm:tree-hpo-coincidence} is the abstract version of Theorem~\ref{thm:rahrp-coincidence}. RA-HRP chooses a particular economically interpretable score, the floored local cluster Sharpe ratio. Exact optimality requires that this score ratio reproduce the exact subtree mass ratio of $\bSigma^{-1}\bmu_e$. Schur-RA-HRP changes the score in order to move the local ratio closer to the Schur-adjusted tangency mass, while still avoiding the global inverse.
\end{remark}

\subsection{The exact Schur-tangency split}

The score-tree theorem identifies the target local mass ratio. The next theorem writes that target in Schur form. This is the return-aware analogue of the GMV Schur identity used earlier for HRP.

\begin{theorem}[Exact Schur-tangency bisection; \textbf{P}]\label{thm:hpo-schur-tangency}
Let $C=A\cup B$ and partition $\bSigma_C$ and $\bmu_C$ conformably:
\[
\bSigma_C=
\begin{pmatrix}
\bSigma_A & \bSigma_{AB}\\
\bSigma_{BA} & \bSigma_B
\end{pmatrix},
\qquad
\bmu_C=
\begin{pmatrix}
\bmu_A\\ \bmu_B
\end{pmatrix}.
\]
Define Schur complements and Schur-adjusted premia
\begin{align*}
\mathbf S_A&:=\bSigma_A-\bSigma_{AB}\bSigma_B^{-1}\bSigma_{BA},
&
\nu_{A\mid B}&:=\bmu_A-\bSigma_{AB}\bSigma_B^{-1}\bmu_B,\\
\mathbf S_B&:=\bSigma_B-\bSigma_{BA}\bSigma_A^{-1}\bSigma_{AB},
&
\nu_{B\mid A}&:=\bmu_B-\bSigma_{BA}\bSigma_A^{-1}\bmu_A.
\end{align*}
Then the unnormalized tangency vector $x_C=\bSigma_C^{-1}\bmu_C$ satisfies
\begin{equation}\label{eq:hpo-schur-tangency-vector}
(x_C)_A=\mathbf S_A^{-1}\nu_{A\mid B},
\qquad
(x_C)_B=\mathbf S_B^{-1}\nu_{B\mid A}.
\end{equation}
Consequently the exact tangency mass allocated to $A$ at this split is
\begin{equation}\label{eq:hpo-exact-tangency-split}
\beta_A^T
=
\frac{\ones_A^\top\mathbf S_A^{-1}\nu_{A\mid B}}
{\ones_A^\top\mathbf S_A^{-1}\nu_{A\mid B}
+\ones_B^\top\mathbf S_B^{-1}\nu_{B\mid A}},
\end{equation}
whenever the denominator is positive and both masses are positive.
\end{theorem}

\begin{proof}
Apply the block-inverse formula. The $A$ block of $\bSigma_C^{-1}\bmu_C$ is
\[
\mathbf S_A^{-1}\bmu_A
-
\mathbf S_A^{-1}\bSigma_{AB}\bSigma_B^{-1}\bmu_B
=
\mathbf S_A^{-1}\nu_{A\mid B}.
\]
The $B$ block is symmetric. Summing coordinates in each block gives the mass formula.
\end{proof}

\begin{corollary}[The HPO approximation ladder; \textbf{P}]\label{cor:hpo-ladder}
At a node $C=A\cup B$, exact local Markowitz/tangency allocation consumes four objects: the Schur covariance $\mathbf S_K$, the Schur-adjusted premium $\nu_{K\mid K^c}$, the inverse $\mathbf S_K^{-1}$, and the resulting subtree mass $\ones_K^\top\mathbf S_K^{-1}\nu_{K\mid K^c}$. The main hierarchical HPO rules are obtained by deleting or compressing these objects in stages:
\[
\begin{array}{rcl}
\textnormal{Exact Schur tangency} &:& G(K)=\ones_K^\top\mathbf S_K^{-1}\nu_{K\mid K^c},\\[0.2em]
\textnormal{Schur-RA-HRP} &:& G(K)=S^+_{\mathrm{Sch}}(K\mid K^c),\\[0.2em]
\textnormal{RA-HRP} &:& G(K)=S^+(K),\\[0.2em]
\textnormal{HRP} &:& G(K)=\widetilde V_K^{-1}.
\end{array}
\]
Thus the ``heuristic'' label describes a precise algebraic operation: replace exact Schur-tangency masses by stable positive scores whose estimation error is smaller and whose economic meaning is transparent.
\end{corollary}

\begin{proof}
The first line is Theorem~\ref{thm:hpo-schur-tangency}. The remaining lines are definitions of the corresponding score-tree rules. The claim is therefore an identity of information sets: each step removes either Schur conditioning, premium conditioning, local inversion, or premium information, and then renormalizes the remaining positive scores through \eqref{eq:hpo-tree-split}.
\end{proof}

\subsection{Sampling error, approximation error, and the HPO bias--variance identity}

HPO is useful only if the information it deletes is sufficiently noisy. The next result gives the exact quadratic accounting identity behind this statement.

\begin{theorem}[Exact mean--variance bias--variance decomposition for estimated HPO; \textbf{P}]\label{thm:hpo-riskdecomp}
Fix $\gamma>0$ and define the budget-constrained mean--variance objective
\[
U_\gamma(\w)=\w^\top\bmu_e-\frac{\gamma}{2}\w^\top\bSigma\w,
\qquad \ones^\top\w=1.
\]
Let $\w_\gamma^\star$ be its unique maximizer on the affine budget hyperplane, and let $\widehat\w_\Phi$ be any random fully invested portfolio produced by an estimated heuristic rule. If $\E\|\widehat\w_\Phi\|^2<\infty$, then
\begin{equation}\label{eq:hpo-biasvar}
\E\bigl[U_\gamma(\w_\gamma^\star)-U_\gamma(\widehat\w_\Phi)\bigr]
=
\frac{\gamma}{2}
\left(
\left\|\E[\widehat\w_\Phi]-\w_\gamma^\star\right\|_{\bSigma}^2
+
\operatorname{tr}\bigl(\bSigma\,\Cov(\widehat\w_\Phi)\bigr)
\right).
\end{equation}
\end{theorem}

\begin{proof}
The first-order condition for the budget-constrained optimum is
\[
\bmu_e-\gamma\bSigma\w_\gamma^\star=\lambda\ones
\]
for some Lagrange multiplier $\lambda$. For any fully invested $\w$, $\ones^\top(\w-\w_\gamma^\star)=0$, so the linear term in the second-order expansion vanishes:
\[
U_\gamma(\w_\gamma^\star)-U_\gamma(\w)
=
\frac{\gamma}{2}(\w-\w_\gamma^\star)^\top\bSigma(\w-\w_\gamma^\star).
\]
Apply this identity to $\widehat\w_\Phi$ and take expectations. The standard Hilbert-space bias--variance decomposition in the $\bSigma$ norm gives
\[
\E\|\widehat\w_\Phi-\w_\gamma^\star\|_{\bSigma}^2
=
\|\E[\widehat\w_\Phi]-\w_\gamma^\star\|_{\bSigma}^2
+\operatorname{tr}\bigl(\bSigma\,\Cov(\widehat\w_\Phi)\bigr).
\]
\end{proof}

\begin{remark}[When a heuristic beats the plug-in optimizer]\label{rem:hpo-beats-plugin}
Equation~\eqref{eq:hpo-biasvar} is the formal answer to the Markowitz optimization enigma. A heuristic can be farther from the population optimum in bias but closer in expected realized utility if it has sufficiently lower estimation variance. HRP, RA-HRP, inverse volatility, ERC, and constrained/shrunk optimizers all exploit this identity in different ways. HRP and RA-HRP reduce variance by replacing global inversion with local score maps; Schur-RA-HRP spends some variance budget on local block inversions to recover cross-cluster hedge information; full MVO spends the most variance budget by estimating and inverting the complete covariance matrix and using the full mean vector.
\end{remark}

\subsection{The master HPO objective}

The preceding results suggest a practical HPO loss for selecting among heuristic families:
\begin{equation}\label{eq:hpo-master}
\mathcal L_{\mathrm{HPO}}(\Phi)
=
\widehat{\mathfrak D}_\Phi
+\lambda_{\mathrm{est}}\,\widehat{\operatorname{tr}\bigl(\bSigma\,\Cov(\widehat\w_\Phi)\bigr)}
+\lambda_{\mathrm{tc}}\,\|\widehat\w_\Phi-\w_-\|_1
+\lambda_{\mathrm{top}}\,d_{\mathcal T}(\widehat{\mathcal T},\mathcal T_-),
\end{equation}
where $\widehat{\mathfrak D}_\Phi$ measures implied-return misspecification, the second term measures weight estimation variance, the third measures turnover, and the fourth measures tree instability when the rule is hierarchical. Equation~\eqref{eq:hpo-master} is not a new theorem; it is the operational synthesis of the theorems above. It turns HPO into a disciplined model-selection problem rather than a collection of unrelated allocation recipes.

\begin{remark}[How this section changes the interpretation of the paper]\label{rem:hpo-interpretation}
The earlier sections answer, for each named rule, ``when is it exactly optimal?'' This section answers the broader question, ``what is a heuristic portfolio rule mathematically?'' It is an information-restricted map; its population error is an implied-return projection residual; its Sharpe loss is the squared sine of an angle; its hierarchical forms are products of local binary kernels; its exact nodewise target is a Schur-tangency mass; and its out-of-sample value is governed by a bias--variance--cost trade-off. That is the mathematics of HPO.
\end{remark}

\section{From HPO to Reinforcement Learning Portfolio Optimization (RLPO)}\label{sec:rlpo}

The preceding section is static: given a statewise estimate $(\widehat\bmu_{e,t},\widehat\bSigma_t)$, possibly together with a tree $\widehat{\mathcal T}_t$, a heuristic map produces a portfolio $\w_t=\Phi(s_t)$. RLPO starts from the same action object but changes the question. It asks how a portfolio agent should rebalance through time when today's action affects future states, trading costs, inventory, drawdown, tax lots, liquidity, and regime exposure. Thus HPO is the \emph{myopic policy face} of RLPO; RLPO is the dynamic control layer that learns when the continuation value of deviating from the HPO prior compensates for its instantaneous optimality defect.

\subsection{Controlled market formulation}

\begin{definition}[RLPO environment]\label{def:rlpo-env}
An RLPO environment is a discounted controlled market
\[
\mathfrak M=(\mathcal S,\{\mathcal A(s)\}_{s\in\mathcal S},P,R,\gamma),
\qquad 0\le \gamma<1,
\]
where $s_t\in\mathcal S$ is the observed market state, $\mathcal A(s_t)\subseteq\Delta^{N-1}$ is the feasible portfolio set, $P(ds'\mid s,w)$ is the controlled transition law, $R(s,w,s')$ is the one-period reward, and $\gamma$ is the discount factor. A Markov policy $\pi$ assigns a distribution $\pi(\cdot\mid s)$ on $\mathcal A(s)$. Its value is
\begin{equation}\label{eq:rlpo-value}
V^\pi(s)=\E^\pi\left[\sum_{t\ge0}\gamma^t R(s_t,w_t,s_{t+1})\mid s_0=s\right].
\end{equation}
Writing $r(s,w):=\E[R(s,w,s')\mid s,w]$, the optimal value satisfies the Bellman equation
\begin{equation}\label{eq:rlpo-bellman}
V^\star(s)=\sup_{w\in\mathcal A(s)}
\left\{r(s,w)+\gamma\E\left[V^\star(s')\mid s,w\right]\right\}.
\end{equation}
\end{definition}

This is the usual dynamic-programming formulation of sequential control \citep{bellman1957,bertsekas2017,suttonbarto2018}. In portfolio applications the state may include lagged returns, macro variables, volatility estimates, factor exposures, the current dendrogram, the previous portfolio $\w_-$, turnover budget, margin constraints, and execution-cost estimates. The reward may be raw excess return, mean--variance utility, downside-risk-adjusted return, a drawdown-penalized objective, or a transaction-cost-adjusted utility as in dynamic trading models with predictable returns and frictions \citep{garleanuPedersen2013}.

\begin{definition}[HPO-induced policy]\label{def:hpo-policy}
Let $G:\mathcal S\to\mathcal Z$ be a state summarizer producing the information consumed by an HPO rule, for example
\[
G(s_t)=(\widehat\bmu_{e,t},\widehat\bSigma_t),
\qquad
G(s_t)=(\widehat\bmu_{e,t},\widehat\bSigma_t,\widehat{\mathcal T}_t),
\]
or the reduced statistics used by EW, IV, ERC, HRP, RA-HRP, or Schur-RA-HRP. Any feasible HPO map $\Phi:\mathcal Z\to\Delta^{N-1}$ induces the deterministic stationary policy
\begin{equation}\label{eq:hpo-induced-policy}
\pi_\Phi(\cdot\mid s)=\delta_{\Phi(G(s))},
\qquad
\Phi(G(s))\in\mathcal A(s).
\end{equation}
\end{definition}

\begin{proposition}[HPO policies as stationary RLPO policies; \textbf{P}]\label{prop:hpo-stationary-rlpo}
Every feasible measurable HPO map induces an admissible deterministic stationary RLPO policy through \eqref{eq:hpo-induced-policy}. If $\gamma=0$, the reward has no next-state dependence, and $R(s,w,s')=u(G(s),w)$ equals the one-period static objective used to define the HPO map, then the RLPO problem reduces state by state to the corresponding static HPO problem. In particular, if
\[
\Phi(G(s))\in \argmax_{w\in\mathcal A(s)} u(G(s),w),
\]
then $\pi_\Phi$ is optimal for the RLPO environment with $\gamma=0$.
\end{proposition}

\begin{proof}
Measurability and feasibility imply that $s\mapsto\delta_{\Phi(G(s))}$ is a valid stationary Markov policy. If $\gamma=0$, then \eqref{eq:rlpo-value} reduces to the one-period conditional objective $V^\pi(s)=\E[R(s,w,s')\mid s,w]$. Under the stated reward specification this is $u(G(s),w)$, so the Bellman maximization is exactly $\max_{w\in\mathcal A(s)}u(G(s),w)$. A statewise optimizer $\Phi(G(s))$ therefore induces an optimal deterministic policy.
\end{proof}

\begin{remark}[The conceptual hierarchy]\label{rem:hpo-rlpo-hierarchy}
HPO answers the question ``given the current estimated pair $(\widehat\bmu_t,\widehat\bSigma_t)$ and any reduced structure such as a tree, which portfolio map is locally justified?'' RLPO answers the dynamic question ``given the current state and the effect of today's action on tomorrow's state, when should the agent deviate from that locally justified map?'' The action space is the same simplex. The difference is the objective: HPO optimizes an instantaneous geometric criterion, while RLPO optimizes the Bellman criterion \eqref{eq:rlpo-bellman}.
\end{remark}

\subsection{Continuation value versus myopic HPO defect}

At state $s$, write
\[
\widehat\bmu_e(s),\qquad \widehat\bSigma(s),\qquad \w_\Phi(s):=\Phi(G(s))
\]
for the estimated excess-premium vector, covariance matrix, and HPO baseline. The statewise implied-return defect is
\begin{equation}\label{eq:statewise-hpo-defect}
\mathfrak d_s(w)
:=
\min_{a>0}
\left\|\widehat\bmu_e(s)-a\widehat\bSigma(s)w\right\|_{\widehat\bSigma(s)^{-1}}^2.
\end{equation}
By Theorem~\ref{thm:hpo-defect}, this is not an ad hoc regularizer: it is exactly
\[
\mathfrak d_s(w)=S_{\max}^2(s)-S_s^2(w),
\]
the statewise squared Sharpe inefficiency measured in the metric of the estimated covariance matrix.

\begin{proposition}[Dynamic deviation principle; \textbf{P}]\label{prop:dynamic-deviation-principle}
Fix a state $s$, a baseline HPO action $w_\Phi=\w_\Phi(s)$, and a candidate action $w\in\mathcal A(s)$. Decompose the one-period reward as
\[
r(s,w)=r_0(s,w)-c_s(w,w_-)-\lambda_{\mathrm{HPO}}\mathfrak d_s(w),
\]
where $r_0$ is the trading reward before HPO regularization, $c_s$ is an implementation cost relative to the previous portfolio $w_-$, and $\lambda_{\mathrm{HPO}}\ge0$. Define the instantaneous loss of choosing $w$ rather than the HPO action by
\begin{equation}\label{eq:inst-loss}
\begin{split}
\Delta_{\mathrm{inst}}(s,w\mid w_\Phi)
:={}& r_0(s,w_\Phi)-r_0(s,w)
+c_s(w,w_-)-c_s(w_\Phi,w_-) \\
&+\lambda_{\mathrm{HPO}}\{\mathfrak d_s(w)-\mathfrak d_s(w_\Phi)\},
\end{split}
\end{equation}
and the continuation-value gain by
\begin{equation}\label{eq:continuation-gain}
\Delta_{\mathrm{cont}}(s,w\mid w_\Phi)
:=
\gamma\E[V^\star(s')\mid s,w]
-
\gamma\E[V^\star(s')\mid s,w_\Phi].
\end{equation}
Then $w$ weakly improves on the HPO action in Bellman value if and only if
\begin{equation}\label{eq:rlpo-justification}
\Delta_{\mathrm{cont}}(s,w\mid w_\Phi)
\ge
\Delta_{\mathrm{inst}}(s,w\mid w_\Phi).
\end{equation}
\end{proposition}

\begin{proof}
Subtract the Bellman score of $w_\Phi$ from that of $w$:
\[
\{r(s,w)+\gamma\E[V^\star(s')\mid s,w]\}
-
\{r(s,w_\Phi)+\gamma\E[V^\star(s')\mid s,w_\Phi]\}.
\]
Substituting the reward decomposition and rearranging gives
$\Delta_{\mathrm{cont}}(s,w\mid w_\Phi)-\Delta_{\mathrm{inst}}(s,w\mid w_\Phi)$. Nonnegativity is exactly \eqref{eq:rlpo-justification}.
\end{proof}

\begin{remark}[Economic interpretation]\label{rem:rlpo-econ}
Equation~\eqref{eq:rlpo-justification} is the precise mathematical statement of the HPO--RLPO relation. A dynamic policy should not deviate from a locally justified HPO portfolio merely because a neural policy can output different weights. It should deviate only when the expected continuation-value gain caused by that deviation is large enough to pay for three quantities: the lost instantaneous reward, the incremental turnover or implementation cost, and the additional implied-return defect. RLPO is therefore not a replacement for HPO. It is a disciplined dynamic policy-improvement layer over an HPO prior.
\end{remark}

\subsection{The exact value gap of an HPO policy}\label{sec:value-gap}

Proposition~\ref{prop:dynamic-deviation-principle} is a one-step comparison. The classical performance-difference identity upgrades it to an exact global accounting of how much value the myopic HPO policy leaves on the table, in the same advantage units, summed along the trajectory of whatever policy challenges it.

\begin{proposition}[Exact value gap of the HPO policy; \textbf{P}, identity after \citealp{kakadelangford2002}]\label{prop:hpo-value-gap}
Assume $|R(s,w,s')|\le\bar R<\infty$ and $\gamma\in[0,1)$, so all values are bounded by $\bar R/(1-\gamma)$. Define the \emph{HPO advantage} of action $w$ at state $s$ by
\begin{equation}\label{eq:hpo-advantage}
A^{\pi_\Phi}(s,w)
:=
r(s,w)+\gamma\E\bigl[V^{\pi_\Phi}(s')\mid s,w\bigr]-V^{\pi_\Phi}(s),
\end{equation}
which satisfies $A^{\pi_\Phi}\bigl(s,\w_\Phi(s)\bigr)=0$ by Bellman consistency of the deterministic stationary policy $\pi_\Phi$. Then for every Markov policy $\pi$ and initial state $s_0$,
\begin{equation}\label{eq:value-gap}
V^{\pi}(s_0)-V^{\pi_\Phi}(s_0)
=
\E^{\pi}\left[\sum_{t\ge0}\gamma^t A^{\pi_\Phi}(s_t,w_t)\,\Big|\,s_0\right]
=
\frac{1}{1-\gamma}\,
\E_{(s,w)\sim d^{\pi}_{s_0}}\bigl[A^{\pi_\Phi}(s,w)\bigr],
\end{equation}
where $d^{\pi}_{s_0}:=(1-\gamma)\sum_{t\ge0}\gamma^t\,\mathrm{Law}^{\pi}\bigl((s_t,w_t)\mid s_0\bigr)$ is the normalized discounted occupancy measure. In particular, taking $\pi=\pi^\star$,
\begin{equation}\label{eq:myopia-price}
V^{\star}(s_0)-V^{\pi_\Phi}(s_0)
=
\frac{1}{1-\gamma}\,
\E_{(s,w)\sim d^{\pi^\star}_{s_0}}\bigl[A^{\pi_\Phi}(s,w)\bigr]:
\end{equation}
the exact price of acting myopically is the occupancy-weighted HPO advantage, and only the states where the optimal policy actually deviates from the heuristic contribute.
\end{proposition}

\begin{proof}
Substituting \eqref{eq:hpo-advantage} and using the tower property under $\pi$,
\[
\E^{\pi}\left[\sum_{t=0}^{T}\gamma^t A^{\pi_\Phi}(s_t,w_t)\right]
=
\E^{\pi}\left[\sum_{t=0}^{T}\gamma^t R_t\right]
+
\E^{\pi}\left[\sum_{t=0}^{T}\Bigl(\gamma^{t+1}V^{\pi_\Phi}(s_{t+1})-\gamma^{t}V^{\pi_\Phi}(s_{t})\Bigr)\right].
\]
The second sum telescopes to
\[
\gamma^{T+1}\E^{\pi}\!\left[V^{\pi_\Phi}(s_{T+1})\right]
- V^{\pi_\Phi}(s_0).
\]
Boundedness sends the discounted terminal term to zero as $T\to\infty$, leaving
$V^{\pi}(s_0)-V^{\pi_\Phi}(s_0)$. Dominated convergence justifies the limit, and the occupancy form is the bounded-Fubini rewriting of the discounted sum. Bellman consistency,
\[
V^{\pi_\Phi}(s)=r(s,\w_\Phi(s))
+\gamma\E\!\left[V^{\pi_\Phi}(s')\mid s,\w_\Phi(s)\right],
\]
gives $A^{\pi_\Phi}(s,\w_\Phi(s))=0$.
\end{proof}

\begin{corollary}[$\varepsilon$-myopic optimality and shaped localization; \textbf{P}]\label{cor:eps-myopic}
\begin{enumerate}
\item[(i)] If no feasible one-step deviation gains more than $\varepsilon\ge0$ against the HPO policy's own continuation value, i.e.
\[
A^{\pi_\Phi}(s,w)\;\le\;\varepsilon
\qquad\textnormal{for every $s$ and every } w\in\mathcal A(s),
\]
then
\[
V^{\star}(s)\;\le\;V^{\pi_\Phi}(s)+\frac{\varepsilon}{1-\gamma}
\qquad\textnormal{for every } s.
\]
\item[(ii)] Under the reward decomposition of Proposition~\ref{prop:dynamic-deviation-principle}, the advantage splits exactly as
\begin{align*}
A^{\pi_\Phi}(s,w)
&=
\widetilde\Delta_{\mathrm{cont}}(s,w\mid \w_\Phi)
-
\Delta_{\mathrm{inst}}(s,w\mid \w_\Phi),
\\
\widetilde\Delta_{\mathrm{cont}}(s,w\mid \w_\Phi)
&:=
\gamma\E\bigl[V^{\pi_\Phi}(s')\mid s,w\bigr]
-\gamma\E\bigl[V^{\pi_\Phi}(s')\mid s,\w_\Phi(s)\bigr],
\end{align*}
with the same instantaneous-loss term \eqref{eq:inst-loss} and the continuation gain evaluated under $V^{\pi_\Phi}$. Hence \eqref{eq:myopia-price} is the occupancy-weighted integral of exactly the quantity the deviation principle tests pointwise; and because the $\lambda_{\mathrm{HPO}}$ term sits inside $\Delta_{\mathrm{inst}}$, every deviation along the optimal trajectory pays its incremental implied-return defect $\mathfrak d_s(w)-\mathfrak d_s(\w_\Phi(s))$ state by state.
\end{enumerate}
\end{corollary}

\begin{proof}
(i) Bound the integrand in \eqref{eq:myopia-price} by $\varepsilon$. (ii) Subtract and add $r(s,\w_\Phi(s))$ in \eqref{eq:hpo-advantage}, use Bellman consistency of $\pi_\Phi$ to cancel $V^{\pi_\Phi}(s)$, and substitute the reward decomposition; the algebra is that of Proposition~\ref{prop:dynamic-deviation-principle} with $V^{\pi_\Phi}$ in place of $V^\star$.
\end{proof}

This closes the logical loop between the two layers. Statically, ``never optimal, rarely costly'' is Proposition~\ref{prop:generic} plus Theorem~\ref{thm:secondorder}. Dynamically, the analogue is \eqref{eq:myopia-price} plus Corollary~\ref{cor:eps-myopic}(i): the HPO policy forfeits at most the discounted occupancy average of its statewise advantage gaps, and when continuation values never reward deviation by more than $\varepsilon$, acting myopically costs at most $\varepsilon/(1-\gamma)$ --- a bound expressed entirely in the advantage units that the deviation principle already prices.

\subsection{HPO-regularized RLPO objectives}

The preceding proposition suggests the following regularized RLPO objective for a parametric policy $\pi_\theta$:
\begin{equation}\label{eq:hpo-regularized-rlpo}
\begin{split}
J_{\mathrm{RLPO}}(\theta)
:=
\E^{\pi_\theta}\sum_{t\ge0}\gamma^t
\Big[
&w_t^\top r_{e,t+1}
-c_t(w_t,w_{t-1})
-\lambda_{\mathrm{risk}}\rho_t(w_t) \\
&-\lambda_{\mathrm{HPO}}\mathfrak d(w_t;\widehat\bmu_{e,t},\widehat\bSigma_t)
-\lambda_{\mathrm{KL}}D_{\mathrm{KL}}(w_t\,\|\,\w_\Phi(s_t))
\Big].
\end{split}
\end{equation}
Here $\rho_t$ is a risk functional, $c_t$ is the trading-cost model, $\mathfrak d$ is the implied-return defect from Definition~\ref{def:hpo-defect}, and the KL term keeps the learned policy near the HPO prior unless the data supply enough dynamic evidence to move away from it. If $\lambda_{\mathrm{HPO}}$ and $\lambda_{\mathrm{KL}}$ are large, the solution collapses toward the HPO baseline. If they are zero, the problem becomes an unconstrained RLPO objective. Intermediate values implement dynamic residual learning around a mathematically interpretable prior.

A particularly stable actor parametrization is the convex-residual policy
\begin{equation}\label{eq:convex-residual-policy}
w_t^\theta
=
(1-\lambda_t^\theta)\w_\Phi(s_t)
+
\lambda_t^\theta\psi_\theta(s_t),
\qquad
\lambda_t^\theta\in[0,1],
\quad
\psi_\theta(s_t)\in\Delta^{N-1}.
\end{equation}
This guarantees long-only full investment by construction. The scalar $\lambda_t^\theta$ has a direct interpretation: the state-dependent degree of trust placed on dynamic RLPO relative to static HPO. A conservative agent learns mostly $\lambda_t^\theta$; a more expressive agent learns both $\lambda_t^\theta$ and the residual portfolio $\psi_\theta$.

\subsection{RA-HRP as a hierarchical policy prior}

RA-HRP is especially natural inside RLPO because the policy is already factorized over a tree. Instead of asking an actor network to output $N$ unrelated weights, one can ask it to output local split probabilities. Fix a tree $\mathcal T$. For each internal node $v$, let $\beta_v^{\mathrm{RA}}(s)$ denote the RA-HRP split fraction at state $s$, and let $a_v^\theta(s)$ be a learned dynamic residual. Define the actor by
\begin{equation}\label{eq:rahrp-logit-actor}
\operatorname{logit}\beta_v^\theta(s)
=
\operatorname{logit}\beta_v^{\mathrm{RA}}(s)
+
a_v^\theta(s),
\qquad
\beta_v^\theta(s)\in(0,1).
\end{equation}
For leaf $i$, with path $\mathcal P(i)$ and branch indicator $\eta_v(i)$, the induced portfolio is
\begin{equation}\label{eq:hierarchical-actor-weight}
w_i^\theta(s)
=
\prod_{v\in\mathcal P(i)}
\left(\beta_v^\theta(s)\right)^{\eta_v(i)}
\left(1-\beta_v^\theta(s)\right)^{1-\eta_v(i)}.
\end{equation}

\begin{proposition}[Budget feasibility and baseline recovery of the RA-HRP actor; \textbf{P}]\label{prop:rahrp-actor-feasible}
For any collection of split probabilities $\{\beta_v^\theta(s)\in(0,1):v\in\mathcal V_{\mathrm{int}}\}$, the weights in \eqref{eq:hierarchical-actor-weight} are strictly positive and sum to one. If $a_v^\theta(s)=0$ for every internal node, then $w^\theta(s)=w^{\mathrm{RA}}(s)$ exactly. Moreover, for two hierarchical policies $A$ and $B$ on the same tree,
\begin{equation}\label{eq:rlpo-tree-kl}
D_{\mathrm{KL}}(w^A(s)\,\|\,w^B(s))
=
\sum_{v\in\mathcal V_{\mathrm{int}}}
W_v^A(s)
\,d_{\mathrm{KL}}\bigl(\beta_v^A(s)\,\|\,\beta_v^B(s)\bigr),
\end{equation}
where $W_v^A(s)$ is the capital mass reaching node $v$ under policy $A$ and $d_{\mathrm{KL}}$ is the binary KL divergence.
\end{proposition}

\begin{proof}
Budget feasibility follows by induction down the tree: the root has mass one, and every internal node sends fractions $\beta_v^\theta$ and $1-\beta_v^\theta$ to its two children, so total mass is conserved at each split. Positivity follows from $\beta_v^\theta\in(0,1)$. If all residuals vanish in \eqref{eq:rahrp-logit-actor}, every local split equals the RA-HRP split, hence the product representation recovers RA-HRP. The KL identity is the same pathwise decomposition proved in Theorem~\ref{thm:rahrp-pathwise}: expand the log ratio of leaf weights as a sum of log ratios of local branch probabilities and exchange the sums over leaves and internal nodes.
\end{proof}

Equation~\eqref{eq:rlpo-tree-kl} is the main reason RA-HRP is a better RLPO prior than a black-box weight vector. The global distance between the learned dynamic policy and the baseline decomposes into node-level binary divergences. Root nodes receive large capital weights and therefore dominate global deviation; deep nodes allow fine local tilts. This gives a precise attribution of the actor's dynamic behavior: one can identify whether RLPO is overriding the top-level sector allocation, a mid-level cluster split, or a leaf-level security choice.

\begin{corollary}[Myopic policy gradient at the hierarchical prior; \textbf{P}]\label{cor:actor-gradient-alpha}
Fix a state $s$ and evaluate the myopic Sharpe objective $S_s(w):=w^\top\widehat\bmu_e(s)/\bigl(w^\top\widehat\bSigma(s)\,w\bigr)^{1/2}$ on the hierarchical actor \eqref{eq:rahrp-logit-actor} at the prior, $a_v^\theta(s)=0$ for all $v$, so that $w^\theta(s)=w^{\mathrm{RA}}(s)$. Then for every internal node $v$,
\begin{equation}\label{eq:actor-gradient}
\frac{\partial S_s\bigl(w^\theta(s)\bigr)}{\partial a_v^\theta}\bigg|_{a^\theta=0}
=
\frac{\beta_v^{\mathrm{RA}}(s)\bigl(1-\beta_v^{\mathrm{RA}}(s)\bigr)}{\sigma_p\bigl(w^{\mathrm{RA}}(s)\bigr)}\;
\alpha_v^{\mathrm{RA}}(s),
\end{equation}
where $\alpha_v^{\mathrm{RA}}(s)=g_v^{\mathrm{RA}}(s)^\top\bigl[\widehat\bmu_e(s)-a^*\bigl(w^{\mathrm{RA}}(s)\bigr)\widehat\bSigma(s)\,w^{\mathrm{RA}}(s)\bigr]$ is the node alpha of Theorem~\ref{thm:homotopy-alpha} computed with the RA-HRP weights and splits in the node-switch vectors of Lemma~\ref{lem:node-switch}. Hence the myopic component of the policy gradient at the prior is, coordinate by coordinate, the nodewise alpha left unexplained by the prior, damped by the logistic variance $\beta_v(1-\beta_v)$: the actor receives its first learning signal exactly where the baseline misprices a child switch, and that signal is automatically attenuated at saturated splits.
\end{corollary}

\begin{proof}
At $a^\theta=0$, $\partial\beta_v^\theta/\partial a_v^\theta=\beta_v^{\mathrm{RA}}(1-\beta_v^{\mathrm{RA}})$ from the logit parametrization, so by the product representation \eqref{eq:hierarchical-actor-weight},
$\partial w_i^\theta/\partial a_v^\theta
=\beta_v^{\mathrm{RA}}(1-\beta_v^{\mathrm{RA}})\,
w_i^{\mathrm{RA}}\bigl(\eta_v(i)/\beta_v^{\mathrm{RA}}-(1-\eta_v(i))/(1-\beta_v^{\mathrm{RA}})\bigr)$,
which is $\beta_v^{\mathrm{RA}}(1-\beta_v^{\mathrm{RA}})$ times the node-switch vector $g_v^{\mathrm{RA}}$ of Lemma~\ref{lem:node-switch} built from the RA weights and splits. Contract with $\nabla S_s(w)=\bigl[\widehat\bmu_e(s)-a^*(w)\widehat\bSigma(s)w\bigr]/\sigma_p(w)$ as in Theorem~\ref{thm:homotopy-alpha}.
\end{proof}

A Schur-RA-HRP prior is obtained by replacing $\beta_v^{\mathrm{RA}}(s)$ in \eqref{eq:rahrp-logit-actor} by the conditional-risk split from Definition~\ref{def:schur-rahrp}. The actor then learns deviations around a baseline that already incorporates local cross-cluster hedging while still avoiding a global inverse. This gives the natural approximation ladder
\begin{equation}\label{eq:rlpo-ladder}
\begin{array}{c}
\text{exact Bellman policy}\\[0.2em]
\downarrow\\[-0.1em]
\text{RL residual over Schur-RA-HRP}\\[0.2em]
\downarrow\\[-0.1em]
\text{RL residual over RA-HRP}\\[0.2em]
\downarrow\\[-0.1em]
\text{RL residual over HRP}.
\end{array}
\end{equation}

\subsection{Tree topology, regime shifts, and the dynamic open problem}

The fixed-tree theory in Sections~\ref{sec:hrp} and~\ref{sec:hpo} is deliberately sharp because it isolates the local score map from the discrete topology. RLPO makes the open problem unavoidable: the state may include a time-varying tree $\mathcal T_t$, and the policy must decide whether changes in the estimated dendrogram are signal, noise, or regime transition. A dynamic tree-regularized objective can be written as
\begin{equation}\label{eq:dynamic-tree-regularizer}
J_{\mathrm{tree}}(\theta)
=
J_{\mathrm{RLPO}}(\theta)
-
\lambda_{\mathcal T}
\E^{\pi_\theta}\sum_{t\ge0}\gamma^t
 d_{\mathcal T}(\mathcal T_t,\mathcal T_{t-1}),
\end{equation}
where $d_{\mathcal T}$ may be a merge-preservation distance, a cophenetic-matrix distance, a leaf-order rank distance, or a split-probability path distance. The pathwise factorization of hierarchical weights suggests the right mathematical route for future theory: prove uniform convergence of the value of node-split policies over a class of random trees, controlling both continuous score perturbations and discrete topology changes. In this sense the dendrogram-stability diagnostics used for RA-HRP are not merely empirical checks; they are the observable coordinates of the missing uniform-in-tree RLPO theorem.

\subsection{Operational blueprint}

The HPO--RLPO synthesis leads to a practical architecture.
\begin{enumerate}
\item Construct a state $s_t$ containing market features, estimates $(\widehat\bmu_{e,t},\widehat\bSigma_t)$, the current tree $\widehat{\mathcal T}_t$ when hierarchical rules are used, previous weights $w_{t-1}$, and cost/liquidity variables.
\item Compute a baseline HPO prior $\w_\Phi(s_t)$, preferably RA-HRP or Schur-RA-HRP when a stable tree is available.
\item Parametrize the actor either as the convex residual \eqref{eq:convex-residual-policy} or, for hierarchical policies, through node-level residual logits \eqref{eq:rahrp-logit-actor}.
\item Train the policy on the cost-adjusted Bellman objective \eqref{eq:hpo-regularized-rlpo}, using the HPO defect and KL penalties to prevent economically unjustified deviations.
\item Attribute the learned deviations through \eqref{eq:rlpo-tree-kl}, separating root-level regime allocation from local security-level tilts.
\end{enumerate}

\begin{table}[htbp]
\centering
\small
\renewcommand{\arraystretch}{1.25}
\begin{tabular}{p{4.2cm} p{4.8cm} p{5.2cm}}
\toprule
\textbf{Static HPO object} & \textbf{RLPO counterpart} & \textbf{Interpretation} \\
\midrule
HPO map $\Phi(G(s))$ & deterministic policy $\pi_\Phi(\cdot\mid s)=\delta_{\Phi(G(s))}$ & every heuristic is a stationary Markov policy \\
Implied-return defect $\mathfrak d_s(w)$ & reward-shaping penalty & dynamic deviations must pay for local Sharpe inefficiency \\
RA-HRP node split $\beta_v^{\mathrm{RA}}$ & actor prior $\operatorname{logit}\beta_v^\theta=\operatorname{logit}\beta_v^{\mathrm{RA}}+a_v^\theta$ & RL learns interpretable residuals rather than unrelated weights \\
Pathwise KL decomposition & tree-level policy regularizer & global policy distance decomposes into node-level deviations \\
Tree-stability diagnostics & topology regularizer $d_{\mathcal T}(\mathcal T_t,\mathcal T_{t-1})$ & regime changes and dendrogram jumps become controlled state variables \\
Homotopy derivative $\sum_v(s_v-r_v)\alpha_v$ & myopic policy-gradient coordinates \eqref{eq:actor-gradient} & the actor's first learning signal is the nodewise alpha the prior leaves unexplained \\
HPO advantage $A^{\pi_\Phi}(s,w)$ & performance-difference identity \eqref{eq:value-gap} & exact discounted price of acting myopically; $\varepsilon$-myopic bound $\varepsilon/(1-\gamma)$ \\
\bottomrule
\end{tabular}
\caption{Static-to-dynamic bridge between HPO and RLPO. HPO supplies the myopic geometry and policy priors; RLPO supplies continuation value and dynamic policy improvement.}
\label{tab:hpo-rlpo-bridge}
\end{table}

The central message is therefore simple. HPO gives exact statewise optimality defects. RA-HRP gives structured hierarchical policy priors. RLPO learns dynamic deviations under costs, regimes, and feedback. The three pieces are complementary rather than competing.

\section{Transfer: Elliptical Risk Measures and Growth Optimality}\label{sec:elliptical}

Mean--variance optimality is the natural benchmark, but the conditions above would be parochial if they were artifacts of the quadratic objective. They are not.

\subsection{Elliptical transfer}

\begin{assumption}[Ellipticity]\label{ass:elliptical}
The return vector is elliptical, $\mathbf{r} \sim E_N(\bmu,\bSigma,\psi)$: the characteristic function is $\phi(\mathbf{t}) = e^{i\mathbf{t}^\top\bmu}\,\psi(\mathbf{t}^\top\bSigma\mathbf{t})$, with finite second moments normalized so $\bSigma$ is the covariance matrix.
\end{assumption}

\begin{theorem}[Risk-measure transfer; \textbf{P}, assembling classical results]\label{thm:elliptical}
Under Assumption~\ref{ass:elliptical}, let $\varrho$ be any risk measure on portfolio losses $L_\w = -\w^\top\mathbf{r}$ that is law-invariant, positively homogeneous, and cash-translative ($\varrho(L-c) = \varrho(L)-c$), with $\varrho(Z) \in (0,\infty)$ for $Z$ the standardized one-dimensional marginal. Then
\begin{equation}\label{eq:locscale}
\varrho(L_\w) \;=\; -\,\w^\top\bmu + \sqrt{\w^\top\bSigma\w}\;\varrho(Z),
\end{equation}
so for every target mean the $\varrho$-minimizing fully invested portfolio is the variance-minimizing one, the mean--$\varrho$ frontier coincides with the mean--variance frontier, and the generalized-Sharpe maximizer $\argmax_\w \w^\top\bmu_e/\varrho(L_\w + r_f)$ is the maximum-Sharpe portfolio. Consequently every coincidence condition of Sections~\ref{sec:flat}--\ref{sec:hrp} is simultaneously the condition for mean--$\CVaR_\alpha$ optimality at every level $\alpha$, and frontier portfolios are optimal for every nonsatiated risk-averse expected-utility investor.
\end{theorem}

\begin{proof}
Ellipticity implies the location--scale property of linear functionals: $\w^\top\mathbf{r} \overset{d}{=} \w^\top\bmu + \sqrt{\w^\top\bSigma\w}\,Z$ with $Z$ standardized, the same law for all $\w$ \citep{owenrabinovitch1983,embrechts2002}. Apply law invariance, translativity, and positive homogeneity to $L_\w = -\w^\top\bmu - \sqrt{\w^\top\bSigma\w}\,Z$ (using $-Z \overset{d}{=} Z$ by symmetry of elliptical marginals) to obtain \eqref{eq:locscale}. For fixed mean, minimizing $\varrho$ is minimizing $\sqrt{\w^\top\bSigma\w}$; CVaR at any $\alpha\in(0,1)$ satisfies the hypotheses with $\varrho(Z) = \CVaR_\alpha(Z) > 0$ \citep{rockafellaruryasev2000}. The expected-utility statement is the characterization of \citet{chamberlain1983}: under ellipticity with finite means, expected utility is a function of mean and variance alone, increasing in the former and decreasing in the latter for concave increasing $u$ --- and Chamberlain's result is an exact characterization, so ellipticity is also essentially the largest class for which this transfer can hold for all utilities.
\end{proof}

An allocator who replaces variance by CVaR escapes none of the conditions in this paper unless she simultaneously departs from ellipticity. The frameworks separate only through tails and asymmetry --- skewness, asymmetric heavy tails, state-dependent dependence --- which is where genuine mean--CVaR-versus-mean--variance trade-offs live.

\subsection{Growth-optimal transfer}

\begin{proposition}[Kelly coincidence; \textbf{P}, classical setting]\label{prop:kelly}
In the standard diffusion market $dS_{i,t}/S_{i,t} = \mu_i\,dt + (\bm\sigma\,dB_t)_i$ with $\bSigma = \bm\sigma\bm\sigma^\top \succ 0$ and money-market rate $r_f$, the almost-sure long-run growth rate of a constant-proportion strategy holding fractions $\w$ in risky assets is
$g(\w) = r_f + \w^\top\bmu_e - \tfrac12\w^\top\bSigma\w$,
maximized uniquely at the Kelly allocation $\w_K = \bSigma^{-1}\bmu_e$ \citep{kelly1956,merton1971}. Hence a fully invested heuristic $\w^\circ$ is proportional to the growth-optimal risky allocation --- i.e.\ is the risky composition of every fractional-Kelly strategy --- if and only if $\bmu_e \propto \bSigma\w^\circ$: \emph{the same condition} \eqref{eq:foc-tangency}. Every optimality set in this paper is simultaneously a Kelly-coincidence set in the diffusion limit.
\end{proposition}

\begin{proof}
It\^o's formula on $\ln W_t$ for the constant-proportion wealth dynamics gives $d\ln W_t = (r_f + \w^\top\bmu_e - \tfrac12\w^\top\bSigma\w)\,dt + \w^\top\bm\sigma\,dB_t$; the time-average of the martingale part vanishes a.s., leaving $g(\w)$, a strictly concave quadratic with maximizer $\bSigma^{-1}\bmu_e$. Proportionality $\w^\circ \propto \w_K$ is condition \eqref{eq:foc-tangency} verbatim.
\end{proof}

\section{Symmetry and Bayes Optimality: When the Heuristic Is the Optimizer}\label{sec:bayes}

The exact conditions classify parameter configurations; this section classifies \emph{beliefs}. The claim to be made precise is: each flat heuristic is the \emph{exactly optimal} Bayes portfolio of an investor whose uncertainty is symmetric under the group corresponding to her information set. Ignorance, formalized as invariance, does not approximately rationalize the heuristics --- it derives them.

\begin{lemma}[Predictive reduction; \textbf{P}]\label{lem:predictive}
Let the investor hold a belief $\pi$ over $(\bmu_e,\bSigma)$ and maximize the one-period objective $\E_\pi[\w^\top\mathbf{r}_e] - \tfrac{\lambda}{2}\Var_\pi(\w^\top\mathbf{r}_e)$, where $\mathbf{r}_e$ is the predictive excess return (parameters integrated out, with $\E[\mathbf{r}_e\mid\bmu_e,\bSigma] = \bmu_e$, $\Var(\mathbf{r}_e\mid\bmu_e,\bSigma) = \bSigma$). Then the objective equals $\w^\top\bar\bmu_e - \tfrac{\lambda}{2}\w^\top\hat\bSigma\w$ with
\[
\bar\bmu_e = \E_\pi[\bmu_e], \qquad
\hat\bSigma = \E_\pi[\bSigma] + \Cov_\pi(\bmu_e) \;\succ 0,
\]
and the Bayes portfolio is $\w^\star = \tfrac{1}{\lambda}\hat\bSigma^{-1}\bar\bmu_e$ --- the tangency portfolio of the \emph{predictive} pair.
\end{lemma}

\begin{proof}
Law of total expectation and total variance: $\Var_\pi(\w^\top\mathbf{r}_e) = \E_\pi[\w^\top\bSigma\w] + \Var_\pi(\w^\top\bmu_e) = \w^\top(\E_\pi[\bSigma] + \Cov_\pi(\bmu_e))\w$. The maximizer is the unconstrained quadratic optimum.
\end{proof}

\begin{lemma}[Commutant of the symmetric group; \textbf{P}]\label{lem:commutant}
If $M \in \R^{N\times N}$ is symmetric and $PMP^\top = M$ for every permutation matrix $P$, then $M = uI + v\ones\ones^\top$ for scalars $u, v$; if additionally $M \succ 0$ then $u > 0$ and $u + Nv > 0$, and $M^{-1}\ones = (u+Nv)^{-1}\ones$.
\end{lemma}

\begin{proof}
Invariance under the transposition of $i$ and $j$ forces $M_{ii} = M_{jj}$ and $M_{ik} = M_{jk}$ for $k \neq i,j$; ranging over transpositions equalizes all diagonal entries and all off-diagonal entries. Positive definiteness gives the eigenvalue conditions ($u$ with multiplicity $N{-}1$ on $\ones^\perp$, $u + Nv$ on $\ones$), and $M\ones = (u+Nv)\ones$ inverts as claimed.
\end{proof}

\begin{theorem}[Exchangeability $\Rightarrow$ equal weight; \textbf{P}]\label{thm:bayesew}
Suppose the belief $\pi$ is invariant under the full symmetric group acting on assets, $(\bmu_e,\bSigma) \mapsto (P\bmu_e, P\bSigma P^\top)$, and $\bar\bmu_e \neq 0$ with $\ones^\top\hat\bSigma^{-1}\bar\bmu_e > 0$. Then the Bayes portfolio of Lemma~\ref{lem:predictive}, normalized to full investment, is exactly the equal-weight portfolio --- for every risk aversion $\lambda$ and every exchangeable $\pi$.
\end{theorem}

\begin{proof}
Invariance of $\pi$ makes $\bar\bmu_e$ a fixed vector of every permutation, hence $\bar\bmu_e = m\ones$; and makes both $\E_\pi[\bSigma]$ and $\Cov_\pi(\bmu_e)$ invariant matrices, hence by Lemma~\ref{lem:commutant} $\hat\bSigma = uI + v\ones\ones^\top$ with $\hat\bSigma^{-1}\ones \propto \ones$. Thus $\w^\star \propto \hat\bSigma^{-1}m\ones \propto \ones$, and the sign condition fixes $m > 0$.
\end{proof}

\begin{theorem}[Standardized exchangeability $\Rightarrow$ inverse volatility; \textbf{P}]\label{thm:bayesiv}
Suppose volatilities $\sigma_1,\dots,\sigma_N$ are known, and the belief about the \emph{standardized} excess returns $\mathbf{z} = D^{-1}\mathbf{r}_e$ is exchangeable, with predictive mean $\bar{\mathbf{z}} = \bar s\ones$, $\bar s > 0$, and predictive covariance $\hat Z \succ 0$. Then the Bayes portfolio, normalized to full investment, is exactly the inverse-volatility portfolio.
\end{theorem}

\begin{proof}
Exchangeability of the predictive law of $\mathbf{z}$ forces $\bar{\mathbf{z}} \propto \ones$ and, by Lemma~\ref{lem:commutant}, $\hat Z = uI + v\ones\ones^\top$. The predictive pair for $\mathbf{r}_e = D\mathbf{z}$ is $(\bar s D\ones,\ D\hat Z D)$, so
\[
\w^\star \propto (D\hat Z D)^{-1}\bar s D\ones = \bar s\,D^{-1}\hat Z^{-1}\ones \propto D^{-1}\ones,
\]
the IV weights, using $\hat Z^{-1}\ones = (u+Nv)^{-1}\ones$.
\end{proof}

\begin{proposition}[Uninformative premia with conjugate uncertainty $\Rightarrow$ minimum variance; \textbf{P}]\label{prop:bayesgmv}
Suppose $\bSigma$ is known, the belief about $\bmu_e$ has mean $\bar\bmu_e = m\ones$, $m>0$, and conjugate-form uncertainty $\Cov_\pi(\bmu_e) = \tau^2\bSigma$, $\tau \ge 0$. Then the Bayes portfolio, normalized to full investment, is exactly the GMV portfolio, for every $\tau$.
\end{proposition}

\begin{proof}
$\hat\bSigma = (1+\tau^2)\bSigma$, so $\w^\star \propto \bSigma^{-1}\ones$.
\end{proof}

\begin{remark}[MD, ERC, HRP at this rung]\label{rem:bayesgaps}
Maximum diversification is the degenerate case: a point-mass belief $\bmu_e = S\bsig$ with known $\bSigma$ gives $\w^\star \propto \bSigma^{-1}\bsig$ immediately (Theorem~\ref{thm:md} read as a Bayes statement). For ERC, HRP, and RA-HRP we know of no theorem of comparable cleanliness: ERC arises from the symmetric beliefs of Theorem~\ref{thm:bayesiv} only on the constant-correlation slice where it collapses to IV (Proposition~\ref{prop:ercconst}), HRP's defining objects --- the tree and the recursive split --- have no invariance-group derivation we are aware of, and RA-HRP additionally requires a belief-level account of why noisy local cluster-Sharpe estimates should be trusted after aggregation but not globally inverted. A robust-control foundation for ERC under Sharpe-ratio ambiguity with a trusted correlation matrix, for HRP under ambiguity confined to inter-cluster blocks (the objects substitutions (i)--(ii) of Corollary~\ref{cor:threesubs} discard), and for RA-HRP under ambiguity-penalized local Sharpe scores are the natural conjectured forms; we state them as open problems, not results.
\end{remark}

The pattern of Theorems~\ref{thm:bayesew}--\ref{thm:bayesiv} and Proposition~\ref{prop:bayesgmv} is the group-theoretic content of the information-set ladder: choose the subgroup of asset relabelings under which your knowledge is genuinely indistinguishable, impose invariance of beliefs under it, and the Bayes portfolio is forced into the corresponding heuristic. The heuristics are not approximations to optimization; they are optimization under symmetry.

\section{Ambiguity, Robustness, and Estimation}\label{sec:ambiguity}

The Bayes theorems condition on symmetric beliefs; robust and finite-sample results show that the same portfolios emerge from adversarial and frequentist formulations.

\begin{proposition}[$1/N$ as the high-ambiguity limit; \textbf{C}]\label{prop:pflug}
In distributionally robust portfolio selection where the return law is known only to lie in a ball (Kantorovich/Wasserstein-type) of radius $\varepsilon$ around a nominal model and the investor minimizes the worst case of a convex risk functional over the ball, \citet{pflug2012} prove that as $\varepsilon \to \infty$ the robust optimal portfolio converges to equal weight: uniform diversification is the exact limit of robust optimization under extreme ambiguity. \citet{garlappi2007} obtain multi-prior portfolios interpolating between the plug-in optimum and conservative diversification as confidence shrinks.
\end{proposition}

\begin{proposition}[Constraints as shrinkage; \textbf{C}]\label{prop:jm}
\citet{jagannathanma2003} prove that the long-only (and upper-bounded) minimum-variance portfolio computed from an estimate $\hat\bSigma$ equals the \emph{unconstrained} minimum-variance portfolio of a modified matrix in which constraint multipliers shrink the covariances of assets at the bounds: imposing the ``wrong'' constraints performs covariance shrinkage \citep[cf.][]{ledoitwolf2004}, which is why constrained GMV often beats unconstrained GMV out of sample.
\end{proposition}

\begin{remark}[Estimation error and the empirical record; \textbf{E}]\label{rem:dgu}
The sample counterpart is well documented. Errors in estimated means degrade mean--variance utility by roughly an order of magnitude more than comparable errors in variances at moderate risk aversion \citep{chopraziemba1993}; optimized weights are extremely sensitive to means \citep{bestgrauer1991,michaud1989}; and expected returns are the slowest quantity to learn from data \citep{merton1980}. \citet{demiguel2009} show across datasets and fourteen optimizing models that none consistently beats $1/N$ out of sample in Sharpe ratio or certainty equivalent, with the estimation window required for the plug-in tangency portfolio to overtake $1/N$ growing far beyond available histories as $N$ increases; \citet{kanzhou2007} derive the optimal finite-sample combination of riskless asset, sample tangency, and sample GMV, formalizing the shrinkage the heuristics implement for free; and \citet{tuZhou2011} show that optimally combining $1/N$ with sophisticated rules can dominate both ingredients --- the sample-level analogue of the HPO projection \eqref{eq:hpo-proj} over a two-member family. These are finite-sample statements, distinct from the population conditions of this paper: $1/N$ can dominate the \emph{estimated} tangency portfolio while being far from the \emph{true} one. Theorem~\ref{thm:secondorder} is the population-side bridge: because efficiency is flat across each coincidence set, the heuristic forfeits only second-order Sharpe when its condition nearly holds, while the plug-in optimizer pays first-order weight instability for second-order objective gains.
\end{remark}

\section{Constrained, Costly, and Implementable Optimality}\label{sec:constrained}

The previous sections use the clean population benchmark: full investment, no short-sale constraints unless a rule imposes them internally, no turnover penalty, and no implementation wedge between a mathematical weight and a tradable portfolio. Actual mandates almost never look like that. They impose long-only constraints, concentration caps, sector limits, benchmark bands, turnover limits, and sometimes explicit transaction-cost models. The central message survives, but the exact condition changes from an equality to a \emph{normal-cone} statement: the missing implied returns may be supplied by active constraints.

\subsection{Normal-cone optimality under mandate constraints}

Let $K \subset \R^N$ be a nonempty closed convex feasible set. In applications $K$ usually includes the full-investment constraint, lower and upper bounds, and possibly linear exposure restrictions. Its normal cone at $\w\in K$ is
\begin{equation}
N_K(\w) = \{\mathbf{n}\in\R^N : \mathbf{n}^\top(\mathbf{v}-\w) \le 0\quad \forall\mathbf{v}\in K\}.
\end{equation}

\begin{theorem}[Constrained implied-return principle; \textbf{P}]\label{thm:normalcone}
Fix $\gamma>0$ and a closed convex mandate set $K$. A feasible portfolio $\w^\circ\in K$ solves the constrained mean--variance problem
\begin{equation}\label{eq:constrainedmv}
\max_{\w\in K}\left\{\w^\top\bmu_e - \frac{\gamma}{2}\w^\top\bSigma\w\right\}
\end{equation}
if and only if there exists $\mathbf{n}\in N_K(\w^\circ)$ such that
\begin{equation}\label{eq:normalconecondition}
\bmu_e = \gamma\bSigma\w^\circ + \mathbf{n}.
\end{equation}
Equivalently, the unconstrained implied-return residual
\begin{equation}
\mathbf{a}(\w^\circ) := \bmu_e - \gamma\bSigma\w^\circ
\end{equation}
must be entirely explainable by the shadow prices of the active mandate constraints.
\end{theorem}

\begin{proof}
The objective in \eqref{eq:constrainedmv} is strictly concave because $\bSigma\succ0$. Therefore first-order variational optimality is necessary and sufficient:
\[
(\bmu_e-\gamma\bSigma\w^\circ)^\top(\mathbf{v}-\w^\circ)\le0\qquad \forall \mathbf{v}\in K.
\]
This is precisely $\bmu_e-\gamma\bSigma\w^\circ\in N_K(\w^\circ)$.
\end{proof}

\begin{corollary}[Box-constrained implied returns; \textbf{P}]\label{cor:boxkkt}
Let
\[
K = \{\w: \ones^\top\w=1,\; \ell_i\le w_i\le u_i\ \forall i\}.
\]
Then $\w^\circ$ solves \eqref{eq:constrainedmv} iff there exist $b\in\R$, lower-bound multipliers $\lambda_i\ge0$, and upper-bound multipliers $\xi_i\ge0$ satisfying complementarity
\[
\lambda_i(w_i^\circ-\ell_i)=0,\qquad \xi_i(u_i-w_i^\circ)=0,
\]
and
\begin{equation}\label{eq:boximplied}
\bmu_e = \gamma\bSigma\w^\circ + b\ones - \bm\lambda + \bm\xi .
\end{equation}
Thus a zero weight does not require the asset to have zero optimal unconstrained demand; it only requires that its negative implied-return residual be absorbed by the lower-bound shadow price. Likewise, a capped name may have a positive residual absorbed by an upper-bound shadow price.
\end{corollary}

\begin{proof}
The normal cone of the equality-plus-box set is the span of $\ones$ plus outward normals to active lower and upper faces. Using the sign convention in Theorem~\ref{thm:normalcone} gives \eqref{eq:boximplied}.
\end{proof}

\begin{remark}[Why this matters for heuristic portfolios]\label{rem:constraintsheuristics}
Theorems~\ref{thm:ew}--\ref{thm:md} and \ref{thm:ercchar} classify exact unconstrained coincidence. Theorem~\ref{thm:normalcone} gives the mandate version. A long-only equal-weight, inverse-volatility, ERC, MD, HRP, or GMV portfolio can be constrained-optimal even when it is not unconstrained-optimal, but only if the residual alphas are located on the active faces of the feasible set. This is the precise population analogue of the finite-sample result of \citet{jagannathanma2003}: constraints do not merely ``regularize'' the optimizer; they add shadow premia.
\end{remark}

\subsection{Transaction-cost wedges}

Suppose the current portfolio is $\w_-$ and the investor pays a quadratic trading penalty with positive semidefinite matrix $\Lambda\succeq0$ and intensity $\kappa\ge0$:
\begin{equation}\label{eq:costobjective}
\max_{\w\in K}\left\{\w^\top\bmu_e - \frac{\gamma}{2}\w^\top\bSigma\w - \frac{\kappa}{2}(\w-\w_-)^\top\Lambda(\w-\w_-)\right\}.
\end{equation}
Quadratic costs are a local model of market impact and turnover aversion, and they are the static one-period analogue of dynamic transaction-cost models such as \citet{garleanuPedersen2013}.

\begin{proposition}[Cost-adjusted implied returns; \textbf{P}]\label{prop:costimplied}
A feasible $\w^\circ\in K$ solves \eqref{eq:costobjective} iff there exists $\mathbf{n}\in N_K(\w^\circ)$ such that
\begin{equation}\label{eq:costimplied}
\bmu_e = \gamma\bSigma\w^\circ + \kappa\Lambda(\w^\circ-\w_-) + \mathbf{n}.
\end{equation}
In the unconstrained full-investment case, this becomes
\begin{equation}\label{eq:costimpliedsimple}
\bmu_e = \gamma\bSigma\w^\circ + \kappa\Lambda(\w^\circ-\w_-) + b\ones.
\end{equation}
\end{proposition}

\begin{proof}
The gradient of the objective in \eqref{eq:costobjective} is $\bmu_e-\gamma\bSigma\w-\kappa\Lambda(\w-\w_-)$. Apply the same variational inequality as in Theorem~\ref{thm:normalcone}.
\end{proof}

\begin{remark}[No-trade interpretation]\label{rem:notrade}
If $\w^\circ=\w_-$, the cost wedge vanishes; if the candidate rule demands a large rebalance, the wedge can dominate the statistical signal. Equation~\eqref{eq:costimplied} says that a rule is implementably optimal only when expected excess returns compensate both risk and the marginal cost of moving from the current book. This is why the same population-optimal rule may be undesirable at daily frequency and reasonable at quarterly frequency.
\end{remark}

\begin{remark}[Linear costs, quadratic wedges, and RA-HRP turnover; \textbf{E}]\label{rem:linear-costs-rahrp}
The quadratic wedge in \eqref{eq:costimplied} is the differentiable local approximation to implementation frictions. A one-way linear cost $c\|\w-\w_-\|_1$ gives the same economic message with a nonsmooth subgradient: the implied-return residual must compensate risk plus an element of $c\,\partial\|\w-\w_-\|_1$, together with any mandate normal cone. The RA-HRP evidence is consistent with both descriptions. The reported rolling implementation has monthly two-way turnover of approximately $10.0\%$ for traditional HRP, $11.1\%$ for RA-HRP, and $12.5\%$ for Schur-RA-HRP; under one-way cost stress tests of 5--20 basis points, the RA-HRP net Sharpe remains close to the gross Sharpe \citep{noguer2026rahrp}. Thus the same structural feature that avoids the global inverse --- local, bounded split revisions on a tree --- also keeps trading frictions modest under linear and quadratic cost models.
\end{remark}

\subsection{Benchmark-relative and factor-neutral mandates}

Many institutional portfolios are not judged in absolute variance but relative to a benchmark $\w_b$. Let active weights be $\mathbf{x}=\w-\w_b$, with $\ones^\top\mathbf{x}=0$, active covariance still $\bSigma$, and active expected excess return $\bmu_a$. A tracking-error mandate solves
\begin{equation}\label{eq:active}
\max_{\mathbf{x}\in K_a}\left\{\mathbf{x}^\top\bmu_a - \frac{\gamma}{2}\mathbf{x}^\top\bSigma\mathbf{x}\right\}.
\end{equation}
The identical normal-cone result gives
\begin{equation}\label{eq:activeimplied}
\bmu_a = \gamma\bSigma\mathbf{x}^\circ + \mathbf{n}_a,
\qquad \mathbf{n}_a\in N_{K_a}(\mathbf{x}^\circ).
\end{equation}
Therefore a heuristic active overlay is optimal exactly when it prices active alphas through active covariances, after accounting for active constraints. In particular, the equal-weight theorem has a benchmark-relative form: equal active weights across an eligible sleeve require zero alphas against the active equal-weight sleeve, not against the absolute portfolio.

\section{Estimation Geometry, High-Dimensional Risk, and Empirical Diagnostics}\label{sec:estimation}

The paper has so far separated population truth from finite-sample practice. This section connects them. The exact optimality sets tell us what must be true; the estimation geometry tells us why, even when those conditions are false, a rule that ignores some inputs may dominate an optimizer that estimates them badly.

\subsection{First-order sensitivity of tangency weights}

Let
\[
\mathcal{T}(\bmu_e,\bSigma)=\frac{\bSigma^{-1}\bmu_e}{\ones^\top\bSigma^{-1}\bmu_e}
\]
be the normalized tangency map and write $c=\ones^\top\bSigma^{-1}\bmu_e>0$, $\w_T=\mathcal{T}(\bmu_e,\bSigma)$.

\begin{theorem}[Differential of the tangency map; \textbf{P}]\label{thm:tangencydiff}
For perturbations $(\delta\bmu,\delta\bSigma)$ small enough that $\bSigma+\delta\bSigma\succ0$ and the denominator remains positive, the first differential of the tangency weights is
\begin{equation}\label{eq:tangencydiff}
\delta\w_T
=
\left(I-\w_T\ones^\top\right)
\left(\frac{1}{c}\bSigma^{-1}\delta\bmu-\bSigma^{-1}\delta\bSigma\,\w_T\right)
+o(\|\delta\bmu\|+\|\delta\bSigma\|).
\end{equation}
Consequently, mean errors enter with amplification $1/c$, while covariance errors enter through the already-chosen direction $\w_T$.
\end{theorem}

\begin{proof}
Let $\mathbf{x}=\bSigma^{-1}\bmu_e$ and $c=\ones^\top\mathbf{x}$, so $\w_T=\mathbf{x}/c$. The differential of $\mathbf{x}$ is
\[
\delta\mathbf{x}=\bSigma^{-1}\delta\bmu-\bSigma^{-1}\delta\bSigma\,\bSigma^{-1}\bmu_e
=\bSigma^{-1}\delta\bmu-c\bSigma^{-1}\delta\bSigma\,\w_T.
\]
Since $\delta c=\ones^\top\delta\mathbf{x}$,
\[
\delta\w_T=\frac{\delta\mathbf{x}}{c}-\frac{\mathbf{x}\,\delta c}{c^2}
=\frac{1}{c}(I-\w_T\ones^\top)\delta\mathbf{x},
\]
which gives \eqref{eq:tangencydiff}.
\end{proof}

\begin{remark}[Why means dominate]\label{rem:meansdominate}
Equation~\eqref{eq:tangencydiff} is the infinitesimal version of the empirical findings of \citet{merton1980}, \citet{bestgrauer1991}, \citet{chopraziemba1993}, and \citet{michaud1989}. Expected returns are small in annualized units, so $c=\ones^\top\bSigma^{-1}\bmu_e$ is often small; the inverse factor $1/c$ turns modest premium errors into large allocation errors. A volatility-only rule may be population-suboptimal yet estimation-robust because it deletes the most unstable term in \eqref{eq:tangencydiff}.
\end{remark}

\subsection{Projection onto a heuristic optimality set}

The exact conditions can be turned into diagnostics. Fix $\bSigma$ and a candidate $\w^\circ$. Its tangency-coincidence set is the ray $\{a\bSigma\w^\circ:a>0\}$. Given an estimated premium vector $\hat\bmu_e$, ask how far it is from that ray.

\begin{proposition}[Closest implied-return ray; \textbf{P}]\label{prop:projection}
Let $M\succ0$ define a diagnostic norm $\|\mathbf{z}\|_M^2=\mathbf{z}^\top M\mathbf{z}$ and set $\mathbf{q}=\bSigma\w^\circ$. The best positive-ray approximation to $\hat\bmu_e$ is
\begin{equation}\label{eq:astar}
a^* = \max\left\{0,\frac{\mathbf{q}^\top M\hat\bmu_e}{\mathbf{q}^\top M\mathbf{q}}\right\},
\qquad
\hat\bmu_e^{\mathrm{proj}}=a^*\mathbf{q},
\end{equation}
and the squared distance
\begin{equation}\label{eq:raydistance}
d_M^2(\hat\bmu_e,\mathcal{O}_{\bSigma}(\w^\circ))
=\|\hat\bmu_e-a^*\bSigma\w^\circ\|_M^2
\end{equation}
quantifies the violation of the heuristic's exact optimality condition.
\end{proposition}

\begin{proof}
Minimize the one-dimensional quadratic $Q(a)=\|\hat\bmu_e-a\mathbf{q}\|_M^2$ over $a\ge0$. The unconstrained minimizer is $(\mathbf{q}^\top M\hat\bmu_e)/(\mathbf{q}^\top M\mathbf{q})$; projection onto $[0,\infty)$ gives \eqref{eq:astar}.
\end{proof}

\begin{remark}[A useful choice of metric]\label{rem:projectionmetric}
Taking $M=\hat\Omega_\mu^{-1}$, where $\hat\Omega_\mu$ is the estimated covariance matrix of premium estimates, turns \eqref{eq:raydistance} into a Wald-type statistic. Taking $M=\bSigma^{-1}$ measures distance in the same dual metric that appears in $S_{\max}^2=\bmu_e^\top\bSigma^{-1}\bmu_e$. Either way, the diagnostic decomposes the failure of a heuristic into a scalar scale estimate $a^*$ plus an orthogonal residual premium vector.
\end{remark}

\subsection{GRS and spanning tests as optimality tests}

Remark~\ref{rem:grs} states the principle. The usable formula is as follows. For $T$ excess-return observations, regress each asset on the candidate heuristic excess return $r^\circ_{e,t}$:
\begin{equation}
r_{e,i,t}=\alpha_i+\beta_i r^\circ_{e,t}+\epsilon_{i,t}.
\end{equation}
Let $\hat{\bm\alpha}$ be the vector of intercepts, $\hat\Sigma_\epsilon$ the residual covariance matrix, $\bar r^\circ$ the sample mean of the candidate return, and $\hat\sigma_\circ^2$ its sample variance. With one candidate factor, the Gibbons--Ross--Shanken statistic is
\begin{equation}\label{eq:grsformula}
\mathrm{GRS}
=
\frac{T-N-1}{N}
\frac{\hat{\bm\alpha}^\top\hat\Sigma_\epsilon^{-1}\hat{\bm\alpha}}{1+(\bar r^\circ)^2/\hat\sigma_\circ^2}
\sim F_{N,T-N-1}
\end{equation}
under the Gaussian exact-efficiency null \citep{gibbons1989}. Small-sample refinements and related asymptotics are discussed by \citet{jobsonkorkie1980} and \citet{memmel2003}. Mean--variance spanning tests such as \citet{hubermankandel1987} provide the multi-portfolio analogue: a set of heuristic sleeves spans the efficient frontier iff adding the remaining assets does not improve the frontier.

\begin{remark}[Interpreting rejection]\label{rem:grsinterpret}
Rejection of \eqref{eq:grsformula} does not mean the heuristic is useless; it means the exact population condition is statistically inconsistent with the sample. Theorem~\ref{thm:secondorder} then gives the next question: is the rejection economically large? A large $p$-value is evidence of statistical non-rejection; a small angular loss $1-\eta^2$ is evidence of economic adequacy. They answer different questions.
\end{remark}

\subsection{High-dimensional covariance and the inversion tax}

Even a rule with a clean population foundation can fail in finite samples if it requires inverting a noisy covariance matrix. This is the high-dimensional counterpart of mean-estimation fragility.

\begin{proposition}[Sample GMV risk inflation; \textbf{C}]\label{prop:rmtgmv}
Let $\hat\bSigma$ be the sample covariance matrix from $T$ independent Gaussian observations on $N$ assets with true covariance $\bSigma\succ0$, and let $\hat\w_{\mathrm{GMV}}=\hat\bSigma^{-1}\ones/(\ones^\top\hat\bSigma^{-1}\ones)$ whenever $\hat\bSigma$ is nonsingular. In the high-dimensional regime $N/T\to q\in(0,1)$, random-matrix results imply the out-of-sample variance inflation
\begin{equation}\label{eq:rmtinflation}
\frac{\hat\w_{\mathrm{GMV}}^\top\bSigma\hat\w_{\mathrm{GMV}}}
{\w_{\mathrm{GMV}}^\top\bSigma\w_{\mathrm{GMV}}}
\longrightarrow \frac{1}{1-q}
\end{equation}
under standard regularity conditions. This is the \emph{inversion tax}: as the cross-section approaches the sample length, the plug-in GMV portfolio becomes mechanically overfit. See \citet{elkaroui2010} and the shrinkage literature of \citet{ledoitwolf2004}.
\end{proposition}

\begin{remark}[What HRP and shrinkage are buying]\label{rem:rmtbuying}
Equation~\eqref{eq:rmtinflation} clarifies the empirical appeal of HRP, RA-HRP, inverse volatility, ERC, and covariance shrinkage. They give up exact global inversion of $\bSigma$ in exchange for reduced estimation variance. Fixed-tree RA-HRP belongs to this low-inversion family: its Jacobian is a composition of local rational split maps whose denominators are cluster variances, score sums, and floors, rather than the smallest eigenvalue of the full covariance matrix. Schur-RA-HRP restores some cross-cluster information through local block inversions, but it still avoids a single global $\bSigma^{-1}$. Cotton's Schur-complement view of HRP and minimum variance, and subsequent hierarchical minimum-variance constructions, make this trade-off explicit by partially restoring the cross-block covariance information that HRP discards \citep{cotton2024,mograby2025}. Cotton's 2024 paper introduces Schur Complementary Allocation as a bridge between HRP and minimum-variance portfolios, while Mograby's 2025 work develops a recursive Schur-complement framework for hierarchical minimum variance; both directly support the interpretation of HRP as a controlled deletion or restoration of inverse-covariance information.
\end{remark}

\subsection{A reproducible empirical protocol}\label{sec:protocol}

A paper claiming that a heuristic is ``optimal'' should report three layers, not one.

\begin{enumerate}
\item \textbf{Population-condition diagnostic.} Estimate the rule-specific restriction: EW zero alphas against EW; IV Sharpe ratios versus correlation row sums; ERC performance parity; MD equal Sharpe ratios; GMV equal premia; HRP Schur-substitution residuals.
\item \textbf{Statistical test.} Report the GRS statistic \eqref{eq:grsformula}, a spanning test when several sleeves are compared, and robust/HAC versions when returns are serially dependent or heteroskedastic.
\item \textbf{Economic loss.} Report realized Sharpe efficiency, angular loss $1-\eta^2$, turnover, transaction-cost-adjusted Sharpe, drawdown, and CVaR. Statistical rejection with economically negligible angular loss is not a failure of the heuristic; non-rejection with high turnover may still be implementation failure.
\item \textbf{Stability.} Use rolling windows, block bootstrap, and perturbation experiments. Report how often the ranking of rules changes under resampling. This directly measures the estimation terms isolated in Theorem~\ref{thm:tangencydiff}.
\item \textbf{Dimensionality.} Always report $q=N/T$ for the covariance window. Without $q$, a comparison between plug-in optimization and heuristics is uninterpretable, because \eqref{eq:rmtinflation} changes the expected cost of inversion.
\end{enumerate}

\subsection{Computational profile of the allocation ladder}\label{sec:computational-profile}

The exact optimality results should be read together with the numerical burden of the rule. In dense implementations, the usual computational profile is summarized in Table~\ref{tab:complexity}. The point is not that lower complexity is automatically better; it is that every avoided inversion, local score, or tree constraint deletes a specific source of estimation variance and replaces it by a specific structural assumption.

\begin{table}[htbp]
\centering
\small
\renewcommand{\arraystretch}{1.25}
\begin{tabularx}{\textwidth}{p{2.6cm} p{3.2cm} p{3.3cm} X}
\toprule
\textbf{Rule} & \textbf{Typical dense cost} & \textbf{Main numerical bottleneck} & \textbf{Interpretation of the cost} \\
\midrule
EW & $O(N)$ & none & Maximum symmetry; no sampling burden beyond defining the universe. \\
IV / IVar & $O(N)$ after vol estimates & volatility estimation & Uses diagonal risk only; deletes correlation and premium information. \\
GMV & $O(N^3)$ after covariance estimation & global covariance inversion and conditioning & Buys the full covariance signal at the price of the high-dimensional inversion tax. \\
ERC & iterative, often $O(KN^2)$ for $K$ solver steps & nonlinear risk-budget solve & Uses covariance without premia; solver stability depends on conditioning and constraints. \\
HRP & clustering plus recursive local variances & dendrogram stability & Replaces global inversion by tree topology and local variance aggregation. \\
RA-HRP & HRP cost plus local mean/Sharpe scores & local premium noise and floor activation & Spends a controlled amount of mean-estimation risk locally rather than globally. \\
Schur-RA-HRP & RA-HRP plus sibling-block inversions & conditioning of local Schur blocks & Restores cross-cluster hedging information without paying the full global inverse. \\
RLPO residual policy & training-dependent & off-policy error, reward design, turnover and regime shift & Learns dynamic deviations only when continuation value pays for the static HPO defect. \\
\bottomrule
\end{tabularx}
\caption{Computational profile of the HPO--RLPO ladder. Costs are indicative for dense implementations and should be reported with the covariance-window ratio $q=N/T$.}
\label{tab:complexity}
\end{table}

A reproducible empirical implementation should therefore report not only realized returns but also the cost of producing the weights: the length of the estimation window, $q=N/T$, covariance-conditioning diagnostics, tree-stability diagnostics, floor-activation frequencies for RA-HRP, local Schur-block conditioning for Schur-RA-HRP, and post-cost turnover. Without these quantities, an empirical comparison confounds portfolio skill with numerical fragility.

\section{Non-Elliptical Tails, Downside Risk, and Where the Conditions Break}\label{sec:tails}

The elliptical transfer theorem is intentionally sharp. It says that replacing variance by CVaR does not change the optimality conditions when all linear portfolios share the same standardized marginal shape. The converse message is equally important: departures from ellipticity are exactly where different objectives begin to disagree.

\begin{definition}[Tail-shape residual]\label{def:tailshape}
For a portfolio $\w$, define the standardized excess-return residual
\begin{equation}
Z_\w = \frac{\w^\top(\mathbf{r}-\bmu)}{\sqrt{\w^\top\bSigma\w}}.
\end{equation}
The tail-shape residual between two portfolios $\w$ and $\mathbf{v}$ at level $\alpha$ is
\begin{equation}
\Delta_\alpha(\w,\mathbf{v})
= \CVaR_\alpha(-Z_\w)-\CVaR_\alpha(-Z_\mathbf{v}).
\end{equation}
Under ellipticity, $\Delta_\alpha(\w,\mathbf{v})=0$ for all $\w,\mathbf{v},\alpha$.
\end{definition}

\begin{proposition}[First-order downside wedge; \textbf{P}]\label{prop:downsidewedge}
Assume a law-invariant homogeneous downside risk measure admits the representation
\begin{equation}
\varrho(L_\w)= -\w^\top\bmu + k_\w\sqrt{\w^\top\bSigma\w},
\end{equation}
where $k_\w$ is differentiable in a neighborhood of $\w^\circ$. Then the first-order condition for maximizing mean per unit downside risk contains the additional tail-shape term $\nabla k_\w|_{\w=\w^\circ}$. In particular, the mean--variance implied-return condition $\bmu_e\propto\bSigma\w^\circ$ remains sufficient only when $k_\w$ is locally constant along feasible directions.
\end{proposition}

\begin{proof}
Differentiate the generalized denominator $k_\w(\w^\top\bSigma\w)^{1/2}$. In addition to the variance-gradient term $k_\w\bSigma\w/(\w^\top\bSigma\w)^{1/2}$, one obtains $(\w^\top\bSigma\w)^{1/2}\nabla k_\w$. The latter vanishes exactly when standardized tail shape is locally invariant.
\end{proof}

\begin{remark}[Practical implication]\label{rem:tailpractical}
This identifies the boundary of the paper's transfer claims. If assets have asymmetric crash exposure, volatility targeting, option-like payoffs, or regime-dependent correlations, then a heuristic can be Sharpe-near-optimal and CVaR-poor. The missing state variable is not variance but tail-shape heterogeneity. Empirically, this means the protocol in Section~\ref{sec:protocol} should report both angular Sharpe loss and tail-shape residuals.
\end{remark}

\section{Synthesis: The Optimality Map}\label{sec:synthesis}

Table~\ref{tab:map} collects the exact population conditions; Table~\ref{tab:bayes} collects the belief-level foundations. Together they give the information-set ladder theorem-level content at every rung where such content currently exists, and mark the missing rungs as open.

Table~\ref{tab:structural} separates the same rules by inputs, inversion profile, estimation vulnerability, and topological coincidence boundary. Table~\ref{tab:hpo-rlpo-synthesis} then places those static rules inside the dynamic RLPO hierarchy developed in Section~\ref{sec:rlpo}. Together with the HPO defect and bias--variance identities of Section~\ref{sec:hpo}, this is the practical map: it explains why a rule can be theoretically suboptimal and still empirically robust when the coordinates required by the optimizer are too noisy, and why an RL policy should deviate from that rule only when continuation value pays for the static defect.

\begin{table}[htbp]
\centering
\scriptsize
\setlength{\tabcolsep}{2pt}
\renewcommand{\arraystretch}{1.22}
\begin{tabular}{p{2.05cm} p{2.55cm} p{2.1cm} p{3.65cm} p{4.6cm}}
\toprule
\textbf{Allocation rule} & \textbf{Inputs consumed} & \textbf{Inversion profile} & \textbf{Primary estimation vulnerability} & \textbf{Topological coincidence boundary} \\
\midrule
Equal weight & None & None & Symmetry misspecification, not estimation noise & $\bmu_e\propto\bSigma\ones$ \\
Inverse volatility & Marginal volatilities $\sigma_i$ & None & Volatility estimates and correlation-row-sum mismatch & $S_i\propto\sum_j\rho_{ij}$ \\
Global min variance & Full $\bSigma$ & Global $\bSigma^{-1}$ & High-dimensional inversion tax $1/(1-q)$ and covariance conditioning & $\bmu_e\propto\ones$ for tangency coincidence \\
Traditional HRP & Tree plus cluster IVP variances & No global inverse & Dendrogram instability and Schur-substitution residuals & Diagonal $\bSigma$ under any tree; uncorrelated exchangeable blocks under balanced tree-aligned bisection; plus equal premia for tangency \\
RA-HRP & Tree plus cluster means and IVP variances & No global inverse & Persistence and noise of local cluster-Sharpe scores & Nonlinear fixed point $\bmu_e=a\bSigma\Phi^{\mathrm{RA}}_{\mathcal T,\varepsilon}(\bmu_e,\bSigma)$; nodewise equality of RA splits and Schur-tangency splits \\
Schur-RA-HRP & Tree, cluster premia, and Schur conditional risks & Local block inversions & Cross-cluster conditional-covariance estimation and block conditioning & Collapses to RA-HRP when sibling hedge terms vanish; exact only when conditional Schur scores match Schur-tangency capital masses node by node \\
\bottomrule
\end{tabular}
\caption{Structural comparison of heuristic, hierarchical, and return-adjusted allocation rules. The table distinguishes the economic input set from the numerical inversion profile and from the exact population boundary required for tangency coincidence.}
\label{tab:structural}
\end{table}

\begin{table}[htbp]
\centering
\small
\renewcommand{\arraystretch}{1.25}
\begin{tabular}{p{3.6cm} p{5.0cm} p{5.5cm}}
\toprule
\textbf{Layer} & \textbf{Mathematical object} & \textbf{Role in the research program} \\
\midrule
HPO / exact optimality & static map $w=\Phi(\widehat\mu,\widehat\Sigma,\widehat{\mathcal T})$ and defect $\mathfrak d$ & identifies when a heuristic is locally optimal and how much myopic Sharpe efficiency it loses \\
RA-HRP / Schur-RA-HRP & tree-factorized map with branch probabilities $\{\beta_v\}$ & supplies interpretable hierarchical policy priors and node-level KL attribution \\
RLPO & Markov policy $\pi_\theta(w\mid s)$ with Bellman value & learns dynamic deviations from HPO under costs, regimes, and continuation value \\
\bottomrule
\end{tabular}
\caption{Division of labor between static HPO, hierarchical RA-HRP priors, and dynamic RLPO.}
\label{tab:hpo-rlpo-synthesis}
\end{table}

\begin{table}[ht]
\centering
\small
\renewcommand{\arraystretch}{1.35}
\begin{tabular}{p{2.45cm} p{3.1cm} p{5.6cm} p{3.1cm}}
\toprule
\textbf{Rule} & \textbf{Uses} & \textbf{Exact tangency condition} (this paper) & \textbf{Informational regime} \\
\midrule
Equal weight & nothing & zero alpha vs.\ EW: $\bmu_e \propto \bSigma\ones$ \hfill [P, Thm.~\ref{thm:ew}] & total ambiguity \\
Inverse volatility & $\sigma_i$ & $S_i \propto \sum_j \rho_{ij}$ \hfill [P, Thm.~\ref{thm:iv}] & vols known, correlations exchangeable \\
ERC / risk parity & $\bSigma$ (no $\bmu$) & performance parity: $w_i\mu_{e,i}$ equal \hfill [P, Thm.~\ref{thm:ercchar}]; sufficient: equal $S_i$, constant $\rho$ [P, Cor.~\ref{cor:mrt}] & $\bSigma$ trusted, premia exchangeable per unit risk budget \\
Max.\ diversification & $\bsig$, $\bSigma$ & equal Sharpe ratios: $\bmu_e \propto \bsig$ \hfill [P, Thm.~\ref{thm:md}] & risk priced uniformly per unit vol \\
HRP & tree $+$ cluster variances & GMV-coincidence: diagonal $\bSigma$, any tree [P, Thm.~\ref{thm:hrpdiag}]; uncorrelated exchangeable blocks [P, Prop.~\ref{prop:hrpblocks}]; then $\bmu_e \propto \ones$ [P, Cor.~\ref{cor:hrpmv}]; gap $=$ three Schur substitutions [P, Cor.~\ref{cor:threesubs}] & ordinal correlation structure trusted, $\bSigma^{-1}$ not \\
RA-HRP & tree $+$ cluster Sharpe scores & nonlinear fixed point $\bmu_e=a\bSigma\Phi^{\mathrm{RA}}_{\mathcal T,\varepsilon}(\bmu_e,\bSigma)$ [P, Thm.~\ref{thm:rahrp-coincidence}]; nodewise Schur-tangency split matching; HRP--RA pathwise/KL attribution [P, Thm.~\ref{thm:rahrp-pathwise}] & tree trusted; premia trusted only through local IVP aggregation/flooring; no global inverse \\
Schur-RA-HRP & tree $+$ cluster Sharpe scores $+$ Schur complements & conditional-risk score from Def.~\ref{def:schur-rahrp}; Schur score wedge [P, Prop.~\ref{prop:schur-score-wedge}]; locally exact only when conditional scores reproduce Schur-tangency masses & tree trusted; cross-cluster covariance trusted locally, but global inverse still avoided \\
Minimum variance & $\bSigma$ & $\bmu_e \propto \ones$ \hfill [P, Thm.~\ref{thm:gmv}] & risk forecastable, premia not \\
Tangency & $\bmu_e$, $\bSigma$ & --- (benchmark) & full knowledge \\
\bottomrule
\end{tabular}
\caption{The optimality map. All conditions are also Kelly-coincidence conditions in the diffusion limit [P, Prop.~\ref{prop:kelly}], transfer to mean--CVaR and expected utility under ellipticity [P, Thm.~\ref{thm:elliptical}], and are GRS-testable as zero-alpha restrictions [C, Rem.~\ref{rem:grs}]. One-factor translations are in Prop.~\ref{prop:onefactor}.}
\label{tab:map}
\end{table}

\begin{table}[ht]
\centering
\small
\renewcommand{\arraystretch}{1.35}
\begin{tabular}{p{2.8cm} p{6.3cm} p{5.2cm}}
\toprule
\textbf{Rule} & \textbf{Belief-level foundation} & \textbf{Status} \\
\midrule
Equal weight & Bayes rule under any fully exchangeable belief & proved [P, Thm.~\ref{thm:bayesew}]; also robust high-ambiguity limit [C, Prop.~\ref{prop:pflug}] \\
Inverse volatility & Bayes rule under exchangeable beliefs on standardized returns, $\sigma_i$ known & proved [P, Thm.~\ref{thm:bayesiv}] \\
Minimum variance & Bayes rule under conjugate mean uncertainty with prior mean $\propto \ones$, $\bSigma$ known & proved [P, Prop.~\ref{prop:bayesgmv}]; constraints-as-shrinkage in sample [C, Prop.~\ref{prop:jm}] \\
Max.\ diversification & degenerate belief $\bmu_e \propto \bsig$, $\bSigma$ known & immediate [P, Rem.~\ref{rem:bayesgaps}] \\
ERC / risk parity & robust foundation under Sharpe ambiguity with trusted correlations & \emph{open} [Rem.~\ref{rem:bayesgaps}] \\
HRP & robust foundation under ambiguity confined to inter-cluster blocks & \emph{open} [Rem.~\ref{rem:bayesgaps}] \\
RA-HRP & local cluster-Sharpe shrinkage on a fixed tree; exact symmetry/robust foundation for noisy premia remains unresolved & fixed-tree theory and empirical proof of concept in \citet{noguer2026rahrp}; belief-level foundation \emph{open} [Rem.~\ref{rem:bayesgaps}] \\
Schur-RA-HRP & local conditional-risk version of RA-HRP; robust foundation under uncertainty about cross-cluster hedge terms remains unresolved & Schur extension in Def.~\ref{def:schur-rahrp}; belief-level foundation \emph{open} [Rem.~\ref{rem:bayesgaps}] \\
\bottomrule
\end{tabular}
\caption{Belief-level foundations: symmetry-derived Bayes optimality and robust limits.}
\label{tab:bayes}
\end{table}

\clearpage

\section{Limitations and Open Problems}\label{sec:limitations}

The paper closes several algebraic gaps but leaves important research problems open. These limitations are useful because they identify the exact objects that a referee, implementer, or subsequent paper should attack.

\paragraph{Random dendrogram asymptotics.}
The fixed-tree theory is sharp, but rolling HRP and RA-HRP estimate the tree. A complete asymptotic theory must combine a continuous delta method on the score maps with a discrete stability theorem for merge sets, cophenetic matrices, or tree-edit representations. The correct theorem should condition on a population merge-separation event, prove topology preservation with high probability, and then control off-event portfolio jumps through the pathwise product formula of Theorem~\ref{thm:rahrp-pathwise}.

\paragraph{Robust-control foundations for ERC, HRP, and RA-HRP.}
The Bayes foundations for EW, IV, and GMV are clean because the relevant symmetry groups are clean. Comparable foundations for ERC, HRP, RA-HRP, and Schur-RA-HRP require ambiguity sets that act on risk budgets, inter-cluster covariance blocks, and local Sharpe scores. Corollary~\ref{cor:threesubs}, Theorem~\ref{thm:rahrp-coincidence}, and Theorem~\ref{thm:hpo-schur-tangency} identify the exact algebraic quantities such ambiguity sets must preserve or penalize.

\paragraph{Non-elliptical downside risk.}
The mean--CVaR transfer is exact under ellipticity and fails precisely through standardized tail-shape heterogeneity. A full downside-risk extension would replace the scalar angular defect by a pair consisting of Sharpe-angle loss and tail-shape residual $\Delta_\alpha(w,v)$, then study when local hierarchical scores can control both simultaneously.

\paragraph{Dynamic market impact and partial observability.}
The RLPO bridge uses a discounted Markov formulation with explicit rewards and costs. Realistic execution adds transient impact, hidden liquidity, delayed fills, and partial observability. The natural extension is a belief-state RLPO formulation in which the HPO prior is computed from filtered estimates and the deviation principle is applied to belief-state continuation values.

\paragraph{Uniform empirical validation.}
The theory supplies falsification diagnostics but does not report new numerics. A definitive empirical companion should run the protocol of Section~\ref{sec:protocol} across asset classes, universes, rebalance frequencies, covariance estimators, cost models, and tree-linkage rules; report split-sample GRS tests and angular losses; and attribute RA-HRP or RLPO gains node by node through the KL and alpha decompositions.

\section{Conclusion}\label{sec:conclusion}

Eight structural facts organize everything above.

First, HPO is a mathematical object in its own right. A heuristic portfolio rule is an information-restricted map from a reduced statistic to the simplex; its population error is the implied-return projection defect $\mathfrak D_\Phi$; and Theorem~\ref{thm:hpo-defect} proves that this defect is exactly squared Sharpe inefficiency. The score-tree formalism then shows that HRP, RA-HRP, Schur-RA-HRP, and their interpolations are products of local binary kernels, while Theorem~\ref{thm:hpo-schur-tangency} identifies the exact Schur-tangency mass each node is trying to approximate.

Second, by the implied-return principle, ``is this heuristic optimal'' is never yes-or-no; it is a question about a closed-form set of parameter configurations, and adopting a rule \emph{is} asserting its set. Equal weight asserts zero alphas against itself; inverse volatility asserts Sharpe ratios aligned with correlation row sums; maximum diversification asserts equal Sharpe ratios with no view on correlations; minimum variance asserts equal premia; risk parity asserts that its risk budget is simultaneously a performance budget; HRP asserts that the three pieces of covariance information it discards at every split --- the Schur correction, the budget tilt, and within-cluster correlation --- are genuinely zero; RA-HRP asserts the nonlinear fixed point in which its floored cluster-Sharpe splits match the exact Schur-tangency splits; Schur-RA-HRP asserts the analogous condition after replacing standalone cluster risk by conditional Schur risk. The unit-free HRP--RA-HRP interpolation shows how to move continuously from the risk-only assertion to the return-adjusted assertion without mixing incompatible score units. Every one of these assertions is GRS-testable, and under one-factor structure every one of them is readable off betas and idiosyncratic variances.

Third, the geometry of suboptimality is angular and flat. Sharpe efficiency is exactly the cosine of the $\bSigma$-angle to tangency; the optimality sets are Lebesgue-null, so exact optimality is a probability-zero event under any diffuse belief; and yet the efficiency surface is stationary on each set, so order-$\varepsilon$ violations cost order-$\varepsilon^2$ efficiency with an explicit constant. ``Never optimal, rarely costly'' is two theorems, and the same flatness, read from the other side, is the first-order weight instability that makes plug-in optimization fragile.

Fourth, none of this is about variance per se. Under elliptical returns the identical conditions govern mean--CVaR at every confidence level and expected utility for every risk-averse investor, and in the diffusion limit they are the conditions for fractional-Kelly growth optimality; only tail asymmetry separates the objectives.

Fifth, the deepest sense in which the heuristics are right is not that their parameter conditions hold --- generically they do not --- but that each rule is \emph{exactly} the optimal policy under the symmetry group of its information set. Full exchangeability of beliefs forces equal weight; exchangeability of standardized returns with known volatilities forces inverse volatility; uninformative premia with conjugate uncertainty force minimum variance; and the equal-weight limit also emerges from adversarial robustness as ambiguity diverges. The heuristics are not shortcuts around optimization; they are what optimization itself returns when the inputs it would need are replaced by honest symmetry. The missing rungs --- a robust-control derivation of ERC under Sharpe ambiguity with trusted correlations, of HRP under ambiguity confined to inter-cluster structure, and of RA-HRP under ambiguity-penalized local Sharpe scores --- are stated here as open problems; Corollary~\ref{cor:threesubs} and Theorem~\ref{thm:rahrp-coincidence} suggest the exact ambiguity sets in which to look.

Sixth, RLPO is the dynamic continuation of HPO, not a replacement for it. Proposition~\ref{prop:hpo-stationary-rlpo} shows that every HPO rule is a deterministic stationary policy inside the RLPO Markov decision problem, and the $\gamma=0$ no-friction Bellman problem collapses to static HPO. Proposition~\ref{prop:dynamic-deviation-principle} gives the economic test for dynamic improvement: a learned policy should deviate from the HPO baseline only when its continuation-value gain exceeds the instantaneous reward loss, incremental turnover cost, and added implied-return defect. This makes the HPO defect a principled reward-shaping term rather than an arbitrary regularizer. Proposition~\ref{prop:hpo-value-gap} completes the accounting: the total value gap of acting myopically is exactly the discounted occupancy average of the HPO advantage, so statewise $\varepsilon$-adequacy of the heuristic implies a global $\varepsilon/(1-\gamma)$ guarantee (Corollary~\ref{cor:eps-myopic}) --- the dynamic analogue of ``never optimal, rarely costly.''

Seventh, RA-HRP supplies the natural actor architecture for RLPO. Its weights factorize through node-level branch probabilities, so a learned policy can be centered on the RA-HRP or Schur-RA-HRP split via residual logits. Proposition~\ref{prop:rahrp-actor-feasible} proves that this hierarchical actor remains long-only and fully invested by construction, recovers RA-HRP when residuals vanish, and decomposes the global KL distance from the baseline into node-level binary divergences. This gives RLPO a transparent attribution system: root-level deviations are regime allocation decisions, and deeper-node deviations are local cluster or security tilts. The homotopy calculus of Section~\ref{sec:homotopy-calculus} supplies the matching learning signal: at the prior, the myopic policy gradient with respect to each residual logit is the nodewise alpha the baseline leaves unexplained, damped by the logistic variance of the split (Corollary~\ref{cor:actor-gradient-alpha}).

Eighth, constraints, costs, topology, and tails are the implementation layer. A constrained heuristic is optimal only when its unexplained alpha vector lies in the normal cone of the mandate, and a costly heuristic is optimal only when premia compensate both risk and marginal turnover. For HRP and RA-HRP, the additional object is tree stability: merge preservation, cophenetic distance, and pathwise split perturbations are the bridge between fixed-tree theorems and traded rolling-window strategies. In RLPO these diagnostics become state variables and topology regularizers. Outside ellipticity, tail-shape residuals are the further object that separates Sharpe optimality from downside-risk optimality.

The result is a precise division of labor. Static optimization is the right tool when the information set includes the coordinates it consumes at the accuracy it requires. HPO is the right tool when those coordinates are noisy, unavailable, or economically untrustworthy. RLPO is the right tool when the state is dynamic and deviations from HPO can be paid for by continuation value. The theorems above say, line by line, which beliefs separate the three situations --- and what each situation costs.

\appendix
\section{Technical Appendix: Identities Used Repeatedly}\label{app:identities}

This appendix records the three algebraic identities used throughout the paper. They are included to make the paper self-contained and to separate reusable linear-algebra facts from portfolio-specific interpretation.

\subsection{Block inverse and Schur complement}

If
\[
M=\begin{pmatrix}A&B\\ C&D\end{pmatrix}\succ 0,
\]
with $A$ and $D$ nonsingular, then the Schur complements $S_A=A-BD^{-1}C$ and $S_D=D-CA^{-1}B$ are positive definite and
\[
M^{-1}=
\begin{pmatrix}
S_A^{-1} & -S_A^{-1}BD^{-1}\\
-D^{-1}CS_A^{-1} & D^{-1}+D^{-1}CS_A^{-1}BD^{-1}
\end{pmatrix}
\]
after eliminating the $D$ block, with the symmetric expression obtained after eliminating the $A$ block. The HRP and Schur-RA-HRP identities in Sections~\ref{sec:hrp} and~\ref{sec:hpo} are portfolio interpretations of this formula.

\subsection{Ray projection}

For any inner product induced by $M\succ 0$, the closest point on the positive ray $\{aq:a\ge 0\}$ to a vector $x$ is
\[
a^\star q,
\qquad
 a^\star=\max\left\{0,\frac{q^\top Mx}{q^\top Mq}\right\}.
\]
Choosing $x=\bmu_e$, $q=\bSigma\w$, and $M=\bSigma^{-1}$ gives the implied-return defect and its equality with squared Sharpe inefficiency.

\subsection{Tree products and binary KL}

If two hierarchical portfolios $p$ and $q$ are generated on the same binary tree by local splits $p_v$ and $q_v$, then every leaf probability is a product of branch probabilities and
\[
D_{\mathrm{KL}}(p\|q)
=
\sum_{v\in V_{\mathrm{int}}} P_v
\left[p_v\log\frac{p_v}{q_v}+(1-p_v)\log\frac{1-p_v}{1-q_v}\right],
\]
where $P_v$ is the capital mass reaching node $v$ under $p$. This is the identity behind the HRP--RA-HRP distortion theorem and the hierarchical RLPO regularizer.

\bibliographystyle{plainnat}
\bibliography{mathematics_of_heuristic_portfolio_optimization_HPO_FINAL_refs}

\end{document}